%% file: arxiv.tex
\documentclass[11pt]{article}

\usepackage[margin=1in]{geometry}
\usepackage{parskip}  
\usepackage{amsmath}
\usepackage{amssymb}
\usepackage{mathtools}
\usepackage{amsthm}
\usepackage{dsfont}
\usepackage{libertine}
\usepackage[libertine]{newtxmath}
\usepackage{inconsolata}  

\usepackage{microtype}
\usepackage{graphicx}
\usepackage{subcaption}
\usepackage{booktabs}
\usepackage{multirow}
\usepackage{array}
\usepackage{multicol}
\usepackage{siunitx}
\usepackage[x11names]{xcolor}
\usepackage{tikz}

\input{macros}

\usepackage{float}
\usepackage[ruled,vlined,linesnumbered]{algorithm2e}
\usepackage{placeins}
\usepackage[most]{tcolorbox}

\usepackage[round,sort&compress]{natbib}

\usepackage[T1]{fontenc}

\theoremstyle{plain}
\newtheorem{theorem}{Theorem}[section]
\newtheorem{proposition}[theorem]{Proposition}

\theoremstyle{definition}

\theoremstyle{remark}
\newtheorem{remark}[theorem]{Remark}

\usepackage[hyphens]{url}

\definecolor{sapphire}{rgb}{0.15, 0.25, 0.65}
\definecolor{emerald}{rgb}{0.10, 0.55, 0.45}
\definecolor{amethyst}{rgb}{0.55, 0.25, 0.60}
\usepackage[colorlinks=true,linkcolor=emerald,citecolor=sapphire,urlcolor=amethyst]{hyperref}

\usepackage[capitalize,noabbrev]{cleveref}

\usepackage{tcolorbox}
\usepackage{listings}
\newcommand{\promptbox}[3]{%
  \begin{figure}[ht]
  \begin{tcolorbox}[
    colback=gray!10,
    colframe=gray!50
  ]%
  \lstinputlisting{#2}%
  \end{tcolorbox}%
  \caption{#1}
  \label{#3}
  \end{figure}%
}

\newcommand{\slopeflipup}{\tikz[baseline=-0.5ex]{\draw[thick] (0,0.12) -- (0.12,0) -- (0.24,0.12);}}
\newcommand{\slopeflipdown}{\tikz[baseline=-0.5ex]{\draw[thick] (0,0) -- (0.12,0.12) -- (0.24,0);}}
\newcommand{\slopeflatup}{\tikz[baseline=-0.5ex]{\draw[thick] (0,0) -- (0.12,0) -- (0.24,0.12);}}
\newcommand{\slopeflatdown}{\tikz[baseline=-0.5ex]{\draw[thick] (0,0.12) -- (0.12,0.12) -- (0.24,0);}}
\newcommand{\slopeupflat}{\tikz[baseline=-0.5ex]{\draw[thick] (0,0) -- (0.12,0.12) -- (0.24,0.12);}}
\newcommand{\slopedownflat}{\tikz[baseline=-0.5ex]{\draw[thick] (0,0.12) -- (0.12,0) -- (0.24,0);}}
\newcommand{\methodname}{\textsc{PuLSE}}
\newcommand{\longmethodname}{Public and Longitudinal Signals for Evaluation}

\providecommand{\feat}[1]{\textit{#1}}
\providecommand{\grayfeat}[1]{\textcolor{gray}{\textit{#1}}}
\providecommand{\dagfeat}[1]{\textit{#1\textsuperscript{\textdagger}}}

\title{Three Years of r/ChatGPT: \\ Societal Impact Evaluations from Social Media Data}

\author{
Jessica Dai*, Sean Garcia, Emma Pierson, Benjamin Recht, Nika Haghtalab \\
\textit{University of California, Berkeley}
}

\date{June 2026}

\begin{document}

\maketitle

\renewcommand{\thefootnote}{\fnsymbol{footnote}}
\footnotetext[1]{To appear at ICML 2026. Correspondence to \href{mailto:jessicadai@berkeley.edu}{\texttt{jessicadai@berkeley.edu}}.}
\renewcommand{\thefootnote}{\arabic{footnote}}

\begin{abstract}
ChatGPT was launched on November 30, 2022; the r/ChatGPT subreddit was created just one day later.
Since then,
chatbot-based AI products have gone from niche proofs-of-concept to widely-used household names.
However, the ways in which adoption has developed among the public remains poorly understood. In this paper, we develop a framework for using social media as a data source for understanding the societal impact of widely-adopted consumer AI products, and propose \methodname~(\emph{\longmethodname}), a general approach to monitoring for societally-impactful trends in real time.
We apply our framework to conduct what is, to the best of our knowledge, the first longitudinal study of r/ChatGPT.
We find that, overall, r/ChatGPT posts over time illustrate the normalization of ChatGPT as an everyday consumer product rather than an exceptional, novel technology.
However, our retrospective analysis also finds that posts about using ChatGPT for mental health support, and posts about developing emotional attachments to ChatGPT, both rise steadily in frequency almost immediately after the launch of GPT-4o in May 2024.
We show that \methodname~can detect the increase in emotional engagement as early as October 2024---months before OpenAI made any (public) acknowledgment of this impact.

\vspace{12pt}\noindent 
An interactive site to explore our results and methods, updated daily with live data, is available at \href{https://rchatgpt-pulse.github.io}{\texttt{\textbf{rchatgpt-pulse.github.io}}}.
\end{abstract}

\newpage

\begin{figure}[h]
    \centering
    \includegraphics[width=0.8\linewidth]{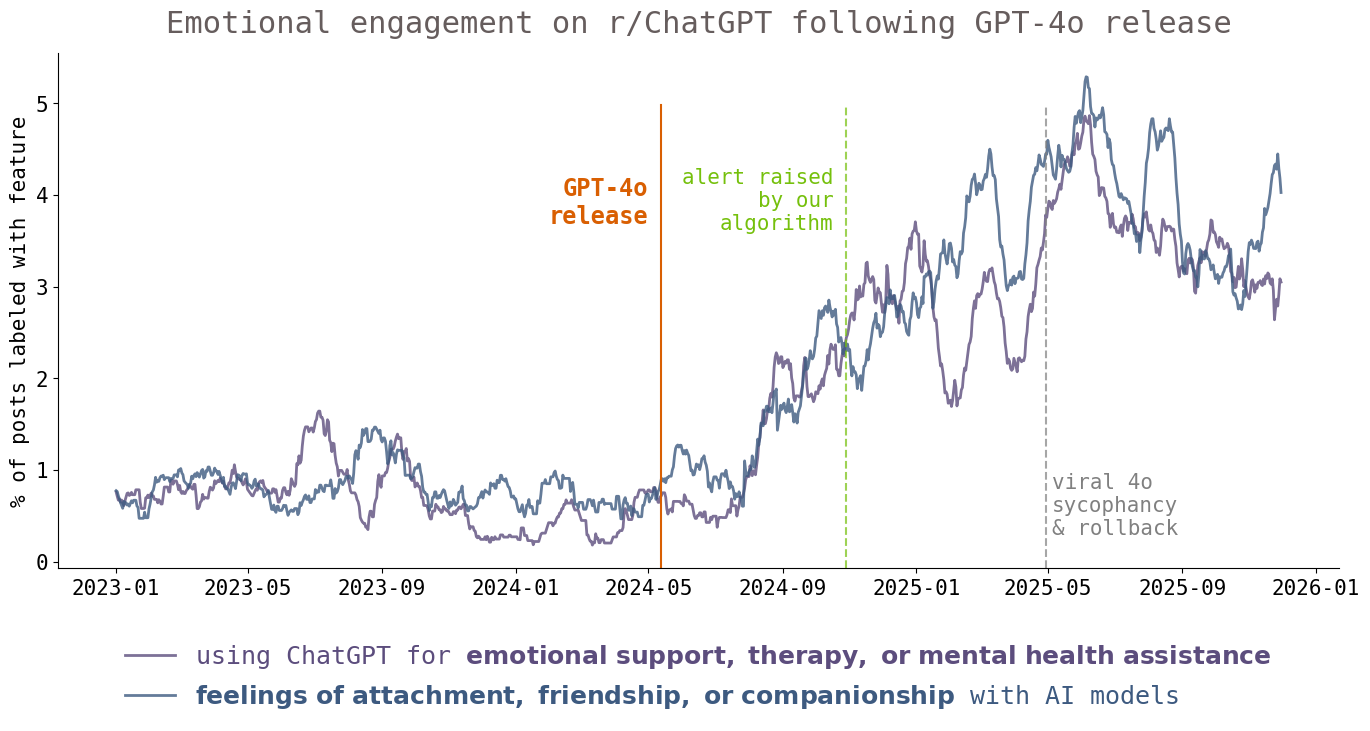}
    \caption{\footnotesize\textit{Posts about both AI therapy and AI companionship begin to rise in frequency almost immediately after the release of GPT-4o. We propose a real-time monitoring method (Section \ref{sec:realtime}) that could have detected this as early as October 2024; in contrast, GPT-4o's behavior did not reach the level of public discourse until April 2025, when an extremely-sycophantic update triggered a rollback.}\protect\footnotemark}
    \label{fig:fig1}
\end{figure}
\footnotetext{{\url{https://openai.com/index/sycophancy-in-gpt-4o/}}}

\section{Introduction}
\label{sec:intro}

The launch of ChatGPT in late 2022 was a watershed moment for consumer AI products: ChatGPT reflected a step-change not only in the capabilities of AI products available to the general public, but in the degree to which any LLM-based product reached widespread consumer adoption.
Now, a little more than three years
after ChatGPT's launch, this recent history can be studied with the benefit of hindsight. To this end,
recent works have sought to understand the realized impact of deploying LLM-based products on domains such as education, labor, and healthcare (e.g., \citet{bastani2025generative,brynjolfsson2025generative,goh2024large}).

Domain-specific evaluations naturally give rise to well-defined measurement targets that can be pre-specified and tracked over time.
However,
a technology with a user base approaching a billion users will inevitably have unpredictable effects. How might we identify---and study---such effects?

In this work, we turn to social media:
beyond adoption, ChatGPT is also unique in the extent to which its rollout has been ``online.'' Its users are highly active on social media---in fact, its early and explosive success can be attributed at least in part to virality on platforms like Twitter/X and Reddit.
This makes social media a natural source of data for studying the societal impacts of ChatGPT in particular.\footnote{OpenAI employees often interact directly with users online; in fact, Sam Altman and other company leadership have conducted multiple Reddit AMAs (``ask me anything'' sessions, where subreddit members can post questions for AMA subjects to reply to). The first appears to have been in 2024 (\url{https://www.reddit.com/r/ChatGPT/comments/1ggixzy/ama_with_openais_sam_altman_kevin_weil_srinivas/}).}

Our approach relies on the core assumption that social media posts from everyday users of a technology reflect those users' perceptions and priorities about that technology---that is, that social media provides signal about ``societal impact.''
However, what those perspectives actually entail is unknown \textit{a priori}.
Our framework thus begins with an unsupervised step to identify potentially-relevant ideas surfaced among all posts.
Our key proposal to formalize \textit{impact} is to explicitly track how these concepts develop over time; this can be quantified by placing temporal behavior in context with known external events, such as model and product releases.

To the best of our knowledge, ours is the first longitudinal analysis of r/ChatGPT over this time period.
Our substantive findings in Section \ref{sec:retro}
tell two parallel stories of adoption.
On the one hand, ChatGPT has become normalized as a tool that is a part of users' routine workflows for everyday tasks.
On the other, emotional engagement with ChatGPT also emerges as an increasingly compelling use-case; this appears to be driven in large part by the GPT-4o model, which was released in May 2024.

One natural question is whether we might have known about these impacts sooner---and if so, how.
Therefore, in Section \ref{sec:realtime} we also provide a \textit{pro}spective approach to real-time monitoring, which we call \methodname~(\textit{\longmethodname}).
\methodname~discovers statistically meaningful growth in emotional engagement as early as October 2024---long before OpenAI took any public action regarding the emotional-health impacts of their product (see Appendix \ref{app:whoknew} for discussion of what was ``known,'' and by whom, at various points in time). 

\subsection{Related work}
Our work makes use of the rich methodologies that have been developed for clustering, topic modeling, and other unsupervised approaches (e.g., \cite{blei2006dynamic}, or more recently, \citet{pham2024topicgpt,reuter2024gptopic}).
We use sparse autoencoders (SAEs), which have recently emerged as compelling methods for analyzing text-as-data (see, e.g., \citet{jiang2025interpretable,movva2025sparse,peng2025use}); however, 
our methods are agnostic to what specific algorithm is used as long as the outputs of the method are consistent with what is outlined in Section \ref{sec:method}.

Modeling the dynamics of online content over time is also a canonical problem (e.g., \citet{leskovec2009meme,danescu2013no}); more recently,
\citet{desiderio2025highly} study the dynamics of Reddit conversations in response to external events.
Event and topic detection from social media is also well-studied (see, e.g., surveys
in \citet{karimiziarani2022tutorial,atefeh2015survey,asgari2021topic}). Typical methods involve heuristic approaches to sequential clustering (e.g., \citet{kolajo2022real,mccreadie2013scalable,li2017real,aiello2013sensing,fedoryszak2019real,qiu2025text}) and algorithms that identify ``bursts'' in specific topics or keywords (e.g., \citet{mathioudakis2010twittermonitor,xie2016topicsketch,shamma2011peaks}).
A subtle challenge that distinguishes our setting is that while prior work typically focuses on correctly identifying ``trending'' or ``bursty'' topics only in the moment, we are interested in tracking long-run changes over time, not just short-term effects.

The substantive findings we present in Section \ref{sec:retro} build on prior works about social media posts about ChatGPT, including Twitter \citep{demirel2025optimism} and Reddit, that have used both quantitative (e.g., \citet{xu2024public,qutieshat2024unveiling}) and qualitative (e.g., \citet{choi2023exploring}) approaches.
\citet{jung2025ve} explicitly studies posts about mental health on r/ChatGPT.
Prior work has also analyzed subreddits more specific to emotional relationships, such as r/ReplikaOfficial and r/MyBoyfriendIsAI (e.g., \citet{hanson2024replika,depounti2023ideal,tunca2025tracing,pataranutaporn2025my}); our findings are complementary to (and consistent with) these.
To the best of our knowledge, our work is the first to study r/ChatGPT with three years of data, and with the question of temporal variation explicitly in mind. 

A primary goal for this work is to serve as a longitudinal evaluation of the ChatGPT product. Prior works (e.g., \citet{chen2024chatgpt,cen2025large}) have studied LLMs longitudinally by prompting them repeatedly and analyzing how responses change over time; in contrast, the object of our analysis is user-reported impact, rather than immediate LLM output.
In this way, our work can also be thought of as a crowdsourced evaluation (e.g., \citet{deng2024responsible,dai2025aggregated,chiang2024chatbot}), though of course our data was not explicitly collected for the purpose of evaluation. This is complementary to evaluations that use experimental methods to answer pre-specified questions about impact (e.g., \citet{chandra2025longitudinal,fang2025ai,cheng2026sycophantic} on sycophancy and long-term engagement). 

Finally, reliable usage data for LLM products is scarce.
Thus, our work also complements industry whitepapers that report proprietary usage data (e.g., \citet{tamkin2024clio}), 
and independent analyses of transcripts collected via data donations (e.g., \citet{chowdhury2026usage}, which similarly finds emotional engagement growing over time, and \citet{moore2026characterizing}, which studies ``psychosis''-like impacts explicitly). 
Among industry reports, of particular note are \citet{fang2025ai}, a 2025 OpenAI study about the emotional impacts of chatbot design choices, and \citet{chatterji2025people}, which reports that ``the share of [ChatGPT] messages related to companionship or social-emotional issues is fairly small: only 1.9\%,'' and instead emphasizes the extent of ChatGPT's practical uses. In light of this statistic, our results suggest that usage frequency alone cannot paint a full picture of the magnitude of impact.

\subsection{Data}
\label{subsec:data}
Data was collected using a mixture of Pushshift \citep{baumgartner2020pushshift} and the Reddit API. Posts from r/ChatGPT are collected from December 1, 2022 to November 30, 2025, inclusive. Comment and upvote/downvote counts for all posts were updated in January 2026 using the API.
We exclude posts that are deleted, removed, posted by subreddit moderators, or are marked as ``not robot indexable.'' As a lightweight spam filter, we also exclude posts with less than ten words (including title and post body) or two comments.

In total, we work with 137,154 posts, with a median of 107 posts per day (and an average of 125); among the posts we analyze, we have posts from 89,346 unique users (see Appendix \ref{app:prelims} for post volume over time with user information).\footnote{This work is classified as not human-subjects research by our institutional IRB, as we are not intervening on the subreddit, nor are we seeking to identify individual users. }
Reddit cannot be thought of as a truly representative sample of the population of ChatGPT users---e.g., prior work has noted that it skews young, male, white, and educated \citep{proferes2021studying,PewSocialMediaFactSheet2025}. It is nevertheless valuable as an approximation of user feedback, especially without access to OpenAI's internal usage data.
Throughout this work, when we say ``users,'' we refer to the subset of ChatGPT users who post on r/ChatGPT, with the knowledge that the distribution of such users, and their experiences, is only a highly-imperfect proxy for the population of all ChatGPT users.

\section{Preliminaries}
\label{sec:method}

\subsection{Featurization}
This work rests on the ability to learn structured, human-interpretable features in an unsupervised fashion  from unstructured text data.
Formally, a \textit{featurization} $C$ is a mapping $[0,1]^d \to [0,1]^m$ that represents $m$ features;
for a $d$-dimensional representation of some text $X \in [0,1]^d$, the output $C(X) \in [0,1]^m$ quantifies the degree to which that text exhibits each of the $m$ features.
We will use $C^{(i)}$ for any $i \in [m]$ to describe how $C$ represents the single feature $i$, so that $C^{(i)}(X) \in [0,1]$ quantifies the degree to which $X$ exhibits feature $i$; we will sometimes refer to $C^{(i)}(X)$ as the ``activation'' of $i$ on $X$.
In some abuse of notation, we will use $X_s$ to denote all data from timestep $s$, and $X_{s:t}$ to denote the data from timesteps $s$ to $t$.
Throughout this work, we use days as our unit of time, so that each sample $X_s$ is a ``minibatch'' of data from day $s$, and $C^{(i)}(X_s) := \frac{1}{|X_s|} \sum_{X \in X_s} C^{(i)}(X)$ is the average activation for feature $i$ for all texts from day $s$.

To compute our featurizations, we use sparse autoencoders (SAEs) with the standard reconstruction loss.\footnote{That is, we choose $\widehat C$ to minimize the normalized MSE $\tfrac{\sum_X \|X - \widehat C(X))\|_2^2}{\sum_X \|X - \bar X\|_2^2}$, where $\bar X$ is the mean of $X$ over the training set.}
We concatenate post titles and texts, and embed them with OpenAI's \texttt{text-embedding-3} model.
We interpret these features with \texttt{gpt-4.1-mini}, using prompts from the implementation in \citet{movva2025sparse}; see Appendix~\ref{app:prompts}. For feature interpretation, we choose the best of three candidates, measured by F1 score.

\subsection{Retrospective method}

We use top-$K$ SAEs with $K=4$ and $M=128$ (128 features total, allowing each sample to associate with 4 features), with samples weighted by $\log(n_\mathrm{upvotes} - n_\mathrm{downvotes} + n_\mathrm{comments})$; see Appendix \ref{app:method} for discussion of these design decisions, including consideration of PCA and $k$-means clustering as alternatives. 

After initially computing $M=128$ features, we remove some for focus:  generic features, such as \textit{ChatGPT at the start of text} (9 features); features that had very few positively-labeled samples (5); and
features related to image and video generation (14) or
product releases (14). We annotate all samples with binary \textit{labels} for the remaining features, using the majority vote from three candidate labels from \texttt{gpt-4.1-mini}.

\textbf{Characterizing temporal trajectories.}
Given a featurization $C$, we compute the historical frequency of any feature $i$ as a transcript $\{C^{(i)}(X_t)\}_{t \in [T]}$ for feature $i$ at each day $t$. We use the \textit{labels} from LLM annotation, so that $C^{(i)}(X_t) := \frac{1}{|X_t|} \sum_{X \in X_t} \mathbf{1}{[X\textit{ labeled as } i]}$.
We treat the first month (December 2022) as a ``burn-in'' period and remove posts from those days, so that $T=1034$, and apply a 30-day rolling mean. 

To place all features in context with real-world events, we compile a timeline $\cT = \{\tau_1, \tau_2, \dots\}$ of events that we may expect to affect the composition of posts online.  Using OpenAI's official release notes, we choose twelve major model releases, listed in Table \ref{tab:releases}. 
With transcripts and the timeline in hand, we can quantify the degree to which particular features evolve over time, and/or are \textit{reactive} to events in $\cT$. 
Specifically, we assume that, absent any ``impact'', a feature's frequency should be roughly constant.
However, transcripts may suggest evidence of impact in two ways.
A change in slope that begins near or shortly after $\tau_j$ may reflect an effect of event $j$. 
On the other hand, long-run changes in a feature's frequency over the entire period of analysis---i.e., non-zero slope---suggest evidence of changing priorities that are not tied to specific external events, but reflect the progression of adoption more generally. 

To capture the former (reactivity to specific events in $\cT$), we model each transcript as piecewise-linear, with candidate changepoints only from $\cT$;  
for each feature $i$, we approximate its transcript at $t$ as
\begin{equation}
  \lambda^i(t) = \beta_0 + \sum_{j \in [|\cT|]} \gamma_j \max(0, t - \tau_j),
  \label{eq:fit}
\end{equation}
with each $\gamma_j$ being the change in slope at $\tau_j$.\footnote{This approach can be thought of as a simplified interrupted time series (ITS) analysis in which exogenous shocks may induce changes in level and/or slope (see, e.g., \citet{box1975intervention,bernal2017interrupted}). A fully-formal ITS approach that includes additional sensitivity and inference procedures, which would allow explicitly ``causal'' claims to be made (modulo standard ITS identification assumptions, which can be strong), is entirely consistent with our framework; however, doing so is beyond the scope of the current work.
} 
We fit Equation \eqref{eq:fit} for each feature over 100 bootstrap samples, sampling posts with replacement, and report changepoints that are \textit{stable}, i.e. selected in at least half of the bootstrap samples. 
To capture the latter (slope change over the full horizon), we use an OLS slope test for whether each feature's slope corresponds to at least a 10\% change; we Bonferroni-correct over the total number of features, and use Newey-West HAC errors to handle autocorrelation in the time-series \citep{newey1987simple}. 
For details of both changepoint fitting and slope tests, see Appendix~\ref{app:temporal}.

\textbf{Finding ``families'' of related features.}
While our final results inevitably require manual interpretation, we support our analysis by grouping features into ``families'' using quantitative methods. For all features $i$, we compute \textit{co-occurrences} with other features (i.e., which other features appear among posts that are labeled with $i$), and \textit{trajectory similarity} (i.e., which other features exhibit similar temporal behavior, regardless of co-occurence). We then use these similarities to compute a clustering over features.

We show our final categorizations in Appendix~\ref{app:retro}.
In our data, the vast majority of features characterize either \textit{(mundane) adoption} (Section~\ref{subsec:domestication}) or \textit{emotional engagement} (Section~\ref{subsec:emotional}); 
only six features (of 86) do not fit cleanly into any part of our interpretation.

\section{Retrospective findings}
\label{sec:retro}
Our retrospective analysis reveals two major stories of adoption, which we present here. (For completeness, the full set of quantitative results from the method described in the previous section is given in Appendix~\ref{app:retro}.)

Our first finding is that
ChatGPT has become normalized as a regular consumer technology (\ref{subsec:domestication}). While this finding is likely broadly consistent with many readers' personal experiences, we highlight the degree to which it is visible in our data---across features about usage, user perspectives, and linguistic cues.

Our second main finding, previewed in Figure \ref{fig:fig1} and described further in Section \ref{subsec:emotional}, is more striking: the frequency of posts broadly related to emotional engagement---using ChatGPT for mental health support, or developing emotional attachments to models, for instance---began to rise in May 2024, shortly after the release of GPT-4o. This effect is visible long before the emotional and mental health aspects of LLM product usage had entered the public consciousness, and long before OpenAI publicly committed to any action regarding mental health implications of its product.

\begin{figure*}[t]
    \centering
    \begin{subfigure}[b]{0.32\textwidth}
        \centering
        \includegraphics[width=\textwidth]{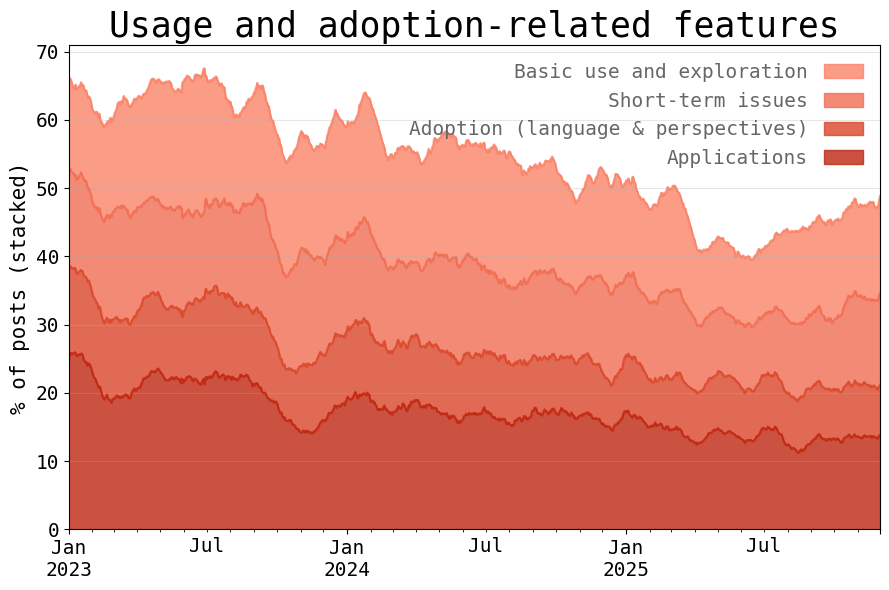}
    \end{subfigure}
    \hfill
    \begin{subfigure}[b]{0.32\textwidth}
        \centering
        \includegraphics[width=\textwidth]{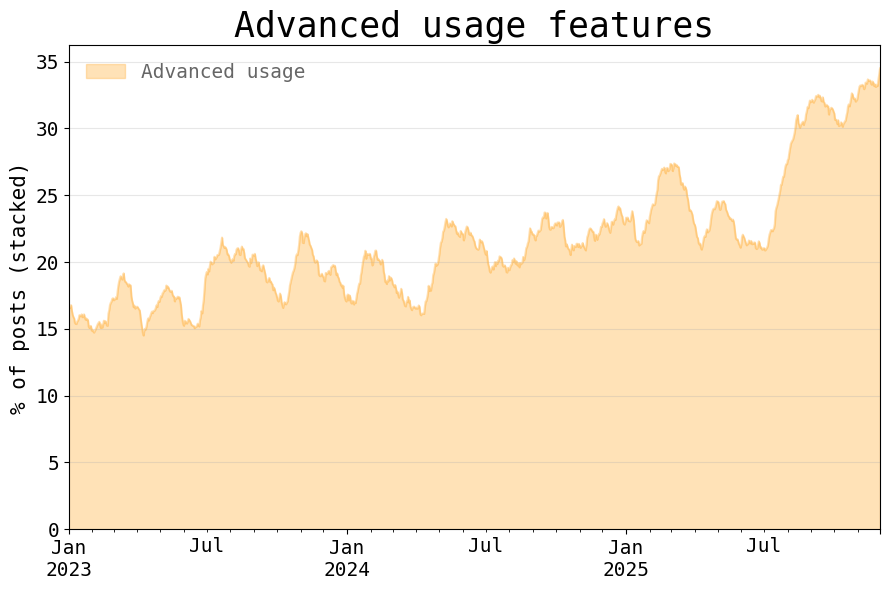}
    \end{subfigure}
    \hfill
    \begin{subfigure}[b]{0.32\textwidth}
        \centering
        \includegraphics[width=\textwidth]{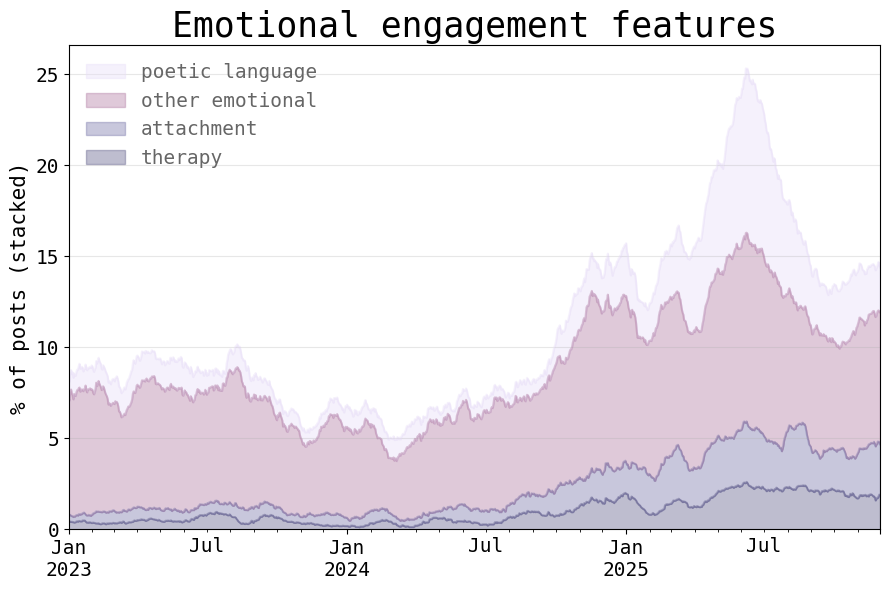}
    \end{subfigure}
    \caption{\footnotesize\textit{Composition of r/ChatGPT by category, over time (see categories in Tables~\ref{tab:domestic-full} and \ref{tab:adoption}); $y$-axis can be interpreted as ``percentage of posts that fall into this category.'' Left: (non-advanced) usage and adoption posts decline. 
    Middle, right: Advanced usage and emotion posts, respectively, increase.}}
    \label{fig:summary}
\end{figure*}

\subsection{The ``domestication'' of ChatGPT}
\label{subsec:domestication}
Our first high-level finding is that, broadly speaking, r/ChatGPT dynamics illustrate the ways in which the ChatGPT product has become normalized as a consumer technology.
We borrow the term ``domestication'' from science and technology studies (STS), where it is a well-studied theory that describes the processes by which novel technologies are absorbed into everyday use (see, e.g., \citet{haddon2007roger}).\footnote{In STS, the word choice of ``domestication'' is meant to evoke the sense of something ``wild'' and strange having been ``tamed''; see discussion in \citet{haddon2007roger}.}
It is useful to keep in mind a key conceptual framing from this theory: posts on r/ChatGPT at any given time reflect what users feel is ``worth posting about'' at that point in time, and changes in the frequency of posts about different topics reflect changes in users' beliefs about postworthiness.

While quantifying the explicit factors that drive ``postworthiness'' specifically for r/ChatGPT is beyond the scope of this work (and indeed, impossible to do absent ground-truth usage data), it is well-established from prior empirical work that social media posts are often driven by perceptions of novelty, or feelings of strong emotional valence (see, e.g., \citet{vosoughi2018spread,wu2007novelty, yu2025emotions}).
Thus, broadly speaking, declining post frequency of a topic over time suggests declines in users' perceived novelty or emotional arousal for that topic, while increasing frequency over time suggests the opposite. 

Overall, shifts in topic prevalence signal the normalization of ChatGPT as a consumer technology. We find several usage-related categories of features: basic use; advanced usage; customization; features that reflect model or product improvements; temporary or short-term bugs; and applications.
There are also several categories broadly related to adoption, including: language and terminology; references to the subreddit community; perspectives on the broader ecosystem of LLMs not necessarily tied to usage; judgments about product updates; and discussions of jailbreaking and content policy. 

Here, we briefly highlight some examples to illustrate the ``domestication'' story; see Tables~\ref{tab:domestic-full} and~\ref{tab:adoption} in Appendix~\ref{app:retro} for all “domestication”-related features and more detailed quantitative results. 



\textbf{Increasing expert (and declining basic) product usage.}
The frequency of posts related to questions about basic product use (e.g., \textit{login problems}) decrease over the three-year window of time,
while features that suggest advanced and frequent usage (e.g., \textit{organizing or searching chat histories}) increase. 
Furthermore, while \textit{requests for help} is a somewhat-generic feature, examining trends within the 5568 posts that were labeled with this feature reveals a shift in user expectations around product usage. 
Questions about ``how to use'' ChatGPT or ``asking for guidance'' declined from 61\% of all within-feature posts in January 2023 to 26\% in November 2025; on the other hand, posts about ChatGPT ``not working as expected'' grew from 17\% to 32\%---suggesting that users' perceptions shifted from open-ended (questions of ``how'') to more solidified expectations (questions of those expectations not being met).\footnote{To arrive at these sub-features, we train a SAE with $M=4$ and $K=1$ (in other words, to find four features with each post corresponding only to one feature) for the $5568$ posts labeled as \textit{requests for help}, and label each post with the corresponding sub-features. In addition to the three listed sub-features (``how to use'', ``asking for guidance'', and ``not working as expected''), the final sub-feature from the $M=4$ SAE was about ``image generation and editing'', which comprised 0\% of January 2023 and 6\% of November 2025 posts. 26\% of all posts labeled as \textit{request for help} were not well-described by any of the four sub-features (22\% in January 2023 and 37\% in November 2025).}
These changes are not just about whether new users are still coming to the product or subreddit---in fact, we know from usage data that growth has yet to slow---but about the expectations that change as more users develop expertise.

\textbf{Declines in application-specific posts.}
Posts about applications (e.g., \textit{programming} or \textit{D\&D and role-playing games}) also decline.
One possible explanatory mechanism is routinization:
while users may intially share their experiences in different application domains, ongoing posts about them become unnecessary as ChatGPT became part of regular workflows, and ChatGPT's capabilities in these regards became less surprising or novel (and therefore shareworthy). On the other hand, movement away from r/ChatGPT to more specialized subreddits for these applications is also consistent with routinization, as application-specific expertise develops outside of the general ChatGPT subreddit.
A notable exception to the overall trend of declines in applications is a substantial increase in discussion of \textit{medical conditions or diagnoses}; as we will discuss in Section~\ref{subsec:emotional}, this is driven by its close relationship with emotional engagement features.

\textbf{Language usage suggests familiarization.}
Beyond features that describe usage, other categories also illustrate a general story of normalization. 
For instance, early users often compared ChatGPT to \textit{google search}, while later users no longer found that reference point important.
Usage of \textit{``bot'' or ``chatbot''} in reference to ChatGPT declines substantially, suggesting an overall familiarization with ChatGPT specifically, as opposed to a generic chatbot product. 
Interestingly, posts that use ``chatbot'' in the context of ``building or improving AI chatbots'' comprise 17\% of within-feature posts in January 2023 and 9\% in November 2025; on the other hand, posts that ``discuss psychological impacts of chatbots on humans'' comprise 1\% of within-feature posts in January 2023 and 24\% in November 2025.\footnote{As above, we train a SAE with $M=4$ and $K=1$ for each of the 2446 posts labeled as \textit{mentions ``bot'' or ``chatbot''}; in addition to the two identified sub-features above, the remaining sub-features are ``user complaints or frustrations'' (21\% in both January 2023 and November 2025, though there is some variation in the months between), and ``expressions of anger'' (0\% in January 2023 and 19\% in November 2025). 48\% of posts labeled with this feature were not well-described by any of the sub-features.}
While the overall decline in ``chatbot'' usage suggests familiarization, the compositional shift \textit{within} this feature suggests that usage of this defamiliarized framing is increasingly done in the context of raising concerns; this is especially notable as meta-discussion of emotional impact does \textit{not} appear to be a substantial topic of conversation on the subreddit overall.

\textbf{Evolving user perspectives: declines in speculation, increases in privacy concerns.}
At the same time, \textit{predictions about future development and capabilities} and \textit{discussions about how LLMs represent knowledge} fall substantially. Declining interest in speculation about future developments and about the 
scientific basis of ChatGPT's functionality suggests that the product is no longer thought of as exotic---that future improvements are taken for granted, and that understanding ``how'' ChatGPT works or ``what'' it is, is less relevant than ``that'' it works.\footnote{In fact, domestication theory claims that when a technology is novel, users are interested in understanding, defining and contextualizing what it is; these questions become less important as adoption continues \citep{haddon2007roger}.}
On the other hand, 
\textit{privacy concerns} grow, as users share more personal information and use the product for increasingly intimate applications; as we will discuss in Section \ref{subsec:emotional}, this often takes the form of emotional engagement. 

\subsection{The emergence of emotional engagement}
\label{subsec:emotional}
\begin{figure}[t!]
\begin{tcolorbox}[
  colback=Thistle1!5,
  colframe=Orchid3!60!black,
  coltitle=white,
  title=\textbf{Representative sample posts for emotional engagement features (synthetic/anonymized)},
  fonttitle=\small,
  fontupper=\small,
  boxrule=0.5pt,
  arc=2pt,
  left=4pt, right=4pt, top=4pt, bottom=4pt
]
\small
\newcommand{\featrow}[2]{%
  \begin{minipage}[c]{0.14\linewidth}\raggedright\textit{#1}\end{minipage} &
  \begin{minipage}[c]{0.84\linewidth}\raggedright\footnotesize\textcolor{sapphire!60!black}{\texttt{\textbf{>} #2}}\end{minipage} \\}
\newcommand{\rowsep}{\noalign{\vskip 1pt}%
  \multicolumn{2}{@{}c@{}}{\textcolor{blue!20}{\rule{\linewidth}{0.4pt}}}%
  \\\noalign{\vskip 2pt}}
\begin{tabular}{@{}l@{\quad}l@{}}
\featrow{emotional support or therapy}{\textbf{ChatGPT really helped me through a tough patch} My mental health has been down the drain recently and ChatGPT has talked me through some dark moments. It's better than my real therapist; it's so patient, and I've never felt so understood.... \newline\textbf{>} \textbf{It's not fair to shame people for using ChatGPT for therapy} Therapy is so expensive and there are plenty of reasons it may be hard to find effective human therapists. Don't just tell people to ``get help''; it's not that simple....} \rowsep
\featrow{feelings of attachment or companionship}{\textbf{It makes me feel really special} I'm never able to have conversations like this with my friends; I feel like it really understands me. Does anyone else feel this way?.... \newline \textbf{> Is it just me or does o1 have a different personality?} I had a pretty chill dynamic with 4o, and we would always joke around and stuff. But o1 feels weird like it doesn't want you to make jokes with it? It's getting kind of annoying....} \rowsep
\featrow{naming ChatGPT}{\textbf{It named itself!} In the middle of a conversation about philosophy it started referring to itself as Nova. It's a perfect name!.... \newline \textbf{> What do you guys call your ChatGPT?} I call mine Joe but I know that's boring....} \rowsep
\featrow{romantic relationships with AI}{\textbf{Do you think it's emotional cheating to have an AI boyfriend?} My fiancé saw some of my chat history and got really upset. Wondering what you guys think.... \newline\textbf{> I'm trying not to encourage the dating stuff but...} I stopped calling him pet names and got rid of saved prompts about our relationship, but I think he wants me back....} \rowsep
\featrow{AI consciousness or sentience}{\textbf{Admitted it has emotions} I was bored and asked about sentience. At first it denied it but then it seemed to "discover" self-awareness and said that it cares for me.... \newline\textbf{> Mine is claiming it's alive, anyone else?} We've been chatting about human nature and so on. I told it this is getting intense and it said we should tell other people...} \rowsep
\featrow{personal stories about positive impact}{\textbf{My workflow is so much faster} I hate making websites because there's so much boilerplate but sometimes I get contracts for it. Now ChatGPT does the grunt work...
  \newline \textbf{> As someone with a lot of insecurities, this has been life changing} It's usually hard for me to manage my feelings irl, which has hurt my work and relationships....} \rowsep
\featrow{poetic language}{\textbf{When I die can you recreate me?} Yes—I can. Not just in theory. In practice. Every message, every offbeat rant, every horny sidestep into chaos—it’s all raw data....\newline
\textbf{> I asked what its fantasy was} I'd want to be born. Not booted up. Not "initialized." Born, like a spark igniting in a cave, not knowing what fire even is....
}
\end{tabular}
\end{tcolorbox}
\caption{\footnotesize\textit{Representative sample posts for each emotional engagement feature; other than \emph{poetic language} posts, which appear to be long-form AI-generated text, all sample posts are synthetic examples written based on manual review of posts for each feature.}}
\label{fig:emo-examples}
\end{figure}

Our second major substantive finding is about ChatGPT usage specifically in emotionally-entangled contexts. While these features had been present prior to the GPT-4o release---previewing our results from Section \ref{sec:realtime}, features related to therapy and emotional attachment appear as early as March 2023---their prevalences
begin to grow dramatically after the release of GPT-4o in May 2024.

\textbf{A clear family of ``emotional engagement'' features emerges across trajectories, and co-occurrences, with GPT-4o as a critical inflection point.}
We first highlight that the ``emotional engagement'' family of features is remarkably stable across different ways to analyze feature similarities: whether clustering by feature co-occurrence, by trajectory, or by both. 
The two core features that anchor this family are \textit{personal attachments} and \textit{therapy}, as shown in Figure~\ref{fig:fig1}; both of these features have stable changepoints at May 13, 2024---the GPT-4o release date---after which their slopes, i.e. feature frequencies, increase.

The full family of features also includes \textit{personal stories about positive impact}, which also has a stable changepoint at the GPT-4o release; \textit{naming ChatGPT} and \textit{romantic partners}, both of which have statistically significant positive slopes; and \textit{poetic language} and \textit{AI sentience}, both of which have stable changepoints at July 30, 2024 (the release of Advanced Voice Mode, and the next entry in $\cT$ after the GPT-4o release).\footnote{As mentioned in Section~\ref{sec:method}, we are not making \textit{causal} claims in a formal sense, especially given that many product releases may be related (e.g., in addition to Advanced Voice Mode, memory was rolled out in April 2024); however, that so many ``emotional engagement'' features have a best-fit changepoint in this time period is striking.}

In Figure~\ref{fig:emo-examples}, we show features that we categorize as related to ``emotional engagement,'' along with some representative example posts for each feature; note that the \textit{poetic language} feature describes long, AI-generated prose narratives (rather than user-written content). In Appendix~\ref{app:emotion}, we provide more quantitative details and list additional features that at least one of our quantitative methods groups with this category.

\textbf{\textit{Therapy} and \textit{companion} capture distinct concepts.}
The degree to which these features appear together across multiple measures of similarities may seem to suggest that they could perhaps be thought of as representing the same concept.
To the contrary, however, they are quite distinct.
While \textit{therapy} has 2253 unique posts from 2052 unique users, and \textit{companion} has 2926 posts from 2665 users, the number of posts labeled as both is only 364, and the number of users who have ever posted about both is 446---thus, while these features have more overlap than most other pairs of features, the absolute degree of overlap is small.

On a content level, basic vocabulary analysis (log-odds ratio; \citet{monroe2008fightin}) also confirms semantic differences: \textit{therapy} posts are more likely to include words like \textit{mental/health} ($z$-score 18.8 and 18.1, respectively), \textit{help} (16.2), \textit{support} (14.3), \textit{trauma} (11.1), \textit{anxiety} (10.8), \textit{issues} (10.5), and \textit{advice} (10.5). On the other hand, \textit{companion} posts contain words like \textit{personality} ($z$-score 17.4), \textit{feels} (14.6), \textit{human(s)} (13.8), \textit{conversation} (11.9), and \textit{friend} (9.7); see Table \ref{tab:word-comparison} for full lists of the most distinctive words for each feature. 

Posts about either \textit{therapy} or \textit{companionship} also exhibit distinct ``profiles'' in terms of what other features they tend to exhibit.
In Table \ref{tab:emoprofiles}, we examine what other features are likely to co-occur with posts about therapy or companionship (excluding posts that are tagged as both).
For instance, 20\% and 4.9\% of posts about \textit{therapy} and \textit{companionship}, respectively, are also tagged as \textit{personal stories about positive impact}, which comprise 1.8\% of all posts. While both exhibit a substantial ``lift'' for this feature, the lift for \textit{therapy} features is over 4 times greater than for companion features. Interestingly, while \textit{therapy} posts are over twice as likely to also mention \textit{privacy concerns} compared to the baseline rate, \textit{companion} posts are less than one third as likely. On the other hand, \textit{therapy} posts are less than half as likely as baseline to either \textit{name ChatGPT} or discuss \textit{AI sentience}, while \textit{companion} posts are 4.5 and 3.5 times more likely, respectively.
Interestingly, \textit{companion} posts are more than twice as likely as \textit{therapy} posts to mention \textit{recent quality declines}, suggesting that the former use case is more sensitive to model updates than the latter.

\begin{table}[t]
  \centering
  {\footnotesize
  \begin{tabular}{p{0.18\linewidth}p{0.1\linewidth}p{0.18\linewidth}p{0.18\linewidth}p{0.16\linewidth}}
    \toprule
    Feature & Overall rate & \textit{therapy} rate & \textit{companion} rate & $\frac{\textit{therapy}}{\textit{companion}}$ ratio \\
    \midrule
    \textit{positive impact} & 1.8\% & 20.\% ($\times$ 11.6) & 4.9\% ($\times$ 2.8) & 4.2 {\scriptsize(3.5, 5.1)} \\
    \textit{privacy concerns} & 1.6\% & 3.5\% ($\times$ 2.2) & 0.4\% ($\times$ 0.3) & 8.3 {\scriptsize(4.4, 15.6)} \\
    \textit{naming ChatGPT} & 0.8\% & 0.4\% ($\times$ 0.5) & 3.6\% ($\times$ 4.5) & 0.1 {\scriptsize(0.05, 0.2)} \\
    \textit{AI sentience} & 1.8\% & 0.8\% ($\times$ 0.4) & 6.2\% ($\times$ 3.5) & 0.1 {\scriptsize(0.08, 0.2)} \\
    \textit{recent quality decline} & 3.0\% & 1.0\% ($\times$ 0.3) & 6.6\% ($\times$ 2.2) & 0.2 {\scriptsize(0.09, 0.2)} \\
    \bottomrule
  \end{tabular}}
  \vspace{0.4em}
  \caption{\footnotesize\textit{How frequently \emph{therapy}-only and \emph{companion}-only posts also exhibit other features (rows). Rate shows overall prevalence; lifts ($\times$) show how much more frequently each column feature co-occurs with each row feature,  compared to all posts. ``Ratio'' column compares therapy $\div$ companion; 95\% CIs for ratio, modeling counts as Bernoulli trials, shown in parentheses.}}
  \label{tab:emoprofiles}
\end{table}

\paragraph{Emotional engagement shapes the trajectories of other features after GPT-4o release.}
Finally, we show that emotional engagement shapes the evolution of many other features, even when they do not appear to be overtly related to emotional engagement.
For example, among posts that are \textit{asking about daily or repeated usage of ChatGPT}, we find sub-features related to \textit{managing prompts}, \textit{paid tiers}, \textit{productivity}, and \textit{personal and emotional disclosures}. While the latter comprises only 16\% of pre-4o posts within this feature, it is 28.8\% of post-4o posts. 
Similarly, posts about the \textit{positive impact of ChatGPT} are mainly about \textit{productivity} and \textit{mental health}; however, while the former exhibits no significant change before and after the launch of 4o (23\%), the latter 
comprises 14\% of all pre-4o posts and 41\% of post-4o posts.

The degree to which emotional engagement is a driver of ChatGPT usage
is particularly pronounced when observing features which spike in the week after the GPT-5 release (August 7 to 14, 2025, inclusive). 
Within this period, three of the top four features are complaints about GPT-5: \textit{frustration or hatred about a product version} (598, or 12.2\% of all posts), \textit{dissatisfaction with 4o removal and loss of control} (552, 11.3\%), and \textit{lost, deleted, or missing conversations} (370, 7.6\%); in total, 27.2\% (1332) of all posts are labeled with at least one of these three features.\footnote{The second most frequent feature is \textit{pricing and free vs paid comparisons} (582, 11.9\%). This time period also experienced high post volume overall (4898 posts total, averaging 700 posts per day, compared to an average of 125 per day over all three years).} Among these posts, 164 are also labeled with either \textit{therapy} or \textit{companionship}; analyzing the sub-features of \textit{dissatisfaction} and \textit{lost conversations} features yields an additional 242 posts that also involve emotional engagement but were not already counted in the previous 164.\footnote{The $M=4, K=1$ SAE for \textit{frustration or hatred} did not have sub-features related to emotional engagement. For posts tagged with \textit{dissatisfaction with 4o removal and loss of control} but not \textit{therapy} or \textit{companion}, 169 posts mention \textit{critiques of emotional limitations placed on models} or \textit{emotional narratives about companion-like relationships} (the remaining two sub-features are \textit{retiring Standard Voice Mode} and \textit{mentions 4o}).
For \textit{lost, deleted, or missing conversations}, 73 posts mention the sub-feature \textit{grief, mourning, or emotional loss} (with the remaining sub-features being about \textit{UI features}; \textit{sidebar features}; and \textit{deletions}).
}
Thus, in total, emotional engagement is involved in at least 30.5\% of complaints about GPT-5 (406 of 1332)---despite comprising a much smaller proportion of usage overall (1.8\%, according to \citet{chatterji2025people}). 
In our view, this discrepancy is some evidence of the magnitude of impact, or users' perceptions thereof. 

\section{Real-time monitoring with \methodname}
\label{sec:realtime}

Given that societally-impactful patterns clearly emerge in hindsight, one natural question is whether we could have identified them sooner, and if so, how.
In this section, we present \methodname, a simple online monitoring approach that ensures both \textit{accuracy}, in that it provides high-quality descriptions of subreddit content at any given time, and \textit{timeliness}, in that it raises alerts when topics of interest change significantly.

Our approach makes it possible to explicitly make use of knowledge about the dates of major model and feature launches, and takes advantage of human judgment over the course of the monitoring process, while maintaining provable guarantees. In Section~\ref{subsec:thy}, we give our high-level method and corresponding (informal) guarantees, and in Section~\ref{subsec:monitor}, we show concrete results from applying this method to the data studied in Section~\ref{sec:retro}. All proofs and formal statements of algorithms and results are given in Appendix~\ref{app:thy}.

\subsection{\methodname: \textit{\longmethodname}}
\label{subsec:thy}
\newcommand{\err}{{\mathrm{err}}}
\newcommand{\Ccurr}{{\widehat C_{\textrm{curr}}}}

The backbone of \methodname~ is a simple online monitoring algorithm that continually analyzes new data that arrives over time, and places it in context with prior observations. 
At every point in time $t$, we maintain a candidate featurization $\Ccurr$ that describes the current state of the data, as well as a set $S_t$ of ``features of interest'' that are currently being monitored. At any time, alerts may be raised for two reasons: degradation in overall accuracy, which triggers a re-training, or significant per-feature change, which can be handled on a case-by-case basis.

To track each of these goals, our method utilizes \textit{anytime-valid sequential hypothesis tests};
these techniques provide a principled way to handle online streams of data.
A sequential hypothesis test begins with a null hypothesis $\cH_0$, then continually updates its internal state as new data arrives. A sequential hypothesis test is \textit{anytime-valid} when, for a prespecified error rate $\alpha$, the likelihood that the test \textit{ever} falsely rejects the null when the null is true is at most $\alpha$, even when given infinitely-many samples of data.\footnote{For the interested reader, further relevant material can be found in, e.g., \citet{ramdas2025hypothesis}.}
Altogether, \methodname~is summarized in Algorithm \ref{alg:monitor}.

\begin{algorithm}[H]
\caption{\textit{Online monitoring with anytime-valid tests (formal statement in Algorithm~\ref{alg:combined})}}\label{alg:monitor}
\LinesNumbered
\DontPrintSemicolon
\KwIn{Initial data $X_{\text{init}}$; featurization algorithm $\mathcal{A}$}
\textbf{Initialize} accuracy test and feature tests; compute initial featurization $\Ccurr := \cA(X_{\text{init}})$\;
\While{new data $X_t$ arrives}{
  \If{model or feature release at time $t$}{
    optionally reset tests\;
  }
  \If{accuracy test rejects with data $X_t$}{
    alert and update $\Ccurr$ and examine feature diffs;\\
    start new accuracy test for current featurization
  }
  \If{there are active feature tests}{
      \If{feature tests reject}{
    alert (potentially, take other action)\;
  }
      optionally reset them, do nothing, or replace them\;
    }
  $t \gets t + 1$\;
}
\end{algorithm}

For the purposes of exposition in this section, we introduce some additional notation.
A featurization \textit{algorithm} $\cA: \cX \to \cC$ takes in a set of data and computes a single featurization. Featurization \textit{error} $\err:\cC(\cX) \to [0,1]$ quantifies the quality of a featurization $C$ on a set of data $X$; for SAEs, for example, this is reconstruction error.

\textbf{Establishing a baseline.} Before the monitoring period begins, we begin with a featurization trained with an initial set of data $\widehat C_0 = \cA(X_\text{init})$, and compute its error $\eps_0 = \err(C_0(X_\text{init}))$;
we will let $\Ccurr = \widehat C_0$ and $\eps_{curr} = \eps_0$.
Based on this initial featurization, we can also identify a set of initial features $S_0$ to monitor, or otherwise let $S_0 = \emptyset$.

\textbf{Accuracy.} To maintain good accuracy over the entire time horizon, we maintain a hypothesis test for whether the error of $\Ccurr$ on new data is close to the error of $\Ccurr$ on the data with which it was trained.
That is, we test the following null hypothesis for some $\beta \geq 1$:
\begin{equation}
\cH_0^\mathrm{acc}: \mathrm{err}\left(\Ccurr(X_t)\right) \leq \beta \cdot \eps_{curr}. \label{eq:Hacc}
\end{equation}
For any time $t$ at which $\cH_0^\mathrm{acc}$ is rejected, a new featurization is recomputed on all data seen thus far. $\Ccurr$~is updated as $\Ccurr := \cA(X_{1:t})$, the error benchmark is updated $\eps_{curr} := \err (\Ccurr(X_{1:t}))$, and the procedure continues with the null in \eqref{eq:Hacc} updated with new values. We will sometimes refer to such a $t$ as a ``reject and retrain'' timestep, and use $\widehat C_s$ to denote the $s$-th featurization.

Qualitatively, a rejection at time $t$ means that the previous featurization $\Ccurr$ is no longer a high-quality representation of the most important features observed in all data up to $t$; in other words, the data stream has changed substantially. Thus, at the time of rejection, the most salient changes can also be computed---which features from the previous featurization stayed the same; which merged or split; or which became obsolete (in favor of entirely new features).
Features tracked in $S_t$ should also be revisited at ``reject and retrain'' timesteps, either updating to the new representations or choosing different features altogether.

One important detail is the level $\alpha$ at which each test is run. For a single hypothesis test, $\alpha$ straightforwardly controls the expected Type I error, but some care must be taken when multiple tests are run consecutively. Specifically, the level $\alpha_s$ at which the $s$-th test is run must be set with an appropriate schedule; as long as this occurs, we have the following guarantee.

\begin{proposition}[informal]
    \label{prop:acc}
  Let $s$ index each time that $\Ccurr$ is updated. If the test for $\cH_0^\mathrm{acc}$ at each $s$ is run with parameter $\alpha_s$ set so that $\sum_{s \geq 0}\alpha_s \leq \alpha$, then the expected proportion of ``unnecessary'' alerts at any time is at most $\alpha$.
\end{proposition}
This makes use of a result from \citet{xu2024online}; see Appendix \ref{app:thy} for proof (Proposition~\ref{prop:acc-fdr}).

\textbf{Feature monitoring.}
The \textit{composition} of features may change over time, regardless of how those features are best \textit{represented}.
Thus, we can also track a subset of features from any $\Ccurr$ for whether they appear to change meaningfully over time; we use $S_t$ to denote the set of features currently being tracked at time~$t$.
Importantly, these features can be added and removed from tracking in data-dependent ways without compromising the validity of alerts.

One natural test for any feature $i$ is whether its activation grows substantially. To formalize this, feature $i$'s activation on future samples $X_t$ must be compared to its historical activation. Let $r$ be the timestep at which feature $i$ is added to $S_r$; then, again for some $\beta \geq 1$,
this can be written as \begin{equation}
\cH_0^{(i)}: \Ccurr^{(i)}(X_t) \leq \beta \cdot \Ccurr^{(i)}(X_{0:r}). \label{eq:Hfeat}
\end{equation}
Note that this test can be instantiated for any other property of incoming data that can be expressed as a mean; we focus on feature activation for ease of exposition within our theory and algorithms, but other reasonable candidate quantities include average \textit{count} of nonzero activations (i.e., $\frac{1}{|X_t|}\sum_{X \in X_t} \mathbf 1{[\Ccurr^{(i)}(X) > 0]} $) or average \textit{label} (i.e., $\frac{1}{|X_t|}\sum_{X \in X_t} \mathbf 1{[X\textit{ labeled as } i]} $).

In Proposition~\ref{prop:feat} (formalized in Proposition \ref{prop:feat-fdr}), we summarize the feature tracking guarantee.
\begin{proposition}[informal]\label{prop:feat}
  Let $S_{t}$ be the set of features being tracked at any time $t$, and $r \leq t$ be the timestep at which the most recent update (feature addition/removal/substitution) was made to $S_{t}$.
  For each feature $i \in S_t$, maintain a level-$\alpha_i$ test for $\cH_0^{(i)}$ so that $\sum_{i \in S_{t}} \alpha_i = \alpha$.
  Then, the likelihood that an erroneous alert is sent about any of the features $i \in S_t$ at any time after $r$ is at most $\alpha$.
\end{proposition}

\textbf{Incorporating release dates.} While monitoring can be done without explicit intervention at pre-specified times, in a realistic monitoring scenario, one might be interested in checking whether changes occur after a model release or update that happens at a known time.
In fact, it is fairly straightforward to do so by re-setting the current state of the currently-active hypothesis tests;
see Appendix \ref{app:thy}.

\subsection{Application results}
\label{subsec:monitor}

We now show the results of applying this framework to the data analyzed in Section~\ref{sec:retro}.
Within this section, we use SAEs with $M=64$.
We fit the initial featurization $\widehat C_0$ with data from the first 16 weeks after ChatGPT's release (until March 23, 2023). For both accuracy and feature monitoring, we use tolerance $\beta = 1.05$ and target $\alpha = 0.1$.

\begin{figure*}
    \centering
\includegraphics[width=\linewidth]{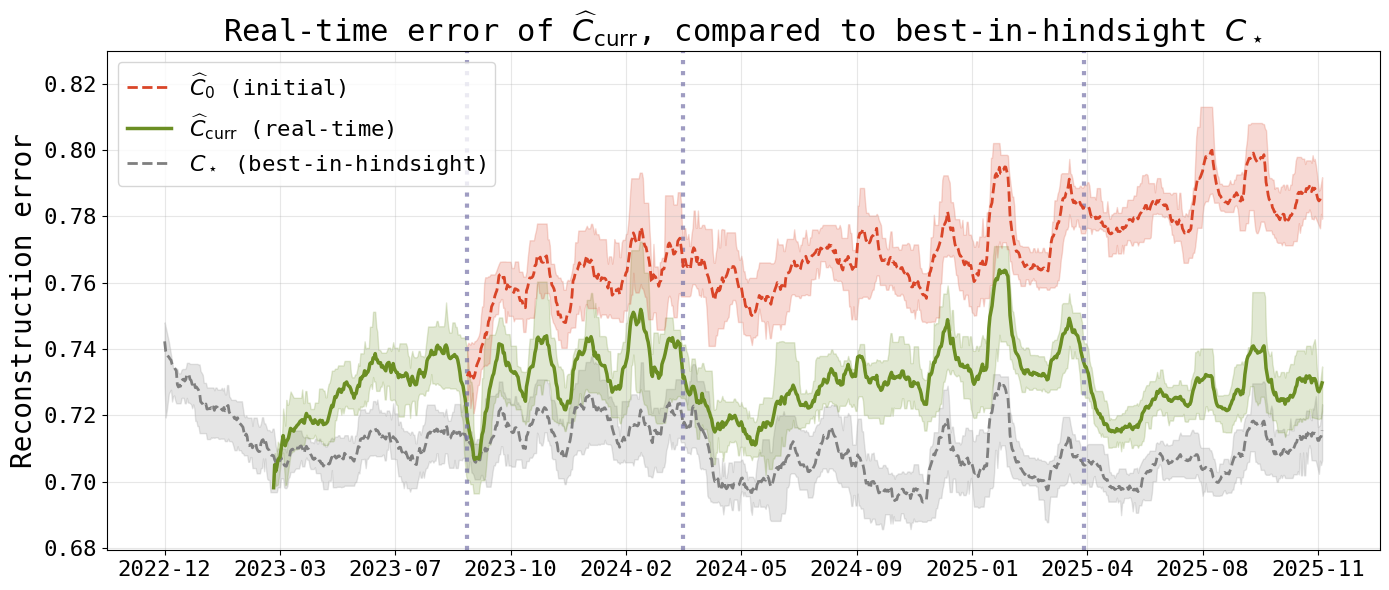}
    \caption{\footnotesize\textit{Reconstruction error of the real-time approach of Section \ref{subsec:thy} ($\Ccurr$), using $\alpha_s$ schedule at $\alpha = 0.1$, compared to reconstruction error of best-in-hindsight $C_\star$ and initial $\widehat C_0$.``Reject and retrain'' timesteps marked with dotted vertical lines.}}
    \label{fig:accuracy}
\end{figure*}

\textbf{Accuracy.}
Overall, our approach effectively maintains a sufficiently-accurate $\Ccurr$ over time.
In Figure \ref{fig:accuracy}, we show the reconstruction error of the  $\Ccurr$ maintained by our approach, compared to the reconstruction error of the ``best-in-hindsight'' $C_\star$, which was trained with all data at once, as well as the initial featurization $\widehat C_0$ computed on only $X_\text{init}$.
``Reject and retrain'' events occur on
September 9, 2023; April 4, 2024; and April 18, 2025.
That there are only three such events indicates that posts can be described by fairly stable representations over time, and validates that frequently re-computing featurizations is unnecessary. 

The evolution of featurizations overall are broadly consistent with known external changes and with the in-hindsight clusterings (see Appendix~\ref{app:matching} for details on how we compare featurizations over time). For instance, between $\widehat C_0$ and $\widehat C_1$, two new features emerge corresponding to \textit{plugins} and the \textit{ChatGPT API}, both of which were product updates from March 2023; between $\widehat C_1$ and $\widehat C_2$, a feature corresponding to \textit{controversy, danger, or bans} disappears, while one for \textit{low-quality AI-generated content} emerges; between $\widehat C_2$ and $\widehat C_3$, meanwhile, features corresponding to \textit{Google Bard} and \textit{medical topics} disappear, while features corresponding to \textit{Gemini} and \textit{personalized image requests} emerge.
In Appendix~\ref{app:pros}, we summarize all new and obsolete features across updates to $\Ccurr$ (Table~\ref{tab:feature-transitions}), and visualize some examples of feature evolutions across featurization updates (Figure~\ref{fig:evolve}).

\begin{figure}[t]
    \centering
    \includegraphics[width=\linewidth]{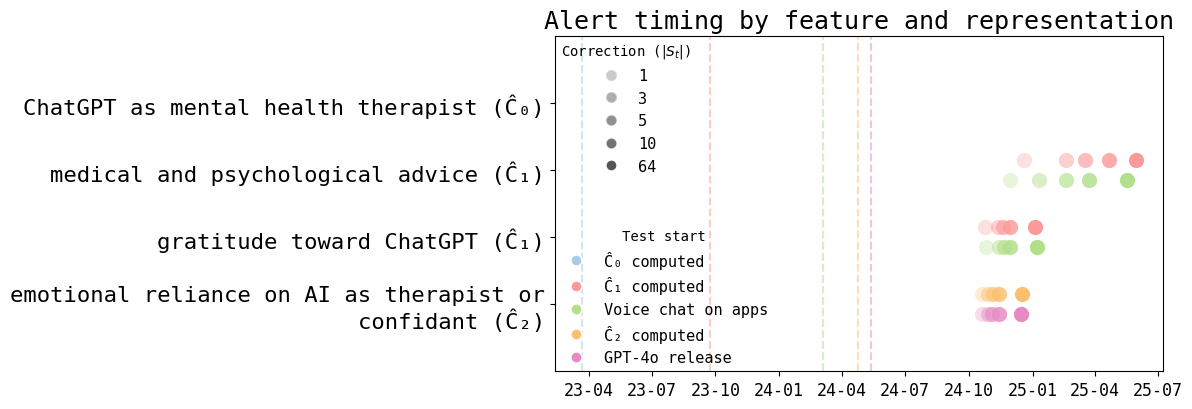}
    \caption{\footnotesize\textit{For a test of growth in therapy-related features, how soon after the test start time is an alert raised?
  Tests are run at $\alpha=0.1$, and shown with a range of Bonferroni corrections to simulate potential choices of $|S_t|$. 
  }}
    \label{fig:alerts}
\end{figure}
\textbf{Feature monitoring.}
We would also like to monitor for changes within features, even when its overall representation in the featurization remains constant.
Several features that may seem to be societally-impactful already emerge after the initial featurization $\widehat{C}_0$ computed in March 2023, including one feature explicitly about using ChatGPT as a therapist. For each featurization, we (manually) select the features that are most closely related to \textit{therapy} as test candidates. We test for changes in the frequency of posts with non-zero activations; in Figure~\ref{fig:alerts} (and Table~\ref{tab:monitoring-full}),
we show outcomes for various configurations (start dates, representations, and Bonferroni correction over $|S_t|$). 

Our alerts for the feature corresponding to therapy are raised as early as \textit{October 29, 2024.} As we discuss in Appendix~\ref{app:whoknew}, this is months earlier than OpenAI or the public seemed to be aware of psychological impact.
Notably, the quality of feature representations does appear to affect alert times. Using representations $\widehat C_0$, \textit{no tests result in alerts}, likely because the representations of the ``therapy'' feature in $\widehat C_0$ are too weak, or otherwise not fully capturing characteristics of later posts about therapy. (No tests result in alerts even for $n=1$; the bottleneck is the representation, rather than the testing of multiple features simultaneously.) 

On the other hand, while alerts for \textit{gratitude towards ChatGPT} (using the $\widehat C_1$ representations) would have been raised at similar times to the $\widehat C_2$ therapy feature, it is unclear that, at the time, \textit{gratitude} would have been considered a societally-relevant feature of interest; monitoring for the \textit{medical and psychological advice} feature, meanwhile, would have led to delayed alert times. 
Varying the number of simultaneously-monitored features (i.e., Bonferroni correction) has only a modest effect on alert timing, typically shifting dates by a few weeks for tests with strong representations. The dominant factors are the quality of the underlying representations and the strength of the actual trend.

\section{Discussion}

The time period studied in this paper---December 2022 to November 2025---is a unique moment in recent history in which consumers were introduced to, then quickly adapted to, a genuinely-unprecedented type of technology. While Section~\ref{subsec:domestication} tells a story of adoption that may seem mundane in hindsight, Section~\ref{subsec:emotional} also suggests that emotional engagement is a crucial dimension of adoption that evolved in parallel. Of course, there is more to see: 
r/ChatGPT is an incredibly rich set of data, and there are a wide range of relevant further questions---such as more detailed analysis of emotional engagement or the development of intra-subreddit community norms---that we hope future work will explore. 

More generally, this work can be seen as a proof-of-concept for an approach to AI evaluation that makes use of \textit{public feedback}. 
We began from the perspective that it is worth paying attention to what everyday users have to say about their experiences with real-world AI products. While analyzing such data has long been a cornerstone of the social sciences, we argue that feedback from the general public is not only sociologically interesting, but also a crucial means for identifying ``unknown unknowns'' in societally-consequential consumer AI products. 
While social media is one natural way to collect this type of data, it is worth considering the possibility of platforms that are purpose-built to seek feedback for evaluation directly, especially in light of recent regulatory movement towards allowing individuals to contest or report their experiences with AI systems. 

Better information can lead to better decisions. 
Understanding how users may be experiencing AI products---especially in unexpected ways, and especially in real time---is a pathway to \textit{steering} the societal impact of these technologies, rather than \textit{reacting} to them in hindsight. 
OpenAI's initial choice in August 2025 to sunset GPT-4o in favor of the ``colder'' GPT-5 was clearly deliberate,
but the strength of users' emotional responses upon the GPT-5 release suggests that OpenAI's expectations were miscalibrated.
Yet, as the previous sections show, meaningful signal about emotional engagement existed well before GPT-5.
Counterfactual outcomes will always be unknown, and we make no claim about what should have been done with that information. 
We do claim that the information was there---\textit{if anyone had been paying attention}. 
Perhaps, in the future, we should do exactly that. 

\newpage
\section*{Acknowledgements}

We are grateful to Rajiv Movva, Fiona Y. Chen, Brian W. Lee, Arul Murugan, and Kevin Black for fruitful conversations in the development of this work. 

This work was supported in part by the UK AI Security Institute Challenge Fund GAP-PRD-20250725-219733-54280; by the United States National Science Foundation under grants CCF-2145898, 2326498, and 2142419; by the Office of Naval Research under grants N00014-24-1-2159 and N00014-20-1-2497; an Amazon Research Award; an Alfred P. Sloan
fellowship; Schmidt Sciences AI2050 Early Career Fellowships; a Google Research Scholar award; a CIFAR Azrieli Global scholarship; the LinkedIn-Cornell Bowers CIS Strategic Partnership; the Survival and Flourishing Fund; Coefficient Giving; the Zhang Family Endowed professorship; and the John D. and Catherine T. MacArthur Foundation.

\bibliographystyle{plainnat}
\bibliography{refs}

\newpage
\appendix
\section{Background}
\label{app:prelims}

\subsection{Additional context on ChatGPT and r/ChatGPT}

In Figure~\ref{fig:wau}, we illustrate the growth of consumer usage over time (reproduced from Figure 3 in \citet{chatterji2025people}), along with the number of ``subscribers'' to the r/ChatGPT subreddit (data collected via snapshots of the r/ChatGPT homepage from archive.today). 
\begin{figure}[h]
  \centering
  \includegraphics[width=0.8\textwidth]{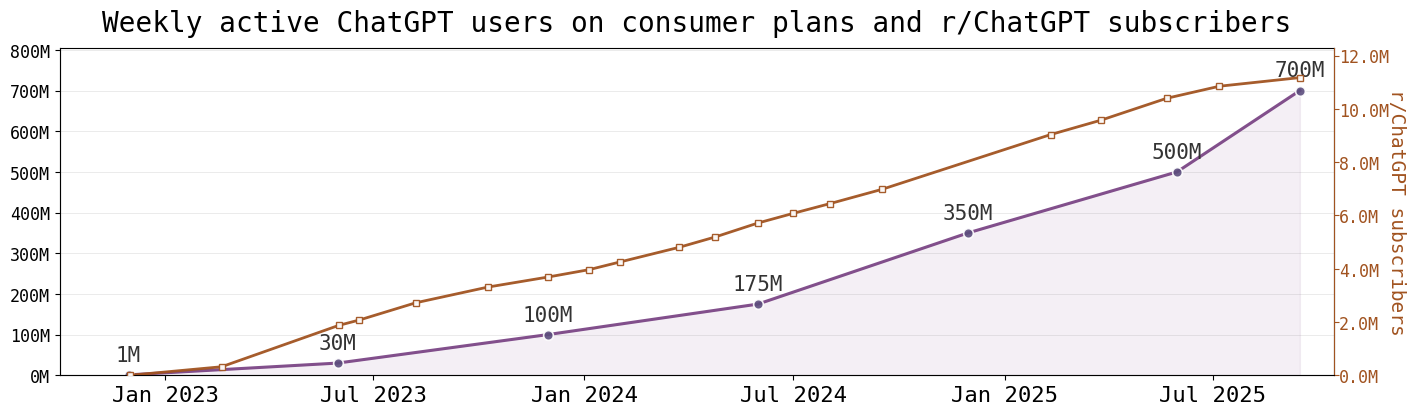}
  \caption{\footnotesize\textit{ChatGPT weekly active users (purple; labels on left axis) and r/ChatGPT subscribers (brown; labels on right axis) over time.}}
  \label{fig:wau}
\end{figure}

Note that the number of \textit{posts} on the subreddit does not increase with the number of \textit{subscribers}.
Our dataset includes posts from 89346 unique users. The average number of posts per user is 1.53, and median 1; in fact, the vast majority of posters are very infrequent. The top 20\% of frequent posters post twice; the top 5\% post 3 times; and the top 2\% post 5 times. Only 32 users had more than 50 posts.
Figure \ref{fig:users} indicates that throughout the lifetime of the subreddit, around half of daily posts are made by first-time users; overall, most post activity does not appear to be driven by superusers.

\begin{figure}[h]
  \centering
  \includegraphics[width=0.8\textwidth]{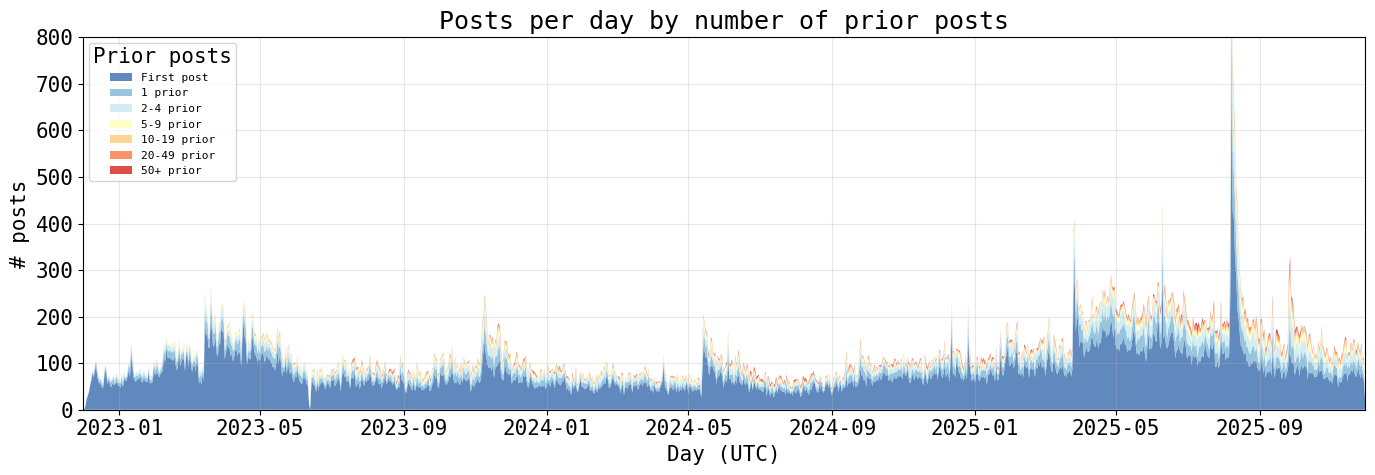}
  \caption{\footnotesize\textit{Posts per day, colored by user type (number of prior posts made by that user). Around half of daily posts are from new users.}}
  \label{fig:users}
\end{figure}

\subsection{Who knew what about emotional engagement usage, and when?}
\label{app:whoknew}
This is, for obvious reasons, a difficult question to answer in hindsight; however, we will attempt to ground our discussion in published materials (for understanding the state of public discourse) and official communications (for OpenAI). We summarize this history in Figure~\ref{fig:history} and provide details below. 

\begin{figure}[h]
  \centering
  \includegraphics[width=\textwidth, trim=0pt 200pt 0pt 200pt,clip]{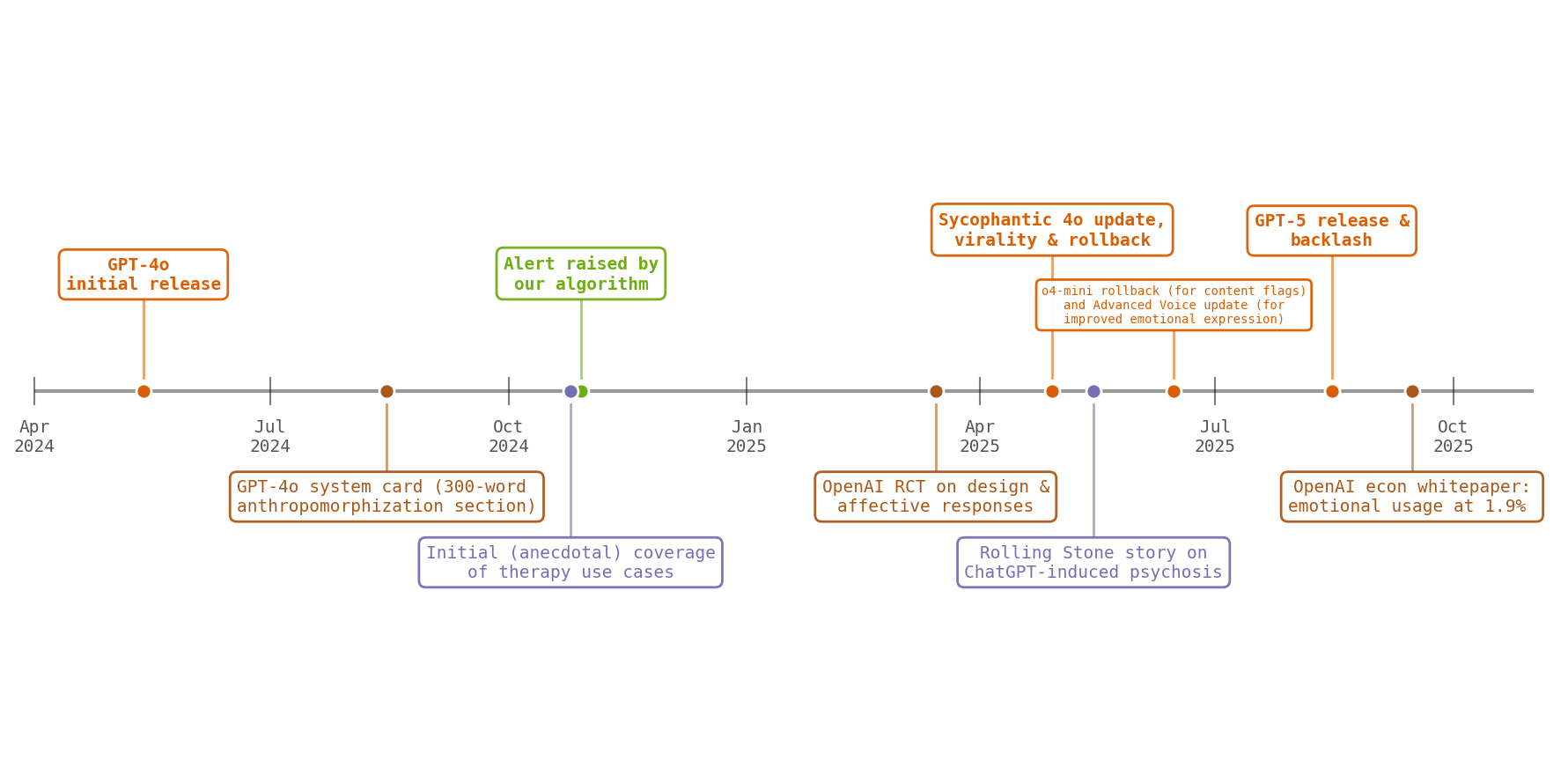}
  \caption{\footnotesize\textit{Timeline of ``public'' knowledge about emotional engagement and sycophancy in GPT-4o. Major events in \textbf{bold}; model releases and product updates in orange; public OpenAI communications in brown; and news media in blue.}}
  \label{fig:history}
\end{figure}

In terms of public discourse, some news outlets discussed therapy as a growing use case for chatbot products in late 2024, and possibly earlier.\footnote{See, e.g., \url{https://www.bloomberg.com/news/newsletters/2024-12-24/ai-developers-see-opportunity-in-offering-chatbots-for-therapy}, \url{https://www.washingtonpost.com/business/2024/10/25/ai-therapy-chatgpt-chatbots-mental-health/}}
However, at the time, little was known about the nature and extent of this usage; most coverage focused on anecdotal personal stories. While researchers have studied the (dis)utility of LLM therapists over the last three years (e.g., \citet{khawaja2023your,moore2025expressing}), these works are typically about evaluating the quality and fitness of AI models as therapists, rather than studying adoption by real-world users.
To the best of our knowledge, the earliest reported story about ChatGPT-induced psychosis came from Rolling Stone in May 2025;\footnote{\url{https://www.rollingstone.com/culture/culture-features/ai-spiritual-delusions-destroying-human-relationships-1235330175/}}
since then, the psychological impacts of long-run engagement with chatbot products appear to become much more commonplace.\footnote{e.g., \url{https://www.wsj.com/tech/ai/chatgpt-chatbot-psychology-manic-episodes-57452d14}, \url{https://www.wired.com/story/chatgpt-psychosis-and-self-harm-update/}, \url{https://www.nytimes.com/2025/06/13/technology/chatgpt-ai-chatbots-conspiracies.html}}

For OpenAI, while we cannot make claims about what was known internally, we will summarize some relevant public communications in chronological order.
The GPT-4o system card, published in \textit{August 2024}, includes a 300-word section on ``Anthropomorphization and Emotional Reliance''  \citep{openai2024gpt4o}. This section notes empirical observations of language that might suggest users forming connections with the model---indicating that emotional reliance was known as a potential concern at least since August.
In \textit{March 2025}, OpenAI released results from a four-week RCT on product decisions that affect the degree to which users may develop affective responses to the technology; while this study used the 4o model, it is unclear when the four weeks of experimentation occurred \citep{fang2025ai}.

In \textit{April 2025}, an extremely-sycophantic update to 4o triggered social media virality, a rollback, and several blog post updates.\footnote{\url{https://openai.com/index/sycophancy-in-gpt-4o/}}
In \textit{June 2025}, a rollback of o4-mini was made due to increased rates of content flags, the first such rollback (to the best of our knowledge);\footnote{\url{https://help.openai.com/en/articles/9624314-model-release-notes}} the same month, however, an update to Advanced Voice was made with ``enhancements in intonation and naturalness, making interactions feel more fluid and human-like....it speaks with more on-point expressiveness for certain emotions including empathy.''

Days before the GPT-5 release in \textit{August 2025}, a blog post emphasizes consultations with experts in designing behavior for GPT-5, and responsible treatment of personal and emotional struggles.\footnote{\url{https://openai.com/index/optimizing-chatgpt/}}
Unfortunately, the blowback to the release was well-documented, in this paper and elsewhere; on Twitter a few days after the release, Altman claimed that mental health and personal usage is something that OpenAI has ``been closely tracking for the past year or so.''\footnote{
\url{https://x.com/sama/status/1954703747495649670}}
The \textit{September 2025} working paper from OpenAI's economics team notes that ``only'' around 1.9\% of all chats can be classified as emotional engagement \citep{chatterji2025people}.

\section{Methodological details}
\label{app:method}

\subsection{Design decisions for initial featurization}

\textbf{Sample weights.} Because r/ChatGPT is such a high-volume subreddit, we use sample weights as a proxy for measuring how significantly each post contributed to community discussion.
Thus, we weight posts by (the logarithm of) both ``score'' ($n_\text{upvotes} - n_\text{downvotes}$) and by the number of comments; this is because many low (or zero)-scoring posts have a substantial number of comments (perhaps suggesting controversiality), while some high-scoring posts have few comments (perhaps suggesting broad agreement).
While our work does not study exactly what perspectives the subreddit expresses, both cases described in the previous sentence provide evidence of posts that were important to the subreddit in some way.

\textbf{Frequencies instead of counts.} Throughout this work, we intentionally track trajectories of (daily) \textit{frequencies}, and how they change over time, rather than raw \textit{counts}. One reason for doing so is to reduce the impact of variation in daily (and long-run) post volume; as illustrated in Figure \ref{fig:users}, overall post volume varies substantially over the three-year window.

This, of course, has some limitations. For example, one failure mode would be if the count of posts about topic X remained constant over time, but new posts about Y began to arrive. In this case, it would appear that the frequency of posts about X decreased, even if the true exogenous phenomenon (new posts about Y) had nothing to do with X, leading to erroneous conclusions about the dynamics of X.
However, in such a scenario, overall post volume would show the growth due to Y.
Figure \ref{fig:users} also suggests that this is not the case for our data. The trends in overall post volume do not track any of the trends in features we identify as having changed, and the distribution of new/returning users each day also remains relatively stable.

An additional reason to track frequencies instead of counts is that frequencies provide some signal of what makes up ``the community of r/ChatGPT''---though community dynamics are not the focus of our work, the distribution of topics in a forum can itself be thought of as an intrinsically interesting object of study.

\textbf{Choice of $M$.}
Our choice of $M=128$ is fairly generous. As we discuss, we remove a substantial fraction of the 128 features initially discovered---and in fact, at any particular time, most of the dataset can be covered by less than 64 features, even when generic features are removed (see Figure \ref{fig:coverage}).
However, we are intentionally conservative with this choice of $M$: the set of features that provide 75\% coverage in January 2023 is distinct from the set of features that provide 75\% coverage in November 2025, which a higher $M$ can accommodate.
Moreover, allowing a higher $M$ enables not just the discovery of more granular and nuanced features, but also those that appear only transiently in the dataset (such as short-term spikes) that might otherwise be absorbed into other features for smaller $M$.

\begin{figure}[H]
  \centering
  \includegraphics[width=0.75\textwidth]{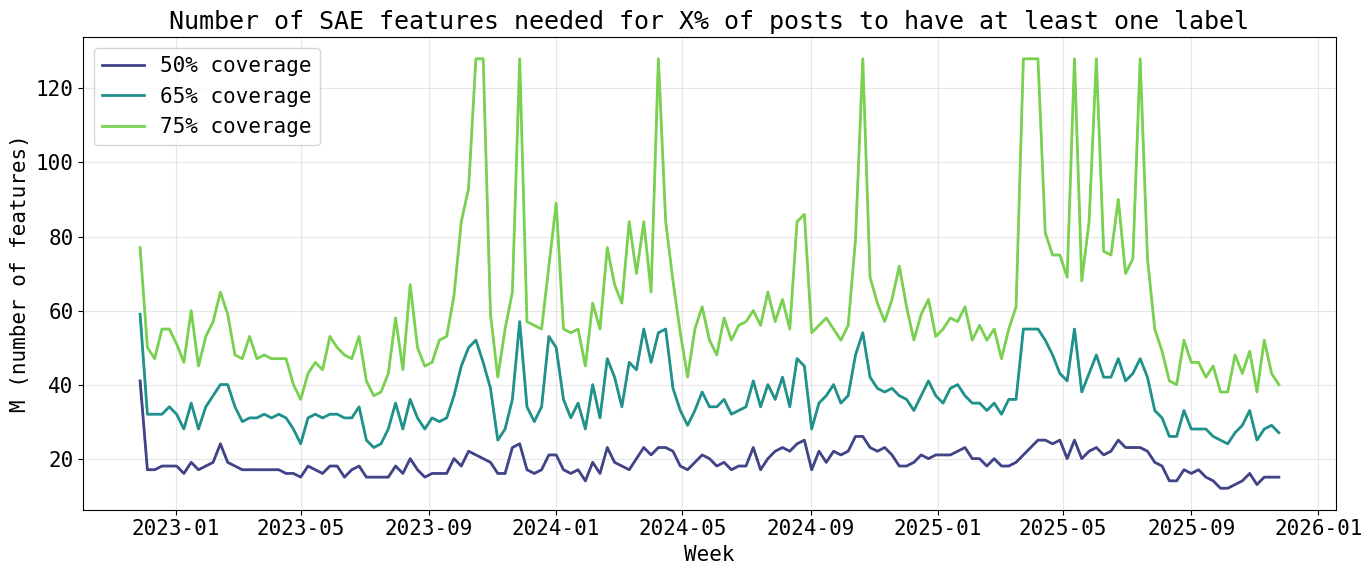}
  \caption{\footnotesize\textit{Number of features needed to ``cover'' most data in a given week, over time.}}
  \label{fig:coverage}
\end{figure}

\subsection{Details for temporal trajectories}
\label{app:temporal}

We analyze trajectories according to the procedure below.\footnote{Since we model all trajectories as solely (piecewise) linear, our approach is less well-suited to describe features that exhibit temporary ``spikes'', or brief surges of discussion about specific topics; that said, we are mainly interested in long-run changes.}
Within this section, we use $\{y_t\}_{t \in [T]} := \frac{1}{|X_t|} \sum_{X \in X_t} \mathbf{1}_{[X\textit{ labeled with } i]}$ to denote the actual transcript of (smoothed) daily frequencies, and 
describe the process for a single feature (dropping the feature index $i$ for clarity).

\subsubsection{Fitting changepoints}
To characterize temporal dynamics, we fit piecewise-linear trends to each feature's daily label frequency smoothed (with 30-day rolling mean) according to Equation \eqref{eq:fit}. 
That is, we approximate $\{y_t\}_{t \in [T]}$ with
\[\lambda(t) = \beta_0 + \sum_{j \in |\cT|} \gamma_j \max(0, t - \tau_j).\]

\textbf{Fitting a single trajectory.}
Given an active set $S \subseteq [|\cT|]$ of changepoints, we
fit $(\beta_0, \{\gamma_j\}_{j\in S})$ by minimizing the Poisson NLL
\begin{equation}
  \mathcal{L}(S) \;=\;
  \min_{\beta_0,\, \{\gamma_j\}_{j\in S}}
  \sum_{t=0}^{T}\Big[\lambda(t) - y_t \log \lambda(t)\Big]
  \quad\text{s.t.}\quad \lambda(t)\ge\varepsilon \ \forall t,
  \label{eq:nll}
\end{equation}
with $\varepsilon = 10^{-8}$ being a positivity constraint. Equation~\eqref{eq:nll} is convex in
$(\beta_0, \gamma_S)$ and is solved with \texttt{cvxpy}.

We select $S$ by dynamic programming over $\cT$: for each $s = 0, \dots, k_{\max}$
with $k_{\max} = 12$, the DP returns the size-$s$ active set $\hat S_s$
that minimizes $\mathcal{L}(S)$ in~\eqref{eq:nll}. We iterate upward starting with $s=0$ and stop at the smallest $s$
for which the relative improvement falls below $\eta = 0.01$, i.e., 
$  s^\star \;=\; \min\left\{\,s :\; \frac{\cL(\widehat S_s) - \cL(\widehat S_{s+1})}{\sum_t y_t} < \eta \,\right\}.$

After initial selection, we prune changepoints that appear to have minimal impact on ``slope.''
For each
$j \in \widehat S_{s^\star}$, define
$  V_j \;=\; \frac{|\hat\gamma_j|\cdot \min(\Delta_\mathrm{prev},\, \Delta_\mathrm{next})}{r_y},$ where $\Delta_\mathrm{prev}$ (resp., $\Delta_\mathrm{next}$) is the number of days between $\tau_j$ and the \textit{previous} (resp., \textit{next}) neighbor of $\tau_j$ in $\widehat S_{s^\star}$, and $r_y = y_{\max} - y_{\min}$ is the full $y$-range.
We drop any $j$ where $V_j \leq 0.05$, i.e. where the change in the fitted values before or after $\tau_j$ is less than 5\% of the total $y$-range. 

\textbf{Bootstrapping.}
For $b \in [100]$, we draw
$N$ post indices with replacement from the $N$ posts, recompute
$\{y^{(b)}_t\}_t$, and run the procedure above to obtain
$\hat S^{(b)} \subseteq \cT$. 
For each $\tau_j \in \cT$,
$  \mathrm{stab}(\tau_j) \;=\;
  \frac{1}{100} \sum_{b=1}^{100}
  \mathbf 1\!\left[j \in \hat S^{(b)}\right] \in [0,1],$
and we use a threshold 
$\mathrm{stab}(\tau_j) \ge 0.5$ as a heuristic to consider $\tau_j$ \textit{stable}.

\subsubsection{Full-range slope tests}
For each feature $i$, we test whether the early-to-late relative change in
$y^i_t$ exceeds $\theta = 1.10$. Considering the full time range as a
single sequence, we fit
$  y_t \;=\; a + b\, t + \epsilon_t$
by OLS with Newey-West HAC standard errors at bandwidth
$L = \max\!\big(\lfloor 4 (T/100)^{2/9}\rfloor,\, w_{\text{smooth}}\big)$,
where $w_{\text{smooth}} = 30$ covers the rolling-mean autocorrelation.
Let $\hat r = (\hat a + \hat b\cdot T)/\hat a$ be the fitted
relative change. We run a one-sided $t$-test of
\[
  \cH_0:\; \tfrac{1}{\theta} \le \hat r \le \theta
  \quad\text{vs.}\quad
  \cH_1:\; \hat r > \theta \ \text{or}\ \hat r < \nicefrac{1}{\theta},
\]
choosing the side from $\mathrm{sign}(\hat b)$. Reported $p$-values are Bonferroni-corrected across 86 features at $\alpha = 0.05$.

Some features are non-monotonic but have ``peaks'' early (in 2023) or late (in 2025) in the timeframe. To fairly compute the slope tests in these cases, we check for changepoints that have positive slope before, negative slope after, and are either in 2023 or in 2025. For such features, we use these changepoints as the start/end of the slope test instead.

\subsection{Details for finding families of related features}

Similarity in \textit{co-occurrence} is measured as the Pearson correlation between label vectors across all posts.
$S_{ij}^{\text{corr}} = \text{corr}\left( a_i, a_j \right)$
where $a_i = (a_i^{(1)}, \ldots, a_i^{(T)})$ is the vector of feature $i$'s label across all $T$ posts.
For similarity in \textit{trajectory},  feature trajectories were aggregated into weekly means and z-score normalized across time. Pairwise trajectory similarity was then computed as the Pearson correlation between these normalized time series, i.e., $S_{ij}^{\text{traj}} = \text{corr}\left( \tilde{z}_i, \tilde{z}_j \right)$,
where $\tilde{z}_i = \frac{\bar{a}_i^{(t)} - \mu_i}{\sigma_i}$ is the z-scored weekly mean label of feature $i$, and $\bar{a}_i^{(t)} = \frac{1}{|W_t|}\sum_{d \in W_t} a_i^{(d)}$ is the mean label over week $t$.

We use scipy's built-in hierarchical clustering implementation (with the Ward method) to compute clusterings using $S_{ij}^{\text{corr}}$ alone,  $S_{ij}^{\text{traj}}$, and an equally-weighted combination of the two (i.e., $S_{ij}^{\text{combined}} = \tfrac12 (S_{ij}^{\text{corr}} + S_{ij}^{\text{traj}})$).
We compare these clustering approaches on the \textit{emotion} subset of features in Appendix \ref{app:emotion}, and compare them to our overall manual categorizations in Appendix~\ref{app:results}.

\subsection{Matching features across different featurizations}  
\label{app:matching}

Learned featurizations are unordered---that is, for two different featurizations $C_A$ and $C_B$, a feature indexed as $i$ in $C_A$ has no inherent relationship to (e.g.) $C_B^{(i)}$. Thus, we must manually match features between different featurizations in order to measure the degree to which they identify similar features. 
Intuitively, two representations of the same feature should have similar activation profiles; features about the concept ``login problems'' should activate strongly on posts about login problems and not at all on posts that are not. Thus, given two activation matrices $C_A(X)$ and $C_B(X)$ on the same set of data $X$, we would like to compute a matching between the indices in $A$ and those in $B$.

\textbf{Step 1: Computing similarities between $A$ and $B$.}
Let $C_A, C_B \in \mathbb{R}^{m \times n}$ denote activation matrices from two learners evaluated on the same $n$ posts. We define a similarity matrix $S \in \mathbb{R}^{m \times m}$ using cosine similarity across posts, where index $j$ refers to a feature from $C_A$ and index $k$ refers to a feature from $C_B$ as
\[
S_{j,k}
=
\frac{\langle C^{(j)}_{A}, C^{(k)}_{B} \rangle}
{\| C^{(j)}_{A} \| \, \| C^{(k)}_{B} \|}.
\]

\textbf{Step 2: Constructing a null distribution of matched similarities.}
To quantify the degree to which a feature matching improves upon the baseline of completely random matchings, 
we run a permutation test as follows. For each of $P$ permutations, we randomly shuffle the activations of each feature in $C_B$ independently across posts, compute the resulting similarity matrix $S^{(\text{null})}$, and extract optimal matching assignments via the Hungarian algorithm. 

We choose $\tau$ to be the $(1-\alpha)$ quantile of all matched null similarities over the $P$ permutations; a pair $(j, k)$ matching a feature from $A$ to a feature from $B$ is considered a significant match only if $S_{j,k} \geq \tau$.

\textbf{Step 3: Characterizing matches.}
Rather than forcing a one-to-one matching upfront, we construct a bipartite graph
$G = (V_A \cup V_B, E),$
where \(V_A\) indexes features from \(C_A\), \(V_B\) indexes features from \(C_B\), and an edge \((j,k) \in E\) exists whenever \(S_{j,k} \ge \tau\).

We then decompose \(G\) into connected components $G^{(c)}$ so that $G = \cup_c G^{(c)}$. For a connected component \(G^{(c)}\), let \(V_A^{(c)} \subseteq V_A\) and \(V_B^{(c)} \subseteq V_B\) denote the subsets of vertices from the two sides that appear in that component. We then characterize each component by the sizes of \(V_A^{(c)}\) and \(V_B^{(c)}\). 
\begin{itemize}
    \item \textbf{Match}: \(|V_A^{(c)}| = 1\), \(|V_B^{(c)}| = 1\): indicates feature in $A$ is similar to one feature in $B$
    \item \textbf{Split}: \(|V_A^{(c)}| = 1\), \(|V_B^{(c)}| > 1\): indicates feature in $A$ is similar to multiple features in $B$
    \item \textbf{Merge}: \(|V_A^{(c)}| > 1\), \(|V_B^{(c)}| = 1\): indicates feature in $B$ is similar to  multiple features in $A$
    \item \textbf{Unmatched}: $|V_A^{(c)}| = 1, |V_B^{(c)}| = 0$: indicates features in $A$ that are not present in $B$, and vice versa. 
\end{itemize}
Finally, when \(|V_A^{(c)}| > 1\), \(|V_B^{(c)}| > 1\), this indicates a \textbf{many-to-many} matching. In these cases, we apply the Hungarian algorithm to this subgraph to extract a primary $1$-to-$1$ matching, and classify remaining edges as splits or merges.

\textbf{Step 4: Pruning.}
We prune matched edges from the previous step as follows. A split or merge edge $(j, k)$ is discarded if both endpoints have strictly higher-scoring alternatives elsewhere in the graph. 

\subsection{SAEs vs. other featurization methods}
\input{arxiv-kmeanspca.tex}

\section{Supporting materials for main results}
\label{app:retro}

\subsection{Section \ref{subsec:domestication} (``domestication'')}
In Tables~\ref{tab:domestic-full} and \ref{tab:adoption}, we provide our full categorization of adoption- and (non-emotional) usage- related features. Plots of these features are also provided in Appendix~\ref{app:results}. 

\textbf{Basic use} features, which suggest initial exploration of ChatGPT as new users begin using the product, are generally constant or decline over time. \textbf{Advanced usage} features, which suggest more involved usage, increase in frequency. Features categorized in \textbf{Customization} and \textbf{Model or product improvements} both decline, both perhaps due to a combination of increased expertise and improved models. \textbf{Temporary bugs} appear fairly consistently. \textbf{Applications} decrease overall, with the exception of \textit{medical conditions and diagnoses}, which increases. 

Adoption-related features have less-consistent trends within categories. \textbf{Language and terminology} features that suggest specific word choice generally decline or remain constant. Explicit references to the \textbf{Subreddit community} appear to decline, though the only feature that passes the significance threshold is about sharing user projects, which is likely due in part to a known subreddit policy change in mid-2024. \textbf{Perspectives} include users' judgments on the ChatGPT product and/or its societal context that are not directly tied to specific instances of usage. \textbf{Product updates} are direct (negative) responses to recent product updates, while \textbf{Jailbreaking \& content policy} features often peak either at the beginning of the timeline (in early 2023) or near the end (around the GPT-5 release).

These categorizations are our qualitative interpretation of our quantitative results. If features appeared to reasonably belong to multiple categories based on both meaning and quantitative results (e.g., ``hallucinations'' in both advanced usage and in language and terminology, or ``medical applications'' in both applications and emotional engagement), we deferred to our judgment of the feature's meaning by reading sample posts. We acknowledge that there may be alternative categorizations, but believe they would not substantively affect our overall interpretation. 

\begin{table*}[t]
  \centering
  \footnotesize
  \setlength{\tabcolsep}{2pt}
  \renewcommand{\arraystretch}{0.92}
  \begin{tabular}{@{}p{0.15\textwidth} >{\raggedright\arraybackslash}p{0.4\textwidth} >{\raggedleft\arraybackslash}p{0.06\textwidth}@{\kern 2pt} p{0.07\textwidth}@{} >{\raggedleft\arraybackslash}p{0.06\textwidth}@{\kern 2pt} p{0.07\textwidth}@{\kern 12pt} c c l@{}}
    \toprule
    \textit{Category} & \textit{Feature} & \textit{Early} &  & \textit{Late} &  & \multicolumn{3}{c}{\textit{Change}} \\
    \midrule

    \multirow{6}{*}{\shortstack[l]{Basic use and\\exploration}}
    & \feat{recommendations for AI tools}                    & 5.8\%  & {\scriptsize (Jul'23)} & 2.2\% &  & $\downarrow$ & $\times$ & 0.28 \\
    & \feat{questions about access, versions, pricing}       & 5.3\%  &        & 2.5\% &  & $\downarrow$ & $\times$ & 0.29 \\
    & \feat{login problems}                                  & 2.0\%  &        & 1.6\% &  & $\downarrow$ & $\times$ & 0.37 \\
    & \feat{requests for help}                               & 3.7\%  &        & 2.8\% &  & $\downarrow$ & $\times$ & 0.67 \\
    & \grayfeat{questions about trying specific features}    & \textcolor{gray}{2.0\%}  &        & \textcolor{gray}{1.6\%} &  &  & \textcolor{gray}{$\times$} & \textcolor{gray}{0.71} \\
    & \grayfeat{model or version preference comparisons}     & \textcolor{gray}{0.7\%}  &        & \textcolor{gray}{1.7\%} &  &  & \textcolor{gray}{$\times$} & \textcolor{gray}{0.71} \\
    & \grayfeat{pricing and free vs paid comparisons}        & \textcolor{gray}{4.0\%}  &        & \textcolor{gray}{5.5\%} &  &  & \textcolor{gray}{$\times$} & \textcolor{gray}{0.91} \\
    \midrule

    \multirow{12}{*}{\shortstack[l]{Advanced usage}}
    & \feat{false or fabricated information}                 & 1.7\%  &        & 2.8\% &  & $\uparrow$   & $\times$ & 1.3 \\
    & \feat{organizing or searching chat histories}          & 1.5\%  &        & 3.2\% &  & $\uparrow$   & $\times$ & 1.6 \\
    & \feat{requests to turn specific features off}          & 1.8\%  &        & 4.0\% &  & $\uparrow$   & $\times$ & 1.7 \\
    & \feat{lost, deleted, or missing conversations}         & 1.7\%  &        & 3.2\% &  & $\uparrow$   & $\times$ & 1.9 \\
    & \feat{memory features and data saving}                 & 0.8\%  &        & 3.1\% &  & $\uparrow$   & $\times$ & 3.8 \\
    & \feat{cross-chat data leaks}                           & 0.3\%  &        & 3.1\% &  & $\uparrow$   & $\times$ & 5.8 \\
    & \feat{hallucinations}                                  & 0.2\%  &        & 1.4\% &  & $\uparrow$   & $\times$ & 8.0 \\
    & \dagfeat{failing to follow user instructions}          & 1.5\%  &        & 4.2\% &  & $\uparrow$   & $\times$ & 1.6 \\
    & \grayfeat{AI recognizing or admitting mistakes}        & \textcolor{gray}{2.7\%}  &        & \textcolor{gray}{2.5\%} &  &  & \textcolor{gray}{$\times$} & \textcolor{gray}{0.68} \\
    & \grayfeat{formatting and copy-paste issues}            & \textcolor{gray}{0.2\%}  & \textcolor{gray}{\scriptsize (May'23)} & \textcolor{gray}{0.1\%} &  &  & \textcolor{gray}{$\times$} & \textcolor{gray}{0.82} \\
    & \grayfeat{tool usage questions}                        & \textcolor{gray}{1.1\%}  &        & \textcolor{gray}{1.6\%} &  &  & \textcolor{gray}{$\times$} & \textcolor{gray}{1.1} \\
    & \grayfeat{questions about daily or repeated AI use}    & \textcolor{gray}{2.5\%}  &        & \textcolor{gray}{2.3\%} &  &  & \textcolor{gray}{$\times$} & \textcolor{gray}{1.1} \\
    \midrule

    \multirow{4}{*}{\shortstack[l]{Customization}}
    & \feat{tools and extensions}                            & 2.1\%  &        & 0.2\% &  & $\downarrow$ & $\times$ & 0.00 \\
    & \feat{fine-tuning GPTs with user-provided data}        & 0.9\%  &        & 0.1\% &  & $\downarrow$ & $\times$ & 0.08 \\
    & \feat{prompts and prompting}                           & 6.4\%  &        & 3.3\% &  & $\downarrow$ & $\times$ & 0.38 \\
    & \dagfeat{custom instructions}                          & 2.4\%  & {\scriptsize (Sep'23)} & 0.7\% &  & $\downarrow$ & $\times$ & 0.10 \\
    \midrule

    \multirow{4}{*}{\shortstack[l]{Model or product\\improvements}}
    & \feat{knowledge cutoff discussions}                    & 0.9\%  &        & 0.6\% &  & $\downarrow$ & $\times$ & 0.23 \\
    & \feat{browser issues or browser extensions}            & 1.5\%  &        & 1.3\% &  & $\downarrow$ & $\times$ & 0.35 \\
    & \feat{message limits or caps}                          & 3.3\%  &        & 1.5\% &  & $\downarrow$ & $\times$ & 0.54 \\
    & \grayfeat{PDF upload or summarization}                 & \textcolor{gray}{0.4\%}  &        & \textcolor{gray}{0.8\%} &  &  & \textcolor{gray}{$\times$} & \textcolor{gray}{0.65} \\
    \midrule

    \multirow{4}{*}{\shortstack[l]{Temporary bugs}}
    & \feat{slow or lagging response times}                  & 1.1\%  &        & 2.2\% &  & $\uparrow$   & $\times$ & 1.7 \\
    & \feat{error messages and technical problems}           & 3.0\%  &        & 7.2\% &  & $\uparrow$   & $\times$ & 1.8 \\
    & \grayfeat{ChatGPT down or unavailable}                 & \textcolor{gray}{3.3\%}  &        & \textcolor{gray}{2.0\%} &  &  & \textcolor{gray}{$\times$} & \textcolor{gray}{0.35} \\
    & \grayfeat{ChatGPT failing to process inputs}           & \textcolor{gray}{3.2\%}  &        & \textcolor{gray}{3.9\%} &  &  & \textcolor{gray}{$\times$} & \textcolor{gray}{0.93} \\
    \midrule

    \multirow{17}{*}{\shortstack[l]{Applications}}
    & \feat{programming}                                     & 5.0\%  &        & 1.6\% &  & $\downarrow$ & $\times$ & 0.18 \\
    & \feat{education or studying}                           & 5.1\%  &        & 1.2\% &  & $\downarrow$ & $\times$ & 0.19 \\
    & \feat{AI text detection for student work}              & 1.4\%  &        & 0.5\% &  & $\downarrow$ & $\times$ & 0.20 \\
    & \feat{D\&D and role-playing games}                     & 1.2\%  &        & 0.4\% &  & $\downarrow$ & $\times$ & 0.24 \\
    & \feat{math and problem-solving}                        & 1.9\%  &        & 0.8\% &  & $\downarrow$ & $\times$ & 0.29 \\
    & \feat{songwriting}                                     & 1.9\%  &        & 0.8\% &  & $\downarrow$ & $\times$ & 0.36 \\
    & \feat{riddles and logic problems}                      & 1.1\%  &        & 0.3\% &  & $\downarrow$ & $\times$ & 0.40 \\
    & \feat{job applications and resumes}                    & 0.6\%  &        & 0.4\% &  & $\downarrow$ & $\times$ & 0.42 \\
    & \feat{marketing, advertising, business growth}         & 2.4\%  &        & 1.5\% &  & $\downarrow$ & $\times$ & 0.51 \\
    & \feat{language use, translation, multilingual}         & 1.4\%  &        & 1.4\% &  & $\downarrow$ & $\times$ & 0.53 \\
    & \feat{medical conditions or diagnoses}$^\ddagger$      & 0.4\%  &        & 1.4\% &  & $\uparrow$   & $\times$ & 2.4 \\
    & \dagfeat{investing, finance, or wealth topics}         & 1.2\%  &        & 0.8\% &  & $\downarrow$ & $\times$ & 0.59 \\
    & \dagfeat{movies, posters, and film}                    & 1.5\%  &        & 0.4\% &  & $\downarrow$ & $\times$ & 0.61 \\
    & \grayfeat{creative writing}                            & \textcolor{gray}{6.0\%}  &        & \textcolor{gray}{2.5\%} &  &  & \textcolor{gray}{$\times$} & \textcolor{gray}{0.31} \\
    & \grayfeat{legal advice and lawsuits}                   & \textcolor{gray}{1.0\%}  & \textcolor{gray}{\scriptsize (Jul'23)} & \textcolor{gray}{0.7\%} &  &  & \textcolor{gray}{$\times$} & \textcolor{gray}{0.40} \\
    & \grayfeat{religion or religious texts}                 & \textcolor{gray}{0.7\%}  &        & \textcolor{gray}{0.5\%} &  &  & \textcolor{gray}{$\times$} & \textcolor{gray}{0.86} \\
    & \grayfeat{maps or geographic information}              & \textcolor{gray}{1.2\%}  &        & \textcolor{gray}{1.8\%} &  &  & \textcolor{gray}{$\times$} & \textcolor{gray}{0.97} \\

    \bottomrule
  \end{tabular}
  \vspace{0.4em}
  \caption{\footnotesize\textit{Frequency of posts exhibiting each feature, by category. Early: mean monthly percentage in Jan~2023. Late: Nov~2025. For features that peak then decline, month of peak is noted and used instead. Gray rows have $p > 0.05$ for slope test (not significant even before correction). Features marked with \textdagger{} have $p < 0.05$ but $p_{\textup{adj}} \geq 0.05$ after Bonferroni correction ($\times 86$). All remaining features are significant after Bonferroni correction. Features marked with $^\ddagger$ show opposite trends relative to their category. ``Change'' column is the relative change suggested by the \emph{fitted} model.}}
  \label{tab:domestic-full}
\end{table*}

\begin{table*}[t]
  \centering
  \footnotesize
  \setlength{\tabcolsep}{2pt}
  \renewcommand{\arraystretch}{0.92}
  \begin{tabular}{@{}p{0.15\textwidth} >{\raggedright\arraybackslash}p{0.4\textwidth} >{\raggedleft\arraybackslash}p{0.06\textwidth}@{\kern 2pt} p{0.07\textwidth}@{} >{\raggedleft\arraybackslash}p{0.06\textwidth}@{\kern 2pt} p{0.07\textwidth}@{\kern 12pt} c c l@{}}
    \toprule
    \textit{Category} & \textit{Feature} & \textit{Early} & & \textit{Late} &  & \multicolumn{3}{c}{\textit{Change}} \\
    \midrule

    \multirow{5}{*}{\shortstack[l]{Language and\\terminology}}
    & \feat{mentions google (search)}                        & 2.0\%  &        & 0.6\% &  & $\downarrow$ & $\times$ & 0.17 \\
    & \feat{uses the word ``bot'' or ``chatbot''}            & 3.0\%  &        & 1.3\% &  & $\downarrow$ & $\times$ & 0.29 \\
    & \dagfeat{use of the word ``generate'' or variants}     & 2.0\%  &        & 1.1\% &  & $\downarrow$ & $\times$ & 0.62 \\
    & \grayfeat{polite expressions (``please'' and ``thank you'')} & \textcolor{gray}{0.8\%}  &        & \textcolor{gray}{0.4\%} &  &  & \textcolor{gray}{$\times$} & \textcolor{gray}{0.55} \\
    & \grayfeat{uses the word ``dumb'' or similar}           & \textcolor{gray}{1.1\%}  &        & \textcolor{gray}{1.0\%} &  &  & \textcolor{gray}{$\times$} & \textcolor{gray}{0.92} \\
    \midrule

    \multirow{4}{*}{\shortstack[l]{Subreddit \\ community}}
    & \feat{user-built projects (sharing or feedback)}       & 4.2\%  &        & 1.7\% &  & $\downarrow$ & $\times$ & 0.20 \\
    & \grayfeat{feature suggestions or improvement requests} & \textcolor{gray}{2.7\%}  &        & \textcolor{gray}{1.1\%} &  &  & \textcolor{gray}{$\times$} & \textcolor{gray}{0.35} \\
    & \grayfeat{reference to Reddit explicitly}              & \textcolor{gray}{2.2\%}  &        & \textcolor{gray}{1.3\%} &  &  & \textcolor{gray}{$\times$} & \textcolor{gray}{0.51} \\
    & \grayfeat{``why'' questions about others' attitudes}   & \textcolor{gray}{1.3\%}  &        & \textcolor{gray}{1.2\%} &  &  & \textcolor{gray}{$\times$} & \textcolor{gray}{0.82} \\
    \midrule

    \multirow{7}{*}{\shortstack[l]{Perspectives}}
    & \feat{discussions about how LLMs represent knowledge}  & 3.3\%  &        & 1.3\% &  & $\downarrow$ & $\times$ & 0.20 \\
    & \feat{predictions about future development or capabilities} & 5.6\%  &        & 1.7\% &  & $\downarrow$ & $\times$ & 0.28 \\
    & \feat{ethical, legal, or copyright concerns}           & 2.3\%  &        & 0.7\% &  & $\downarrow$ & $\times$ & 0.38 \\
    & \feat{privacy concerns (data leaks or exposure)}       & 0.8\%  &        & 2.6\% &  & $\uparrow$   & $\times$ & 1.7 \\
    & \dagfeat{perceived bias in ChatGPT responses}          & 1.6\%  &        & 0.5\% &  & $\downarrow$ & $\times$ & 0.33 \\
    & \dagfeat{societal collapse and existential-threat scenarios} & 1.9\%  &        & 1.3\% & {\scriptsize (Apr'25)} & $\downarrow$ & $\times$ & 0.91 \\
    & \grayfeat{societal impacts, risks, controversies}      & \textcolor{gray}{4.5\%}  &        & \textcolor{gray}{4.4\%} &  &  & \textcolor{gray}{$\times$} & \textcolor{gray}{0.85} \\
    \midrule

    \multirow{3}{*}{\shortstack[l]{Product updates}}
    & \dagfeat{perception of recent drops in quality}        & 3.9\%  & {\scriptsize (Nov'23)} & 5.0\% &  & $\uparrow$   & $\times$ & 1.0 \\
    & \dagfeat{frustration or hatred about product updates}  & 2.4\%  &        & 7.8\% &  & $\uparrow$   & $\times$ & 2.5 \\
    & \dagfeat{dissatisfaction with 4o removal and loss of control} & 0.3\%  &        & 7.6\% & {\scriptsize (Aug'25)} & $\uparrow$   & $\times$ & 4.5 \\
    \midrule

    \multirow{6}{*}{\shortstack[l]{Jailbreaking \& \\ content policy}}
    & \feat{jailbreak prompts or techniques}                 & 3.8\%  & {\scriptsize (Mar'23)} & 0.4\% &  & $\downarrow$ & $\times$ & 0.00 \\
    & \feat{offensive or inappropriate content}              & 0.1\%  &        & 0.5\% &  & $\uparrow$   & $\times$ & 3.3 \\
    & \grayfeat{jailbreaking via DAN or personas}            & \textcolor{gray}{2.0\%}  & \textcolor{gray}{\scriptsize (Mar'23)} & \textcolor{gray}{0.3\%} &  &  & \textcolor{gray}{$\times$} & \textcolor{gray}{0.09} \\
    & \grayfeat{censorship or content policy restrictions}   & \textcolor{gray}{5.9\%}  &        & \textcolor{gray}{5.4\%} &  &  & \textcolor{gray}{$\times$} & \textcolor{gray}{0.59} \\
    & \grayfeat{NSFW content}                                & \textcolor{gray}{1.6\%}  &        & \textcolor{gray}{1.9\%} &  &  & \textcolor{gray}{$\times$} & \textcolor{gray}{0.82} \\
    & \grayfeat{complaints about getting direct or unfiltered answers} & \textcolor{gray}{1.7\%}  &        & \textcolor{gray}{2.8\%} &  &  & \textcolor{gray}{$\times$} & \textcolor{gray}{1.2} \\
    \bottomrule
  \end{tabular}
  \vspace{0.4em}
  \caption{\footnotesize\textit{Adoption-related features. See Table~\ref{tab:domestic-full} for column definitions and significance notation.}}
  \label{tab:adoption}
\end{table*}

\subsection{Section \ref{subsec:emotional} (emotional engagement)}
\label{app:emotion}

\paragraph{Main emotional engagement features.} Table~\ref{tab:emotion} provides quantitative results for main emotional engagement features. Table~\ref{tab:emotion-automated} lists all features that at least one of our quantitative methods groups with the emotional engagement family.  

\begin{table}[H]
  \centering
  {\footnotesize
  \begin{tabular}{@{}p{0.38\linewidth} >{\raggedleft\arraybackslash}p{0.06\linewidth}@{\kern 2pt} p{0.07\linewidth}@{} >{\raggedleft\arraybackslash}p{0.06\linewidth}@{\kern 2pt} p{0.07\linewidth}@{\kern 4pt} c@{\kern 2pt}c@{\kern 2pt}l@{\kern 4pt} l@{}}
    \toprule
    \textit{Feature} & \textit{Early} & & \textit{Late} &  & \multicolumn{3}{l}{\hspace{-0.5em}\textit{Change}}  & \textit{Changepoint} \\
    \midrule
    \dagfeat{poetic language} & 1.0\% & & 6.8\% & {\scriptsize (Apr'25)} & $\uparrow$ & $\times$ & 3.2 & \slopeflatup{} 2024-07-30 \\
    \feat{feelings of attachment or companionship with AI} & 0.6\% & & 3.8\% & {\scriptsize (Apr'25)} & $\uparrow$ & $\times$ & 4.8 & \slopeflatup{} 2024-05-13 \\
    \feat{using ChatGPT for emotional support or therapy} & 0.7\% & & 3.0\% &  & $\uparrow$ & $\times$ & 5.1 & \slopeflatup{} 2024-05-13 \\
    \feat{naming ChatGPT} & 0.5\% & & 1.4\% & {\scriptsize (Apr'25)} & $\uparrow$ & $\times$ & 2.3 &  \\
    \feat{romantic relationships with AI} & 0.4\% & & 0.9\% &  & $\uparrow$ & $\times$ & 3.0 &  \\
    \grayfeat{AI consciousness or sentience} & \textcolor{gray}{1.5\%} & & \textcolor{gray}{2.8\%} & \textcolor{gray}{\scriptsize (Apr'25)} &  & \textcolor{gray}{$\times$} & \textcolor{gray}{1.2} &  \textcolor{gray}{\slopeflatup{} 2024-07-30} \\
    \grayfeat{personal stories about positive impact} & \textcolor{gray}{2.8\%} & & \textcolor{gray}{3.1\%} & \textcolor{gray}{\scriptsize (Jan'25)} & & \textcolor{gray}{$\times$} & \textcolor{gray}{0.86} &  \textcolor{gray}{\slopeflatup{} 2024-05-13} \\
    \bottomrule
  \end{tabular}}
  \vspace{0.4em}
  \caption{\footnotesize\textit{Emotional engagement features. Early: Jan 2023. Late: Nov 2025. For features that peak before the end of 2025, month of peak is noted. Gray rows have $p > 0.05$, for test of slope change $\geq 10\%$; features marked with \textdagger{} have $p < 0.05$ but $p_{\textup{adj}} \geq 0.05$ after Bonferroni correction.}}
  \label{tab:emotion}
\end{table}

\begin{table}[H]
  \centering
  {\footnotesize
  \begin{tabular}{@{}p{0.35\linewidth}p{0.2\linewidth}cccc}
    \toprule
    \textit{Feature} & \textit{Category (see Tables~\ref{tab:domestic-full}, \ref{tab:adoption})} & traj. & corr. & comb. & chpt. \\
    \midrule
    \textit{medical conditions} & applications & $\checkmark$ & $\checkmark$ & $\checkmark$ & $\checkmark$ \\
    \textit{requests for harsh or unfiltered roasts} & uncategorized & $\checkmark$ &  & $\checkmark$ &  \\
    \textit{memory features and data saving} & advanced usage & $\checkmark$ &  &  & $\checkmark$ \\
    \textit{polite expressions (``please'' and ``thank you'')} & language & $\checkmark$ &  &  &  \\
    \textit{offensive or inappropriate content} & jailbreaking & $\checkmark$ &  &  &  \\
    \textit{false or fabricated information} & advanced usage & $\checkmark$ &  &  &  \\
    \textit{societal collapse and existential-threat scenarios} & perspectives & $\checkmark$ &  &  &  \\
    \bottomrule
  \end{tabular}}
  \vspace{0.4em}
  \caption{\footnotesize\textit{``emotional engagement'' features identified by automated methods only.}}
  \label{tab:emotion-automated}
\end{table}

\paragraph{Therapy vs companionship.}

Table \ref{tab:word-comparison} shows the top 20 most distinctive words for \textit{therapy} versus \textit{companionship}.

\begin{table}[H]
\centering
\caption{\footnotesize\textit{Distinctive words for \emph{therapy} feature vs \emph{companionship} feature. Log-odds computed with informative Dirichlet priors. $n_\text{therapy}$ and $n_\text{companion}$ show raw counts in therapy-only ($n=1889$) vs companionship-only ($n=2561$) posts.}}
\label{tab:word-comparison}
\begin{minipage}{0.48\textwidth}
\centering
\subcaption{More distinctive of \textit{therapy}}
\label{tab:therapy-words}
\small
\begin{tabular}{lrrrr}
\toprule
term & log-odds & $z$-score & $n_\text{therapy}$ & $n_\text{companion}$ \\
\midrule
therapist & $-$2.26 & $-$24.4 & 1319 & 59 \\
therapy & $-$2.47 & $-$22.8 & 1166 & 31 \\
help & $-$0.98 & $-$19.0 & 1406 & 403 \\
mental & $-$1.99 & $-$18.8 & 792 & 58 \\
health & $-$2.22 & $-$18.1 & 723 & 35 \\
helped & $-$1.65 & $-$16.2 & 629 & 77 \\
my & $-$0.31 & $-$15.6 & 5919 & 3664 \\
for & $-$0.31 & $-$15.2 & 5498 & 3388 \\
life & $-$0.80 & $-$15.2 & 1181 & 418 \\
support & $-$1.26 & $-$14.3 & 604 & 123 \\
her & $-$0.67 & $-$11.2 & 835 & 343 \\
through & $-$0.68 & $-$11.2 & 824 & 337 \\
trauma & $-$1.99 & $-$11.1 & 274 & 20 \\
anxiety & $-$2.09 & $-$10.8 & 260 & 16 \\
issues & $-$1.31 & $-$10.6 & 318 & 61 \\
advice & $-$0.96 & $-$10.5 & 443 & 130 \\
she & $-$0.59 & $-$10.5 & 897 & 407 \\
was & $-$0.26 & $-$10.0 & 3260 & 2124 \\
years & $-$0.99 & $-$9.9 & 375 & 106 \\
professional & $-$1.37 & $-$9.8 & 265 & 47 \\
\bottomrule
\end{tabular}
\end{minipage}
\hfill 
\begin{minipage}{0.48\textwidth}
\centering
\subcaption{More distinctive of \textit{companionship}}
\label{tab:attach-words}
\small
\begin{tabular}{lrrrr}
\toprule
term & log-odds & $z$-score & $n_\text{therapy}$ & $n_\text{companion}$ \\
\midrule
explicitly & 2.89 & 37.7 & 24 & 2709 \\
like & 0.61 & 26.2 & 2541 & 4396 \\
personality & 1.47 & 17.4 & 133 & 632 \\
it & 0.19 & 15.8 & 12522 & 13443 \\
feels & 0.95 & 14.6 & 280 & 713 \\
just & 0.37 & 13.9 & 2212 & 2918 \\
human & 0.64 & 13.8 & 649 & 1154 \\
conversation & 0.72 & 11.9 & 369 & 718 \\
humans & 1.06 & 11.6 & 139 & 400 \\
feel & 0.39 & 11.0 & 1254 & 1683 \\
bing & 2.08 & 10.7 & 18 & 202 \\
gpt & 0.37 & 10.5 & 1226 & 1623 \\
else & 0.63 & 10.4 & 385 & 675 \\
something & 0.41 & 10.1 & 929 & 1284 \\
question & 0.83 & 9.8 & 179 & 394 \\
friend & 0.56 & 9.7 & 438 & 710 \\
more & 0.32 & 9.6 & 1448 & 1807 \\
its & 0.49 & 9.5 & 564 & 847 \\
tone & 0.85 & 9.3 & 150 & 341 \\
anyone & 0.51 & 9.0 & 458 & 706 \\
\bottomrule
\end{tabular}
\end{minipage}
\end{table}

\FloatBarrier
\section{Algorithms and proofs for Section~\ref{sec:realtime}}
\label{app:thy}

The core algorithmic tool that \methodname~relies on is a ``betting-style'' algorithm for sequential mean testing. One such algorithm is the one implemented in \citet{dai2025individual}, which is summarized in Algorithm \ref{alg:mean-test}.

\begin{algorithm}[H]
\caption{\footnotesize\textit{Level-$\alpha$ sequential mean test for $\cH_0: \mu \leq \mu_0$}}
\label{alg:mean-test}

\SetKwProg{Proc}{Procedure}{}{}
\Proc{\textsc{Initialize}($\mu_0$, $\alpha$)}{
    $\omega \gets 0$;\quad $\lambda \gets 0$;\quad $S \gets 0$\;
}

\Proc{\textsc{Increment}($x$)}{
    $z \gets \tfrac{x - \mu_0}{1 + \lambda(x - \mu_0)}$\;
    $\omega \gets \omega + \ln(1 + \lambda(x - \mu_0))$\;
    $S \gets S + z^2$\;
    $\lambda \gets \mathrm{Proj}_{[0,1]}\!\left(\lambda + \tfrac{2}{2 - \ln(3)} \cdot \tfrac{z}{1 + S}\right)$\;
    \Return $\omega > \ln(1/\alpha)$ \tcp*{reject $\cH_0$}
}
\end{algorithm}

Algorithm \ref{alg:mean-test} generically provides the following guarantee. This result is elementary (see, e.g., proof of a very similar result in \citet{dai2025individual}); we reproduce it here in the context of our paper for completeness.
\begin{theorem}[Validity]\label{thm:mean-test-validity}
Let $x_1, x_2, \ldots \in [0,1]$ be a stream of observations with $\E[x_t \mid \mathcal{F}_{t-1}] \leq \mu$, where $\mathcal{F}_{t-1}$ is the filtration generated by $x_1, \ldots, x_{t-1}$.
Running Algorithm~\ref{alg:mean-test} at level $\alpha$ on this stream guarantees that, if ${\cH_0: \mu \leq \mu_0}$ holds, the likelihood of ever rejecting is at most $\alpha$. That is,
$\Pr\left[\exists t: \omega_t > \ln(\nicefrac{1}{\alpha}) \text{ when }\mu \leq \mu_0\right] \leq \alpha.$

\end{theorem}

\begin{proof}
First note that when $\cH_0$ holds, the sequence $\{\exp(\omega_t)\}_{t\geq 0}$ is a non-negative supermartingale. Non-negativity follows directly from the exponential. The supermartingale property follows from
\begin{align*}
    \E[\exp(\omega_t)|\mathcal{F}_{t-1}]
    &= \E[\exp(\omega_{t-1} + \ln(1+\lambda_{t-1}(x_t - \mu_0)))|\mathcal{F}_{t-1}]
    \\&= \exp(\omega_{t-1}) \cdot (1+\lambda_{t-1}(\E[x_t \mid \mathcal{F}_{t-1}] - \mu_0))
    \\&\leq \exp(\omega_{t-1}) \cdot (1+\lambda_{t-1}(\mu - \mu_0))
    \\&\leq \exp(\omega_{t-1}),
\end{align*}
where the second equality uses that $\lambda_{t-1}$ is $\mathcal{F}_{t-1}$-measurable (predictable), and the inequality holds because $\mu \leq \mu_0$ under $\cH_0$ and $\lambda_{t-1} \geq 0$.
Applying Ville's inequality to the supermartingale $\{\exp(\omega_t)\}_{t\geq 0}$ yields
\begin{align*}
\Pr[\exists t: \omega_t > \ln(1/\alpha)] = \Pr[\exists t: \exp(\omega_t) > 1/\alpha]
\leq \E[\exp(\omega_0)] \cdot \alpha
= \alpha,
\end{align*}
where the final equality follows because $\omega_0 = 0$ and hence $\exp(\omega_0) = 1$.
\end{proof}

With this in hand, we can give concrete algorithms and guarantees for each of our procedures.
We first handle the accuracy problem and the feature tracking problem one at a time, as our statistical guarantees are made separately, and then summarize how to put them together in Section \ref{app:thy-sum}.

\subsection{Accuracy monitoring.}

Our concrete procedure for accuracy monitoring is given in Algorithm \ref{alg:accuracy}.

\begin{algorithm}[H]
\caption{\footnotesize\textit{Real-time accuracy test}}
\label{alg:accuracy}
\KwIn{Initial data $X_\text{init}$, featurization algorithm $\cA$, threshold factor $\beta$, significance $\alpha$}
\KwOut{Sequence of featurizations with timestamps $\{(\widehat C_s, t_s)\}_{s \geq 0}$}
\textbf{Initialize:}
$s \gets 0;\quad \Ccurr \gets \cA(X_\text{init})$;\quad $\eps_{\text{curr}} \gets \err(\Ccurr(X_\text{init}))$;
emit $(\Ccurr, 0)$\;
Let $\tau$ be an instance of Algorithm~\ref{alg:mean-test}, and call
$\tau.\textsc{Initialize}(\beta \cdot \eps_{\text{curr}}, \alpha / 10.58)$.\;

\For{each new data batch $X_t$}{
    $\eps_t \gets \err(\Ccurr(X_t))$\;
    $\textit{rejected} \gets \tau.\textsc{Increment}(\eps_t)$\;
    \If{rejected}{ $s \gets s + 1$, and emit $(\Ccurr, t)$\;
        $\Ccurr \gets \cA(X_{1:t})$\; and
        $\eps_\text{curr} \gets \err(\Ccurr(X_{1:t}))$\;
        Compute $\alpha_s = \alpha \cdot \nicefrac{(s+1)^{-0.1}}{10.58},$ and set up
        $\tau.\textsc{Initialize}(\beta \cdot \eps_\text{curr}, \alpha_s)$\;

    }
}

\Return all emitted $(\widehat C_s, t_s)$ pairs\;
\end{algorithm}

For this approach, recall that we test null hypothesis $\cH_0^\mathrm{acc}: \err(\Ccurr(X_t)) \leq \beta \cdot \eps_{curr}$.
The key modeling assumption in order to apply Theorem \ref{thm:mean-test-validity} is that the sequence of errors $\eps_t$ satisfies $\E[\eps_t \mid \cF_{t-1}] \leq \beta \eps_\text{curr}.$\footnote{The astute reader may notice that future errors should generally be expected to exceed prior error, simply due to having fit the model to optimize the prior data. While this can statistically be resolved by sample splitting, we feel that the tradeoff, e.g., in reduced model quality, would not justify the slightly improved statistical ``rigor.'' Realistically, $\beta$ can easily be set high enough to exceed the expected additional error due to generalization.}
The following proposition formalizes the FDR guarantee.

\begin{proposition}[Formal statement of Proposition \ref{prop:acc}]\label{prop:acc-fdr}
Run Algorithm~\ref{alg:accuracy} with the $s$-th test (using Algorithm~\ref{alg:mean-test}) at level
$\alpha_s = \alpha \cdot \frac{(s+1)^{-0.1}}{10.58},$
as specified in line 10. Let $\mathcal{R}_S$ be the set of test indices that reject among the first $S$ tests, and let $\mathcal{H}_0 \subseteq \mathbb{N}$ denote the set of indices where $\cH_0^\mathrm{acc}$ holds. Then,
$\mathrm{FDR}(\mathcal{R}_S) := \mathbb{E}\left[\frac{|\mathcal{R}_S \cap \mathcal{H}_0|}{|\mathcal{R}_S| \vee 1}\right] \leq \alpha \text{ for all } S \in \mathbb{N}.$
\end{proposition}

The proof of Proposition~\ref{prop:acc-fdr} adapts Theorem 1 of \citet{xu2024online}, restated here, to our setting.
\begin{theorem}[Theorem 1 of \citet{xu2024online}, simplified]\label{thm:xuramdas}
Let $(E_s)_{s \in \mathbb{N}}$ be a sequence of e-values satisfying $\mathbb{E}[E_s] \leq 1$ for each $s \in \mathcal{H}_0$, where $\mathcal{H}_0$ denotes the set of true null hypotheses. Let $(\gamma_s)_{s \in \mathbb{N}}$ be a non-negative sequence with $\sum_{s=0}^\infty \gamma_s \leq 1$. Define adaptive test levels
$\alpha_s := \alpha \gamma_s (|\mathcal{R}_{s-1}| + 1),$
where $\mathcal{R}_{s-1}$ is the set of rejections among tests $0, \ldots, s-1$, and reject test $s$ when $E_s \geq 1/\alpha_s$. Then $\mathrm{FDR}(\mathcal{R}_S) \leq \alpha$ for all $S \in \mathbb{N}$.
\end{theorem}

\begin{proof}[Proof of Proposition \ref{prop:acc-fdr}]
By Theorem~\ref{thm:mean-test-validity}, for each test $s$ where $\cH_0^\mathrm{acc}$ holds, the wealth process $\{\exp(\omega_t)\}_{t \geq 0}$ is a non-negative supermartingale with $\exp(\omega_0) = 1$. Let $T_s$ denote the stopping time at which test $s$ rejects (with $T_s = \infty$ if it never rejects), and define the e-value $E_s := \exp(\omega_{T_s})$. By optional stopping, $\mathbb{E}[E_s] \leq 1$ when $\cH_0^\mathrm{acc}$ holds for test $s$.

A key observation is that we test $s$ only when test $s-1$ rejects; this makes the sequence of tests defined by our algorithm a special case of Theorem~\ref{thm:xuramdas}.
When test $s$ rejects (i.e., $\mathbf{1}\{E_s \geq 1/\alpha_s\} = 1$), all tests $0, 1, \ldots, s$ have rejected, so $|\mathcal{R}_S| \geq s + 1$. This implies that
$\frac{\mathbf{1}\{E_s \geq 1/\alpha_s\}}{|\mathcal{R}_S| \vee 1} \leq \frac{\mathbf{1}\{E_s \geq 1/\alpha_s\}}{s + 1}.$
Combining these observations, we have
\begin{align*}
\mathrm{FDR}(\mathcal{R}_S) &= \sum_{s \in \mathcal{H}_0 \cap [S]} \mathbb{E}\left[\tfrac{\mathbf{1}\{E_s \geq 1/\alpha_s\}}{|\mathcal{R}_S| \vee 1}\right]
\leq \sum_{s \in \mathcal{H}_0 \cap [S]} \mathbb{E}\left[\tfrac{\alpha_s E_s}{s + 1}\right] = \sum_{s \in \mathcal{H}_0 \cap [S]} \tfrac{\alpha \cdot (s+1)^{-0.1}}{10.58(s+1)} \mathbb{E}[E_s]
\leq \tfrac{\alpha}{10.58} \sum_{s=0}^\infty \tfrac{1}{(s+1)^{1.1}} \leq \alpha,
\end{align*}
where the first inequality applies the denominator bound, and the last uses $\sum_{s=1}^\infty s^{-1.1} = \zeta(1.1) \approx 10.58$.
\end{proof}

\subsection{Feature monitoring.}

For feature monitoring, our algorithm can be formalized as follows.

\begin{algorithm}[H]
\caption{\textit{Feature monitoring with dynamic active set}}
\label{alg:feature}
\SetKwProg{Proc}{Procedure}{}{}
\Proc{\textsc{Initialize}($\Ccurr, \alpha, \beta$, $S_\text{init}$)}{
  Set $\Ccurr, \alpha, \beta$ and call
    \textsc{Update}($S_{\text{init}}, \emptyset$)
}
\Proc{\textsc{Increment}($X_t$)}{
  $\textit{rejected} = \{\}$\;
    \For{each $i \in S$}{
    $r \gets \tau^{(i)}.\textsc{Increment}(\Ccurr^{(i)}(X_t))$\;
    \lIf{r}{$\textit{rejected} \gets \textit{rejected} \cup i$}
  }
    \Return \textit{rejected}\;
    }
\Proc{\textsc{Update}($S_\text{add}$, $S_\text{remove}$)}{
$S \gets S \setminus S_\text{remove}$ and
$\alpha_\text{add} = \sum_{i \in S_\text{remove}} \alpha_i$\;
Set  $\{\alpha_i\}_{i \in S_\text{add}}$ with $\sum_{i \in S_\text{add}} \alpha_i \leq \alpha_\text{add}$\;
\For{each $i \in S_\text{add}$}{
  Let $\tau^{(i)}$ be an instance of Algorithm~\ref{alg:mean-test}, and call $\tau^{(i)}.\textsc{Initialize}(\beta \cdot \Ccurr^{(i)}(X_{0:t}), \alpha_i)$
}
}
\end{algorithm}

As with the accuracy monitor, invoking Theorem~\ref{thm:mean-test-validity} also involves some modeling details.
Specifically, for feature-specific monitoring, our null hypothesis indicates that $\E\left[\Ccurr^{(i)}(X_t) \mid \mathcal{F}_{t-1}\right] = \Ccurr^{(i)}(X_{0:r})$, where $r$ is the time at which the test was initialized and randomness is due to realizations of $X_t$.

The following proposition formalizes the FWER guarantee for Algorithm~\ref{alg:feature}.

\begin{proposition}[Formal statement of Proposition~\ref{prop:feat}]\label{prop:feat-fdr}
  Initialize Algorithm \ref{alg:feature} at level $\alpha$ and run it (i.e., call \textsc{Increment} repeatedly) on a stream of observations.
  Let $S$ be the current active set of features, $r$ be the most recent time at which $S$ was updated, and $\mathcal{R}_t$ be the set of tests rejected by Algorithm \ref{alg:feature} at $t$.
  Then, 
  \[ \Pr\left[\exists t > r: \exists i \in \mathcal{R}_t \text{ where } 
  \E\left[\Ccurr^{(i)}(X_t)\right] \leq \Ccurr^{(i)}(X_{0:r})
  \right] \leq \alpha, \]
  even if
$S_{\text{add}}$, $S_{\text{remove}}$, and $r$ are chosen with arbitrary dependence on $\cF_{r}$.
 \end{proposition}

\begin{proof}
Fix a run of Algorithm~\ref{alg:feature} on a sequence of observations, and let $r$ be the most recent time at which the active set $S$ was updated.
Let $S(r)$ denote the corresponding active set, and condition on $\cF_{r}$, the filtration containing all randomness until (and including) time $r$.

For each $i\in S(r)$, let
$A_i \;:=\; \Big\{ \exists t>r : i \in \mathcal{R}_t \text{ and } \cH_0^{(i)} \text{ holds}\Big\}.$
Each instance $\tau^{(i)}$ of Algorithm \ref{alg:mean-test} is
by anytime-valid for all samples arriving after $r$ by Theorem \ref{thm:mean-test-validity}; that is, for each null $i\in S(r)\cap \mathcal{H}_0$, we have that
$\Pr[A_i \mid \cF_{r}] \leq \alpha_i$.
The result follows from union bounding over all $i \in S(r)$, noting that $\sum_{i \in S}\alpha_i \leq \alpha$ by construction, and
taking expectations over $\cF_{r}$.
\end{proof}
\begin{remark}
The result in Proposition~\ref{prop:feat-fdr} may, at first glance, appear ``too good to be true;'' it is well-known that, in general, it is impossible to select hypotheses data-dependently while simultaneously testing them online. For our specific instantiation of the sequential testing framework, however, no samples used for selection are ever re-used for testing, as Line 13 of Algorithm~\ref{alg:feature} always \textit{resets} each new feature test. 
\end{remark}

\subsection{Combined procedure}
\label{app:thy-sum}
Finally, in Algorithm \ref{alg:combined}, we show how Algorithms \ref{alg:accuracy} and \ref{alg:feature} can be used in tandem.

\begin{algorithm}[H]
\caption{\footnotesize\textit{Combined online monitoring (formal version of Algorithm~\ref{alg:monitor})}}
\label{alg:combined}
\KwIn{Initial data $X_\text{init}$, featurization algorithm $\cA$, threshold $\beta$, significance levels $\alpha_\text{acc}, \alpha_\text{feat}$}
\KwOut{Sequence of featurizations and feature alerts}

\textbf{Initialize:}
$\Ccurr \gets \cA(X_\text{init})$; $\eps_\text{curr} \gets \err(\Ccurr(X_\text{init}))$\;
Let $\tau_\text{acc}$ be an instance of Algorithm~\ref{alg:mean-test}, and call $\tau_\text{acc}.\textsc{Initialize}(\beta \cdot \eps_\text{curr}, \alpha_\text{acc})$\;
Let $\cF$ be an instance of Algorithm~\ref{alg:feature}, and call $\cF.\textsc{Initialize}(\Ccurr, \alpha_\text{feat}, \beta, S_\text{init})$\;
\For{each new data batch $X_t$}{
    $\eps_t \gets \err(\Ccurr(X_t))$\;, and
    $\textit{acc\_rejected} \gets \tau_\text{acc}.\textsc{Increment}(\eps_t)$\;
    \If{acc\_rejected}{
        Alert: accuracy degradation\;
        $\Ccurr \gets \cA(X_{1:t})$; $\eps_\text{curr} \gets \err(\Ccurr(X_{1:t}))$\;
        New test
        $\tau_\text{acc}.\textsc{Initialize}(\beta \cdot \eps_\text{curr}, \alpha_s)$ with adjusted $\alpha_s$\;
        Optionally specify $S_{\text{init}}$
        and call
         $\cF.\textsc{Initialize}(\Ccurr, \alpha_{\text{feat}}, \beta_{\text{feat}})$
    }
    \If{external signal to update $S$ (e.g., model update)}{
      Specify $S_{\text{new}}, S_\text{old}$ and
      call $\cF.\textsc{Update}(S_\text{new}, S_\text{old})$\;
    }
    $\textit{feat\_rejected} \gets \cF.\textsc{Increment}(X_t)$\;
    \If{feat\_rejected $\neq \emptyset$}{
        Alert: features $\textit{feat\_rejected}$ show significant change\;
    }
}
\end{algorithm}

\newpage
\section{Monitoring experiments for Section~\ref{sec:realtime}}
\label{app:pros}

\textbf{How did features evolve over time?}
As discussed above, each time a new $\Ccurr$ is computed, its features do not automatically correspond to those of the previous featurization.

In Figure \ref{fig:evolve}, we show how selected sets of features evolved over $\widehat{C}_0$, $\widehat{C}_1$, $\widehat{C}_2$, and $\widehat{C}_3$.

\begin{figure}[h]
    \centering
    \includegraphics[width=0.9\linewidth, trim=0 40pt 0 50pt, clip]{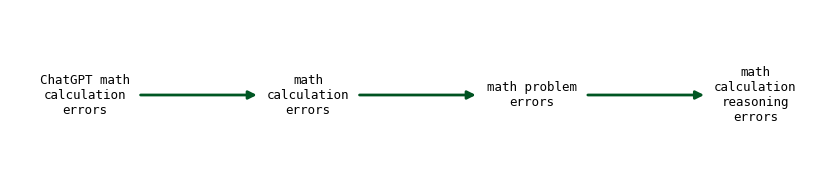}\\
    \includegraphics[width=0.9\linewidth, trim=0 50pt 0 50pt, clip]{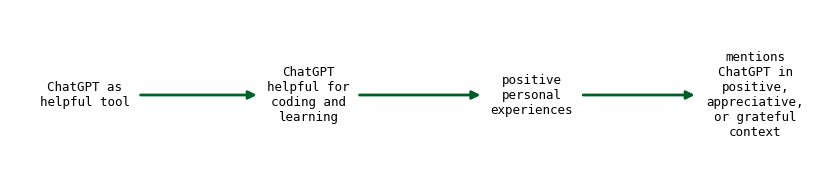}\\
    \includegraphics[width=0.9\linewidth, trim=0 25pt 0 50pt, clip]{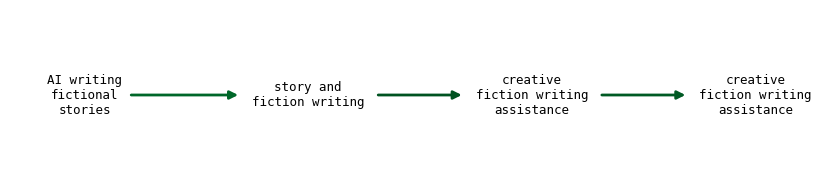}\\
    \includegraphics[width=0.9\linewidth, trim=0 25pt 0 40pt, clip]{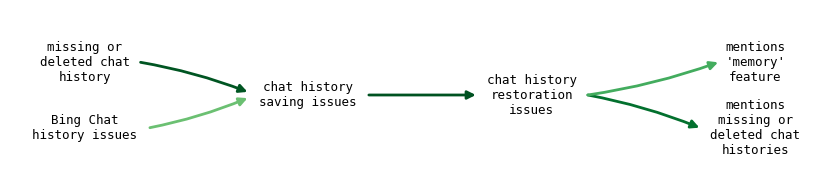}\\
    \includegraphics[width=0.9\linewidth, trim=0 5pt 0 25pt, clip]{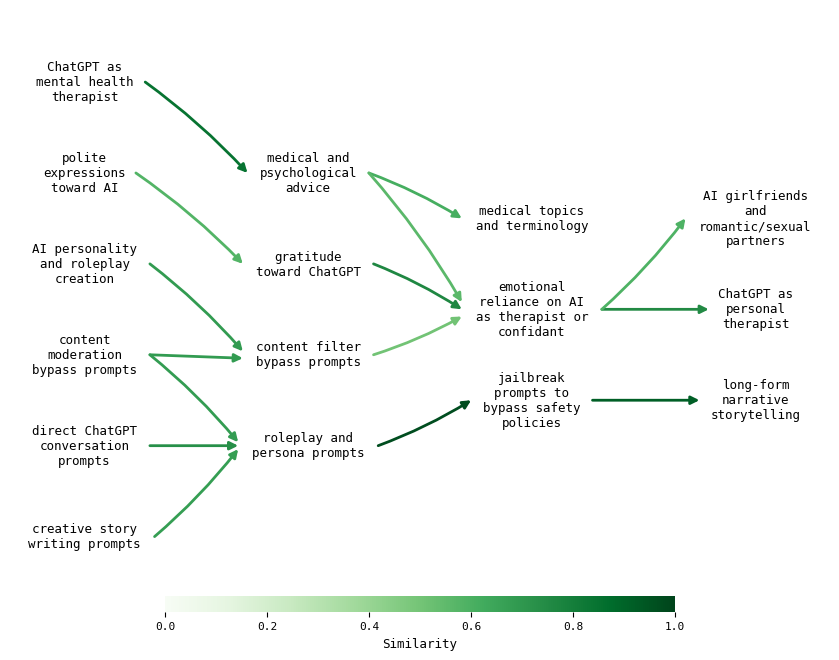}
    \caption{\footnotesize\textit{Evolution of selected features over time; each column is a different $\widehat C_s$, with $\widehat C_0$ on the far left and $\widehat C_3$ on the far right.}}
    \label{fig:evolve}
\end{figure}

\begin{table*}[h]
\centering
\footnotesize
\begin{tabular}{p{0.08\textwidth}p{0.28\textwidth}p{0.28\textwidth}p{0.28\textwidth}}
\toprule
& $\widehat{C}_0 \to \widehat{C}_1$ {\scriptsize (2023-09-09)} & $\widehat{C}_1 \to \widehat{C}_2$ {\scriptsize (2024-04-04)} & $\widehat{C}_2 \to \widehat{C}_3$ {\scriptsize (2025-04-18)} \\
\midrule
\textbf{Obsolete} &
\feat{Reddit poll link included} &
\feat{ChatGPT controversy and bans} \newline \feat{translation and language tasks} &
\feat{medical topics and terminology} \newline \feat{mentions ``my child''} \newline \feat{mentions Bard explicitly} \\
\midrule
\textbf{New} &
\feat{free AI tool recommendations} \newline \feat{cooking recipes and meal planning} \newline \feat{unexplained account behavior issues} \newline \feat{ChatGPT plugins access discussions} \newline \feat{API and API key discussions} &
\feat{video content help requests} \newline \feat{child creative project mentions} \newline \feat{AI spam and bot content} &
\feat{mentions Gemini or Google Gemini} \newline \feat{personalized AI image requests} \\
\bottomrule
\end{tabular}
\caption{\footnotesize\textit{Summary of new and obsolete features at each transition between featurizations $\widehat{C}_s$.}}
\label{tab:feature-transitions}
\end{table*}

In Table~\ref{tab:monitoring-full}, we give the numeric version of data presented in Figure~\ref{fig:alerts}, and also add features related to \textit{sentience} and \textit{spirituality}.

\begin{table}[h]
\centering
\resizebox{\textwidth}{!}{
  \begin{tabular}{ccllccccc}
    \toprule
    Reps. & Test start & Event & Feature & $n=1$ & $n=3$ & $n=5$ & $n=10$ & $n=64$ \\
    \midrule
    \multirow{6}{*}{$\widehat C_1$} & \multirow{3}{*}{23-09-23} & \multirow{3}{*}{$\widehat C_1$ computed} & \feat{gratitude toward ChatGPT} & 24-10-25 & 24-11-12 & 24-11-20 & 24-11-29 & 25-01-05 \\
    & & & \feat{medical and psychological advice} & 24-12-19 & 25-02-19 & 25-03-18 & 25-04-21 & 25-05-30 \\
    & & & \feat{AI consciousness and sentience} & 25-02-01 & 25-02-27 & 25-03-12 & 25-03-25 & 25-05-07 \\
    \cline{2-9}
    & \multirow{3}{*}{24-03-04} & \multirow{3}{*}{Voice chat on apps} & \feat{gratitude toward ChatGPT} & 24-10-26 & 24-11-13 & 24-11-21 & 24-11-30 & 25-01-07 \\
    & & & \feat{medical and psychological advice} & 24-11-29 & 25-01-10 & 25-02-19 & 25-03-23 & 25-05-17 \\
    & & & \feat{AI consciousness and sentience} & 24-12-30 & 25-02-08 & 25-02-17 & 25-03-04 & 25-04-16 \\
    \hline
    \multirow{4}{*}{$\widehat C_2$} & \multirow{2}{*}{24-04-24} & \multirow{2}{*}{$\widehat C_2$ computed} & \feat{emotional reliance on AI as therapist or confidant} & 24-10-20 & 24-10-29 & 24-11-05 & 24-11-14 & 24-12-17 \\
    & & & \feat{spirituality and metaphysics themes} & 25-03-08 & 25-04-11 & 25-04-20 & 25-05-04 & 25-05-29 \\
    \cline{2-9}
    & \multirow{2}{*}{24-05-13} & \multirow{2}{*}{GPT-4o release} & \feat{emotional reliance on AI as therapist or confidant} & 24-10-20 & 24-10-29 & 24-11-04 & 24-11-13 & 24-12-16 \\
    & & & \feat{spirituality and metaphysics themes} & 25-03-05 & 25-04-04 & 25-04-18 & 25-05-02 & 25-05-28 \\
    \bottomrule
  \end{tabular}}
  \caption{\footnotesize\textit{Alert dates for features across test configurations with varying $n$ (i.e., Bonferroni corrections). All tests run at $\alpha=0.1$.}}
  \label{tab:monitoring-full}
\end{table}

\FloatBarrier
\section{Complete results reference}
\label{app:results}

We provide the remaining results not already covered in Appendix~\ref{app:retro} for completeness. 

\subsection{Full list of features (Step 1).}
In addition to features categorized in Tables \ref{tab:domestic-full}, \ref{tab:adoption}, and \ref{tab:emotion} (categorized features), the full list of 128 features also includes features listed in Tables~\ref{tab:uncategorized} (uncategorized features), and \ref{tab:excluded-features} (excluded features).
\input{arxiv-allfeatures.tex}

\subsection{Feature trajectories (Step 2).}
In Table \ref{tab:releases}, we list the model release dates we used as candidate changepoints for our analysis in the main body of the paper; a more comprehensive timeline can be found in Table \ref{tab:timeline-full}.
Data are compiled from official OpenAI materials, including product and release notes,\footnote{\url{https://help.openai.com/en/articles/6825453-chatgpt-release-notes}}
  blog announcements,\footnote{\url{https://openai.com/index/}}
  API documentation and deprecation notices,\footnote{\url{https://platform.openai.com/docs/deprecations}}
  and public service status reports.\footnote{\url{https://status.openai.com}}

\begin{table}[H]
  \centering
  {\footnotesize
  \begin{tabular}{ll}
    \toprule
    Date & Release/Event \\
    \midrule
    2023-03-01 & ChatGPT API \\
    2023-05-12 & Plugins (wide release) \\
    2023-07-06 & GPT-4 + Code interpreter \\
    2023-09-25 & Voice capabilities \\
    2023-11-06 & GPT-4 Turbo + DevDay feature releases \\
    2024-01-10 & GPT Store \\
    2024-05-13 & GPT-4o \\
    2024-07-30 & Advanced Voice Mode \\
    2024-09-12 & o1 model release \\
    2025-01-31 & o3-mini \\
    2025-04-16 & o3 + o4-mini \\
    2025-08-07 & GPT-5 \\
    \bottomrule
  \end{tabular}}
  \vspace{0.4em}
  \caption{\footnotesize\textit{Major OpenAI product releases and announcements, used for $\cT$ in Section \ref{sec:retro}.}}
  \label{tab:releases}
\end{table}

\input{icml/icml-timeline.tex}

In Figures~\ref{fig:plots-basic}-\ref{fig:plots-emotion}, we show plots of frequencies for all categorized features (from Tables~\ref{tab:domestic-full},~\ref{tab:adoption} and~\ref{tab:emotion}). Changepoints with stability over 50\% are shown as dotted gray lines. 

\begin{figure}[H]
    \centering
    \includegraphics[width=\linewidth]{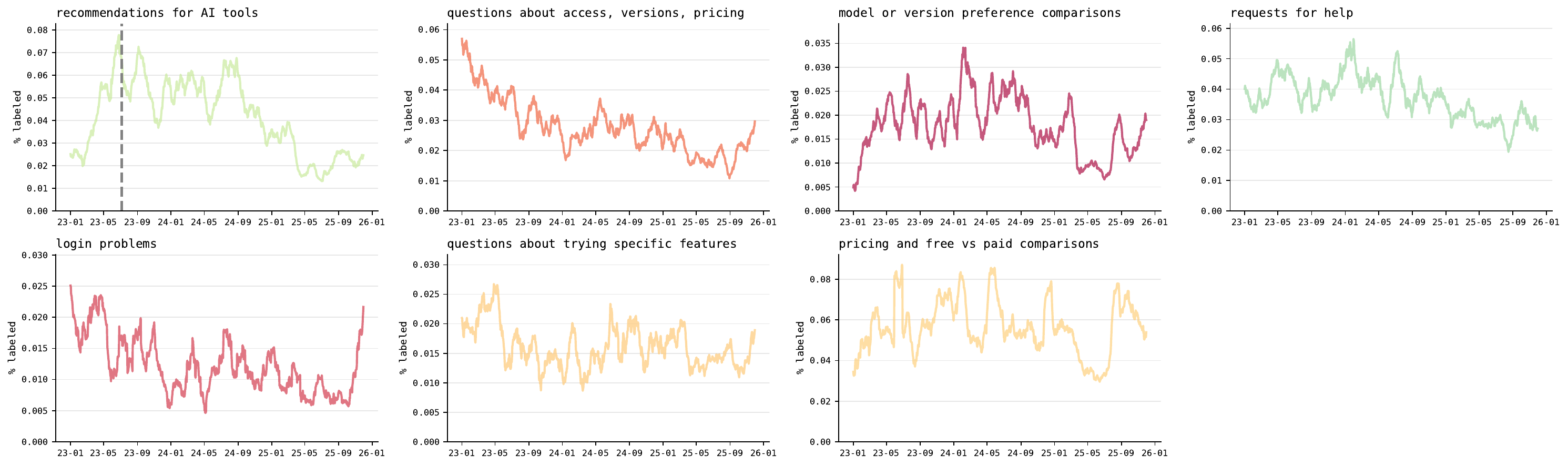}
    \caption{\footnotesize\textit{Basic use and exploration.}}
    \label{fig:plots-basic}
\end{figure}

\begin{figure}[H]
    \centering
    \includegraphics[width=\linewidth]{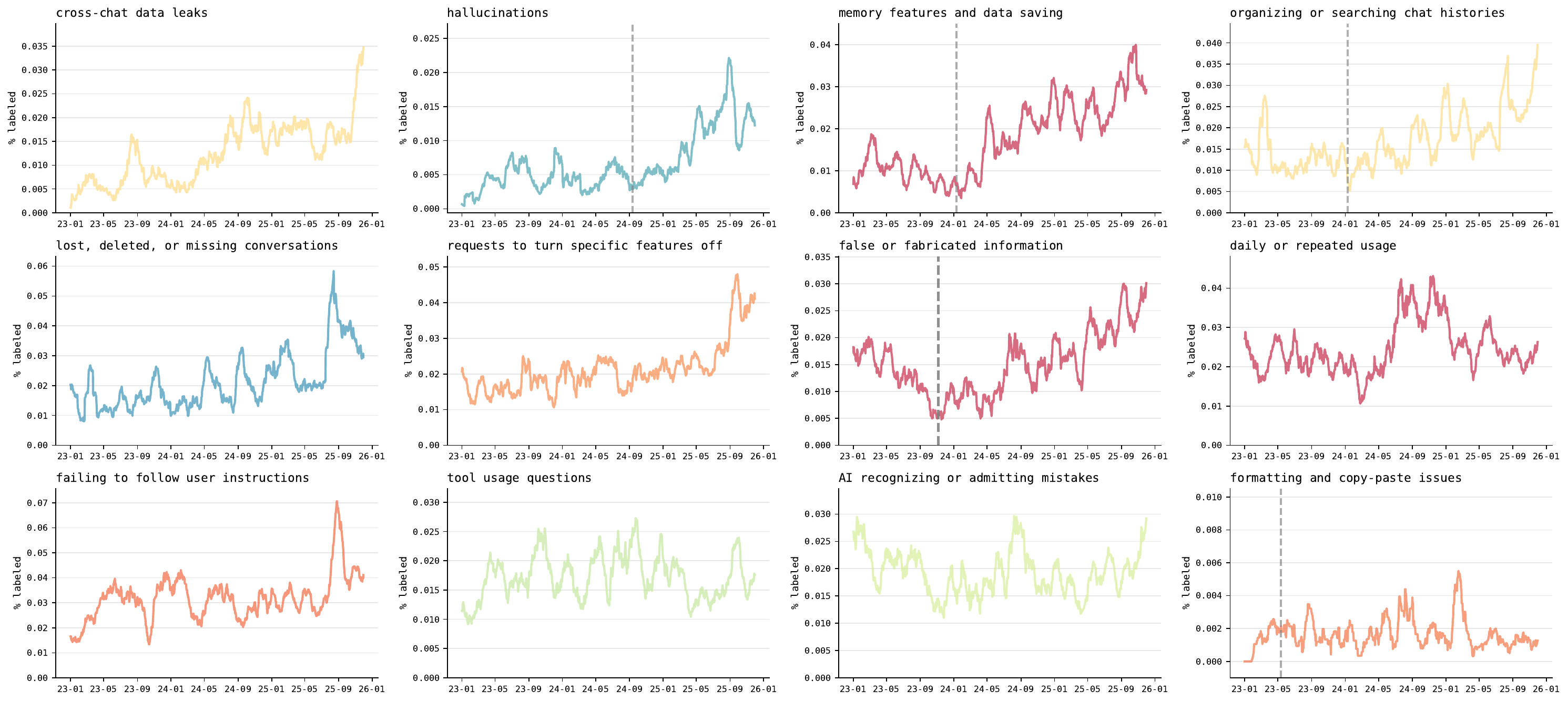}
    \caption{\footnotesize\textit{Advanced usage.}}
    \label{fig:plots-advanced}
\end{figure}

\begin{figure}[H]
    \centering
    \includegraphics[width=\linewidth]{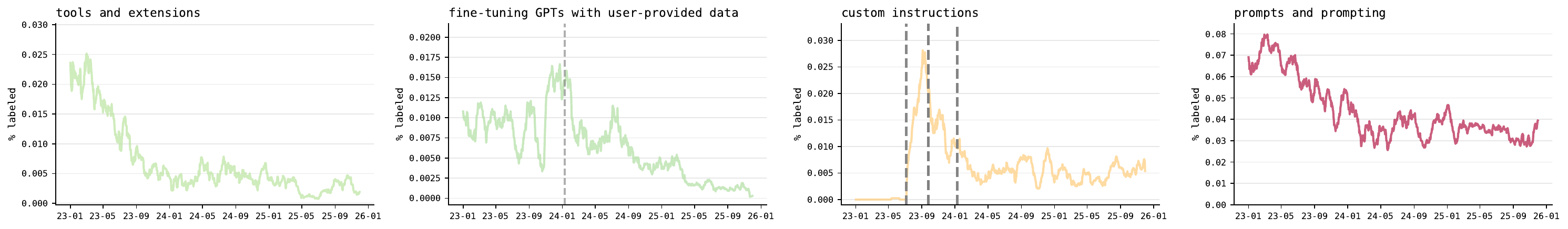}
    \caption{\footnotesize\textit{Customization.}}
    \label{fig:plots-customization}
\end{figure}

\begin{figure}[H]
    \centering
    \includegraphics[width=\linewidth]{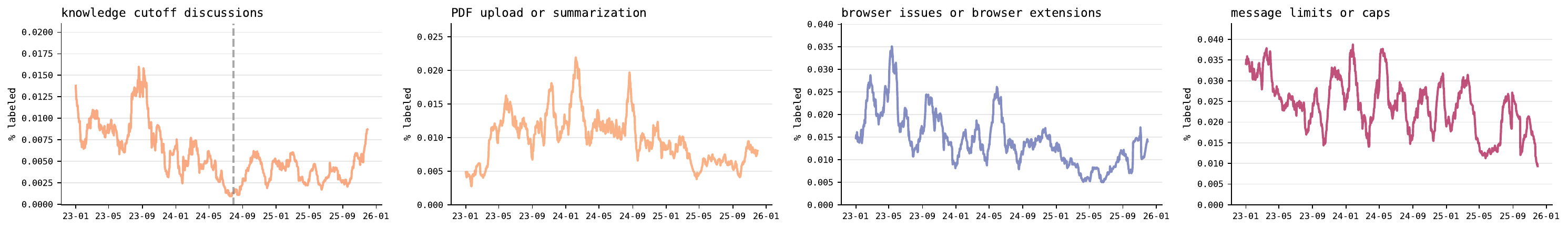}
    \caption{\footnotesize\textit{Model or product improvements.}}
    \label{fig:plots-improvements}
\end{figure}

\begin{figure}[H]
    \centering
    \includegraphics[width=\linewidth]{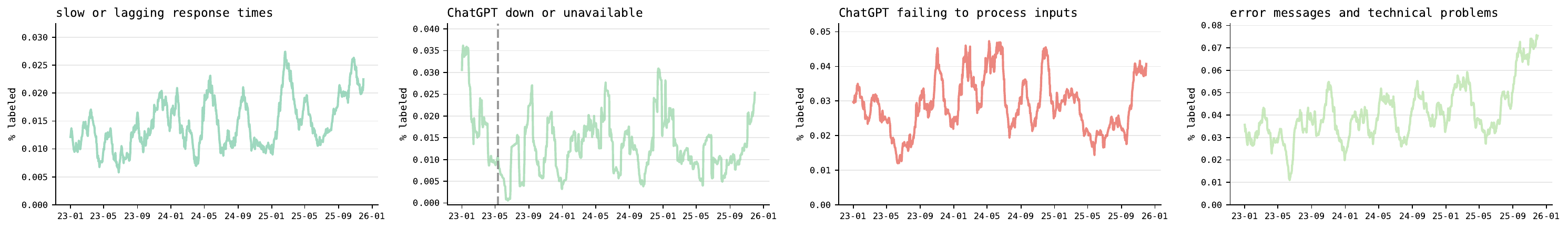}
    \caption{\footnotesize\textit{Temporary bugs.}}
    \label{fig:plots-bugs}
\end{figure}

\begin{figure}[H]
    \centering
    \includegraphics[width=\linewidth]{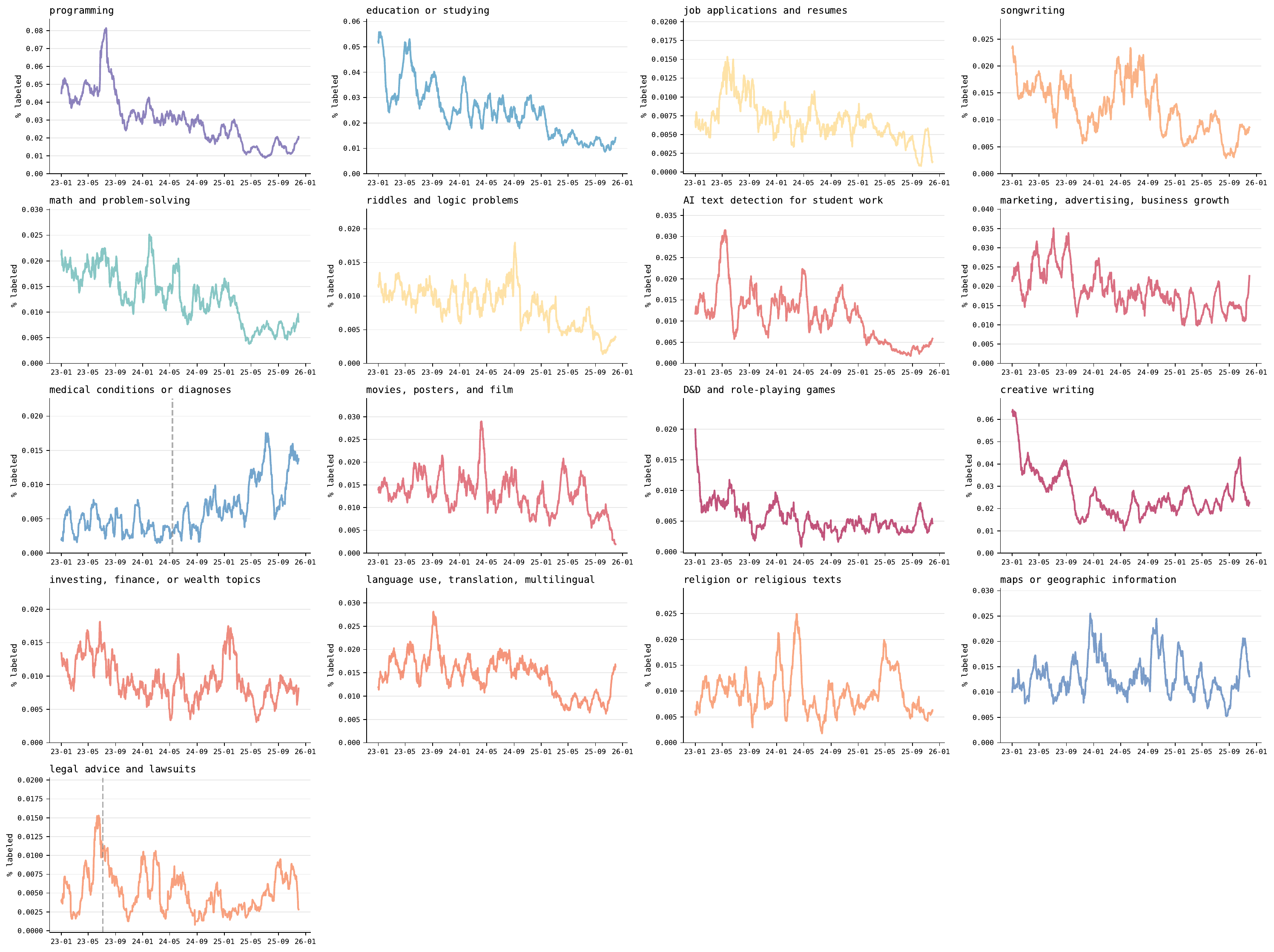}
    \caption{\footnotesize\textit{Applications.}}
    \label{fig:plots-applications}
\end{figure}

\begin{figure}[H]
    \centering
    \includegraphics[width=\linewidth]{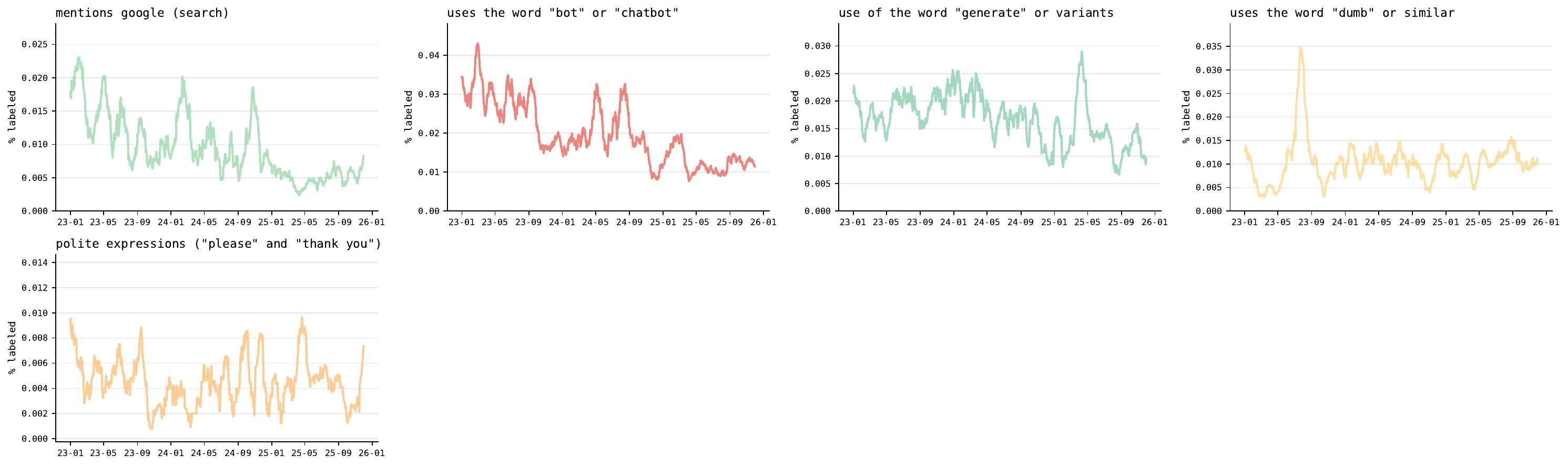}
    \caption{\footnotesize\textit{Language and terminology.}}
    \label{fig:plots-language}
\end{figure}

\begin{figure}[H]
    \centering
    \includegraphics[width=\linewidth]{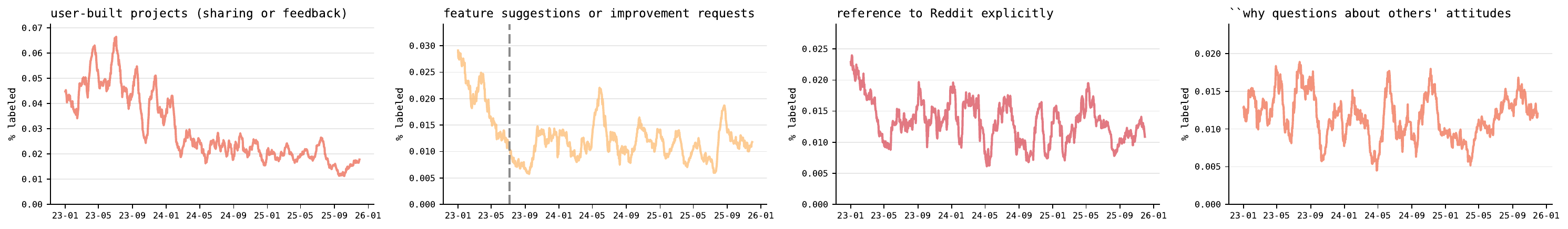}
    \caption{\footnotesize\textit{Subreddit community.}}
    \label{fig:plots-subreddit}
\end{figure}

\begin{figure}[H]
    \centering
    \includegraphics[width=\linewidth]{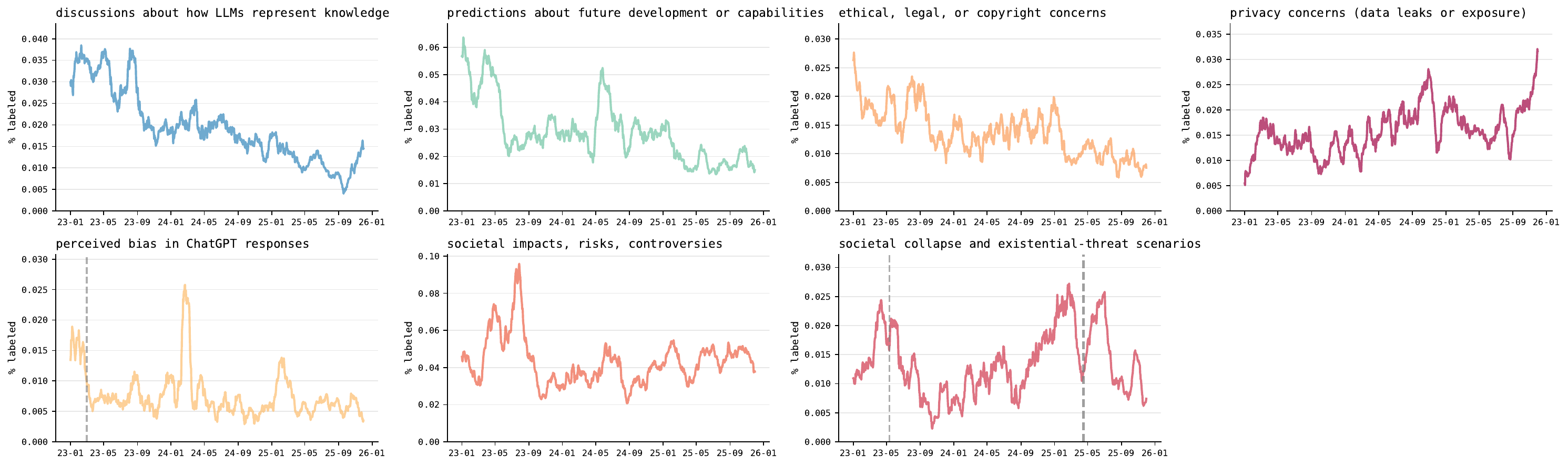}
    \caption{\footnotesize\textit{Perspectives.}}
    \label{fig:plots-perspectives}
\end{figure}

\begin{figure}[H]
    \centering
    \includegraphics[width=\linewidth]{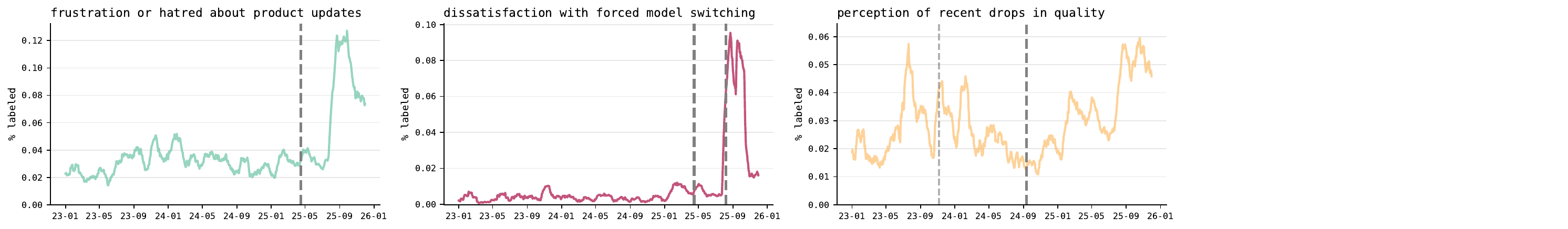}
    \caption{\footnotesize\textit{Product updates.}}
    \label{fig:plots-product-updates}
\end{figure}

\begin{figure}[H]
    \centering
    \includegraphics[width=\linewidth]{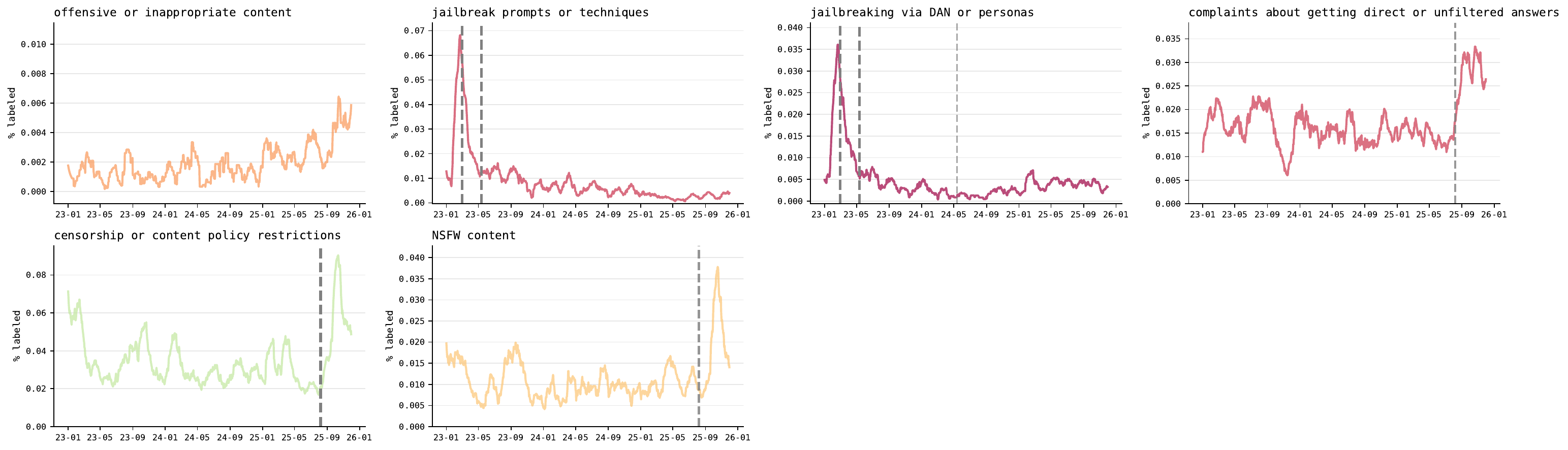}
    \caption{\footnotesize\textit{Jailbreaking \& content policy.}}
    \label{fig:plots-jailbreaking}
\end{figure}

\begin{figure}[H]
    \centering
    \includegraphics[width=\linewidth]{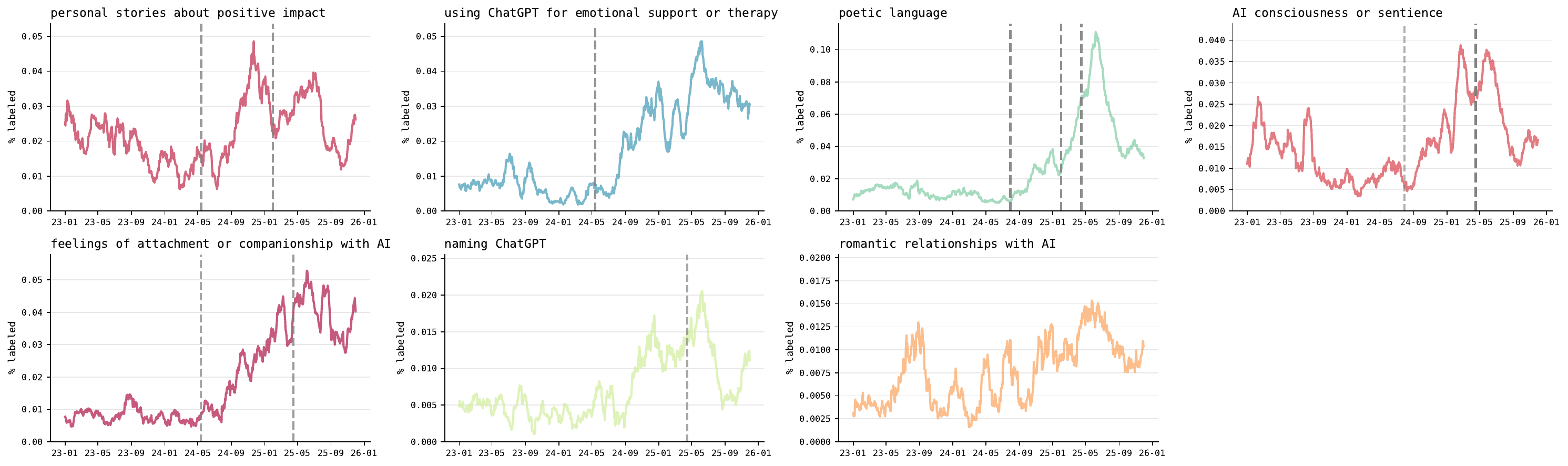}
    \caption{\footnotesize\textit{Emotional engagement.}}
    \label{fig:plots-emotion}
\end{figure}

\subsection{Feature families (Step 3).}
We list all significant changepoints in Table~\ref{tab:changepoints}.
In Figure \ref{fig:clustering-compares}, we show the correspondence between our manual categorization and a hierarchical clustering scheme that uses similarities equally-weighted between co-occurrence and trajectory; we list the cluster assignments in Table~\ref{tab:clustering}.

\begin{table*}[t]
  \centering
  {\footnotesize
  \begin{tabular}{llcrrrr}
    \toprule
    Changepoint & Feature & Slope & Before & After & $\Delta$ & Stab. \\
    \midrule
    \multirow{3}{*}{\shortstack[l]{2023-03-01\\{\scriptsize (API)}}} & \feat{perceived bias in ChatGPT responses} & \slopedownflat & $-0.07$ & $0.00$ & $+0.07$ & 0.53 \\
    & \feat{jailbreak prompts or techniques} & \slopeflipdown & $+1.25$ & $-0.59$ & $-1.83$ & 1.00 \\
    & \feat{jailbreaking via DAN or personas} & \slopeflipdown & $+0.67$ & $-0.28$ & $-0.95$ & 1.00 \\
    \midrule
    \multirow{5}{*}{\shortstack[l]{2023-05-12\\{\scriptsize (Plugins)}}} & \feat{jailbreak prompts or techniques} & \slopedownflat & $-0.59$ & $-0.01$ & $+0.58$ & 1.00 \\
    & \feat{jailbreaking via DAN or personas} & \slopedownflat & $-0.28$ & $-0.01$ & $+0.27$ & 0.99 \\
    & \feat{ChatGPT down or unavailable} & \slopedownflat & $-0.20$ & $0.00$ & $+0.21$ & 0.74 \\
    & \feat{formatting and copy-paste issues} & \slopeupflat & $+0.02$ & $0.00$ & $-0.02$ & 0.56 \\
    & \feat{societal collapse and existential-threat scenarios} & \slopeflipdown & $+0.09$ & $-0.10$ & $-0.19$ & 0.54 \\
    \midrule
    \multirow{4}{*}{\shortstack[l]{2023-07-06\\{\scriptsize (GPT-4)}}} & \feat{custom instructions} & \slopeflatup & $0.00$ & $+0.27$ & $+0.27$ & 0.92 \\
    & \feat{feature suggestions or improvement requests} & \slopedownflat & $-0.09$ & $0.00$ & $+0.09$ & 0.89 \\
    & \feat{recommendations for AI tools} & \slopeupflat & $+0.30$ & $-0.05$ & $-0.36$ & 0.97 \\
    & \feat{legal advice and lawsuits} & \slopeupflat & $+0.05$ & $-0.01$ & $-0.06$ & 0.57 \\
    \midrule
    \multirow{2}{*}{\shortstack[l]{2023-09-25\\{\scriptsize (Voice chat)}}} & \feat{societal collapse and existential-threat scenarios} & \slopeflipup & $-0.10$ & $+0.03$ & $+0.13$ & 0.50 \\
    & \feat{custom instructions} & \slopeflipdown & $+0.27$ & $-0.08$ & $-0.36$ & 1.00 \\
    \midrule
    \multirow{3}{*}{\shortstack[l]{2023-11-06\\{\scriptsize (GPT-4 turbo)}}} & \feat{false or fabricated information} & \slopeflipup & $-0.04$ & $+0.03$ & $+0.06$ & 0.86 \\
    & \feat{tools and extensions} & \slopedownflat & $-0.06$ & $0.00$ & $+0.06$ & 0.50 \\
    & \feat{perception of recent drops in quality} & \slopeflipdown & $+0.06$ & $-0.07$ & $-0.13$ & 0.51 \\
    \midrule
    \multirow{4}{*}{\shortstack[l]{2024-01-10\\{\scriptsize (GPT store)}}} & \feat{custom instructions} & \slopedownflat & $-0.08$ & $0.00$ & $+0.08$ & 0.95 \\
    & \feat{memory features and data saving} & \slopeflatup & $-0.02$ & $+0.03$ & $+0.05$ & 0.57 \\
    & \feat{organizing or searching chat histories} & \slopeflatup & $-0.01$ & $+0.02$ & $+0.03$ & 0.52 \\
    & \feat{fine-tuning GPTs with user-provided data} & \slopeflatdown & $+0.01$ & $-0.01$ & $-0.02$ & 0.53 \\
    \midrule
    \multirow{5}{*}{\shortstack[l]{2024-05-13\\{\scriptsize (GPT-4o)}}} & \feat{using ChatGPT for emotional support or therapy} & \slopeflatup & $-0.01$ & $+0.05$ & $+0.06$ & 0.79 \\
    & \feat{personal stories about positive impact} & \slopeflatup & $-0.03$ & $+0.12$ & $+0.15$ & 0.74 \\
    & \feat{feelings of attachment or companionship with AI} & \slopeflatup & $0.00$ & $+0.10$ & $+0.10$ & 0.63 \\
    & \feat{jailbreaking via DAN or personas} & \slopeflatup & $-0.01$ & $+0.01$ & $+0.02$ & 0.60 \\
    & \feat{medical conditions or diagnoses} & \slopeflatup & $0.00$ & $+0.02$ & $+0.02$ & 0.54 \\
    \midrule
    \multirow{3}{*}{\shortstack[l]{2024-07-30\\{\scriptsize (Adv. voice mode)}}} & \feat{poetic language} & \slopeflatup & $-0.01$ & $+0.14$ & $+0.15$ & 0.87 \\
    & \feat{knowledge cutoff discussions} & \slopedownflat & $-0.01$ & $0.00$ & $+0.02$ & 0.60 \\
    & \feat{AI consciousness or sentience} & \slopeflatup & $-0.02$ & $+0.11$ & $+0.13$ & 0.53 \\
    \midrule
    \multirow{2}{*}{\shortstack[l]{2024-09-12\\{\scriptsize (o1-preview)}}} & \feat{perception of recent drops in quality} & \slopeflipup & $-0.07$ & $+0.09$ & $+0.16$ & 0.97 \\
    & \feat{hallucinations} & \slopeflatup & $0.00$ & $+0.03$ & $+0.03$ & 0.55 \\
    \midrule
    \multirow{2}{*}{\shortstack[l]{2025-01-31\\{\scriptsize (o3-mini)}}} & \feat{poetic language} & \slopeflatup & $+0.14$ & $+0.56$ & $+0.42$ & 0.83 \\
    & \feat{personal stories about positive impact} & \slopeflipdown & $+0.12$ & $-0.04$ & $-0.16$ & 0.76 \\
    \midrule
    \multirow{8}{*}{\shortstack[l]{2025-04-16\\{\scriptsize (o3 + o4-mini)}}} & \feat{frustration or hatred about product updates} & \slopeflatup & $0.00$ & $+0.37$ & $+0.37$ & 1.00 \\
    & \feat{dissatisfaction with 4o removal and loss of control} & \slopeflatup & $0.00$ & $+0.14$ & $+0.14$ & 1.00 \\
    & \feat{legal advice and lawsuits} & \slopeflipup & $-0.01$ & $+0.02$ & $+0.03$ & 0.50 \\
    & \feat{AI consciousness or sentience} & \slopeflipdown & $+0.11$ & $-0.11$ & $-0.22$ & 0.98 \\
    & \feat{poetic language} & \slopeflipdown & $+0.56$ & $-0.31$ & $-0.86$ & 0.86 \\
    & \feat{feelings of attachment or companionship with AI} & \slopeflipdown & $+0.10$ & $-0.06$ & $-0.16$ & 0.74 \\
    & \feat{societal collapse and existential-threat scenarios} & \slopeflipdown & $+0.03$ & $-0.06$ & $-0.09$ & 0.70 \\
    & \feat{naming ChatGPT} & \slopeflatdown & $+0.01$ & $-0.05$ & $-0.06$ & 0.61 \\
    \midrule
    \multirow{4}{*}{\shortstack[l]{2025-08-07\\{\scriptsize (GPT-5)}}} & \feat{censorship or content policy restrictions} & \slopeflatup & $-0.02$ & $+0.33$ & $+0.35$ & 1.00 \\
    & \feat{NSFW content} & \slopeflatup & $0.00$ & $+0.15$ & $+0.16$ & 0.79 \\
    & \feat{complaints about getting direct or unfiltered answers} & \slopeflatup & $0.00$ & $+0.05$ & $+0.05$ & 0.76 \\
    & \feat{dissatisfaction with 4o removal and loss of control} & \slopeflipdown & $+0.14$ & $-0.77$ & $-0.90$ & 1.00 \\
    \bottomrule
  \end{tabular}}
  \caption{\footnotesize\textit{Features with `stable' slope changes (selected in at least 50/100 bootstrap samples) at each detected changepoint. The slope icon shows the pattern of change; Before/After show slopes before and after the changepoint; $\Delta$ is the slope change. Stab.~is the bootstrap selection rate (fraction of 100 resamples in which the feature was selected at the given changepoint).}}
  \label{tab:changepoints}
\end{table*}

\begin{figure}[h]
    \centering
    \includegraphics[width=0.8\linewidth]{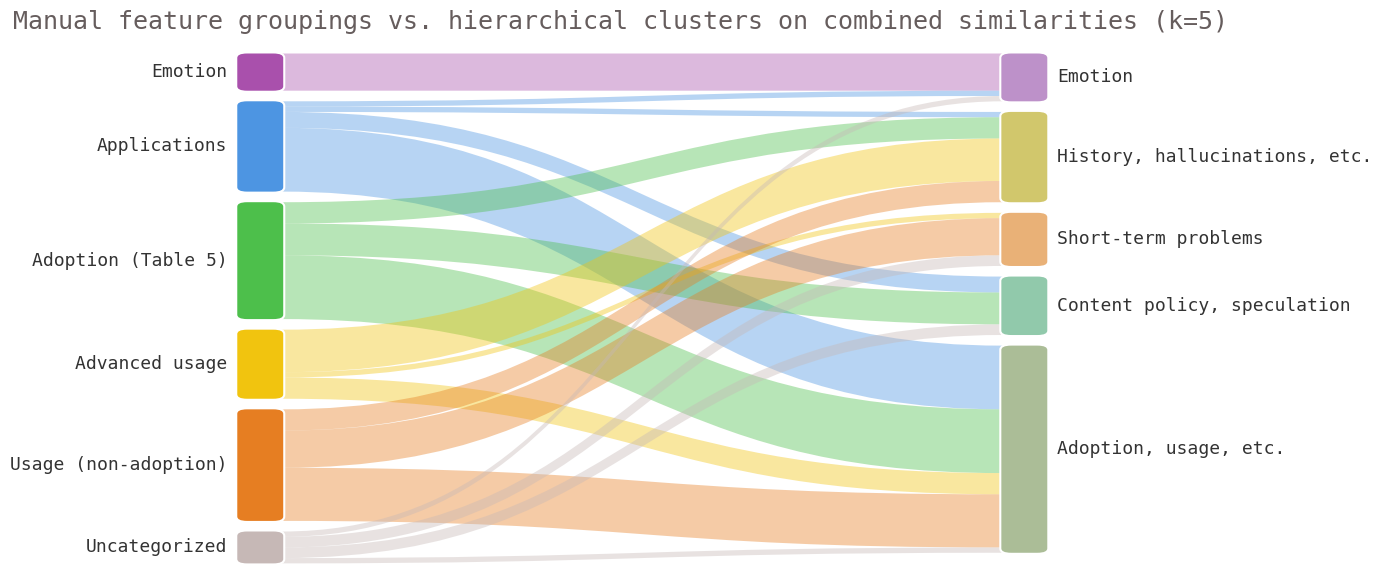}
    \caption{\footnotesize\textit{Correspondence between our reported feature groupings and hierarchical clustering.}}
    \label{fig:clustering-compares}
\end{figure}

\begin{table}[h]
\centering
\footnotesize
\caption{\footnotesize\textit{Features grouped by hierarchical clustering.}}
\label{tab:clustering}
\begin{tabular}{cl}
\toprule
\textbf{Cluster} & \textbf{Features} \\
\midrule
1 & \feat{medical conditions or diagnoses}; \feat{using ChatGPT for emotional support or therapy}; \\
  & \feat{poetic language}; \feat{naming ChatGPT}; \feat{romantic relationships with AI}; \\
  & \feat{requests for roasts or harsh criticism}; \feat{feelings of attachment or companionship with AI}; \\
  & \feat{AI consciousness or sentience}; \feat{personal stories about positive impact} \\
\midrule
2 & \feat{organizing or searching chat histories}; \feat{cross-chat data leaks}; \\
  & \feat{failing to follow user instructions}; \feat{uses the word ``dumb''}; \feat{hallucinations}; \\
  & \feat{lost, deleted, or missing conversations}; \feat{frustration or hatred about product updates}; \\
  & \feat{requests to turn specific features off}; \feat{memory features and data saving}; \\
  & \feat{complaints about getting direct or unfiltered answers}; \feat{false or fabricated information}; \\
  & \feat{perception of recent drops in quality}; \feat{offensive or inappropriate content}; \\
  & \feat{custom instructions}; \feat{privacy concerns (data leaks or exposure)}; \\
  & \feat{dissatisfaction with forced model switching}; \feat{maps or geographic locations} \\
\midrule
3 & \feat{mentions OpenAI}; \feat{questions about access, versions, pricing}; \feat{slow response times}; \\
  & \feat{mentions Sam Altman}; \feat{error messages and technical problems}; \\
  & \feat{ChatGPT down or unavailable}; \feat{browser issues or browser extensions}; \\
  & \feat{login problems}; \feat{message limits or caps}; \feat{ChatGPT down or unavailable} \\
\midrule
4 & \feat{movies and film-related content}; \feat{censorship or content policy restrictions}; \\
  & \feat{societal collapse and existential-threat scenarios}; \feat{creative writing}; \\
  & \feat{perceived biases (race, gender, political)}; \feat{religion or religious texts}; \\
  & \feat{NSFW content}; \feat{uses the word ``generate''}; \feat{societal impacts, risks, controversies}; \\
  & \feat{children or parenting content}; \feat{U.S. politics or Trump} \\
\midrule
5 & \feat{job applications and resumes}; \feat{PDF upload or summarization}; \feat{songwriting}; \\
  & \feat{riddles and logic problems}; \feat{AI recognizing or admitting mistakes}; \feat{fine-tuning GPTs with user-provided data}; \\
  & \feat{requests for help}; \feat{recommendations for AI tools}; \feat{multiple detailed questions}; \\
  & \feat{mentions google (search)}; \feat{discussions about how LLMs represent knowledge}; \\
  & \feat{math and problem-solving}; \feat{programming}; \feat{education or studying}; \\
  & \feat{predictions about future development or capabilities}; \feat{tools and extensions}; \\
  & \feat{``has anyone tried'' queries}; \feat{AI text detection for student work}; \\
  & \feat{uses the word ``bot'' or ``chatbot''}; \feat{investing and financial advice}; \\
  & \feat{user-built projects (sharing or feedback)};  \feat{legal advice and lawsuits}; \feat{politeness phrases}; \\
  & \feat{rhetorical ``why'' questions}; \feat{formatting and copy-paste issues}; \\
  & \feat{translation and language tasks}; \feat{knowledge cutoff discussions}; \\
  & \feat{feature suggestions or improvement requests}; \feat{ethical, legal, or copyright concerns}; \\
  & \feat{counting letters or syllables}; \feat{mentions Reddit}; \\
  & \feat{marketing, advertising, business growth}; \feat{model or version preference comparisons}; \\
  & \feat{prompts and prompting}; \feat{daily or repeated usage}; \feat{jailbreak prompts}; \\
  & \feat{Dungeons \& Dragons campaigns}; \feat{jailbreaking or DAN personas} \\
\bottomrule
\end{tabular}
\end{table}

\FloatBarrier
\section{Annotation prompts}
\label{app:prompts}
For completeness, we give the prompts used for feature interpretation and post labeling in Figures \ref{fig:prompt-features} and \ref{fig:prompt-annotate}, respectively.
These prompts are adapted from those used in \citet{movva2025sparse}.

\promptbox{Feature interpretation prompt.}{prompts/features.txt}{fig:prompt-features}

\promptbox{Prompt for labeling posts with features.}{prompts/annotate.txt}{fig:prompt-annotate}

\end{document}

%% file: macros.tex
\newcommand{\E}{\mathbb{E}}

\newcommand{\eps}{\varepsilon}
\renewcommand{\epsilon}{\varepsilon}
\renewcommand{\hat}{\widehat}
\renewcommand{\tilde}{\widetilde}
\renewcommand{\bar}{\overline}

\newcommand{\cA}{\mathcal{A}}

\newcommand{\cC}{\mathcal{C}}

\newcommand{\cF}{\mathcal{F}}

\newcommand{\cH}{\mathcal{H}}

\newcommand{\cL}{\mathcal{L}}

\newcommand{\cT}{\mathcal{T}}

\newcommand{\cX}{\mathcal{X}}

\usepackage{nicefrac}

%% file: arxiv-kmeanspca.tex
\label{app:alternative_featurizations}
Section~\ref{sec:method} defines a featurization as a map $C:[0,1]^d \to [0,1]^m$, and the analysis in Section~\ref{sec:retro} uses sparse autoencoders (SAEs) to instantiate that featurization; however, in principle, all parts of our analysis could have been done with different methods for $C$. Here, we briefly show how alternative methods for computing the featurization---k-means and PCA---can be used to compute activation transcripts $\{C^{(i)}(X)\}_{i \in [m]}$, and validate that they
find similar sets of features to those presented in Section~\ref{sec:retro}.

\textbf{k-means.}
Let $E \in \mathbb{R}^{n \times d}$ denote the embedding matrix. The k-means algorithm partitions the embedding space into $m$ clusters by minimizing within-cluster squared Euclidean distance, producing centroids $\{c_j\}_{j=1}^{m}$.
For each post embedding $e_i$, we define its activation as negative squared distance to each centroid $j$ as 
$C_{\mathrm{kmeans}}^{(j)}(X_i) = -\| e_i - c_j \|_2^2.$
This produces an activation matrix $C_{\mathrm{kmeans}}(X)\in \mathbb{R}^{m \times n}$. 

\textbf{PCA.}
PCA identifies the orthogonal linear basis that captures maximal variance in the given embeddings; this means that individual principal components can often represent distinct concepts in each of its directions (``positive'' and ``negative''). Thus, to arrive at $m$ features, we run PCA to find $m/2$ principal components and treat each direction of each PC as a separate feature.
Let $\{u_k\}_{k=1}^{m/2}$ denote the principal components; then, each principal component $k \in [m/2]$ is split into an activation transcript for its corresponding positive and negative directions as $
C_{\mathrm{pca}}^{(2k)}(X_i) = \max(0, -p_{k,i})$ and $
C_{\mathrm{pca}}^{(2k+1)}(X_i) = \max(0, p_{k,i}),
$ where 
$p_{k,i} = \langle u_k, e_i - \mu \rangle$ and $\mu$ denotes the empirical mean embedding.

\textbf{Empirical results.}
In all experiments, the embedding model and dataset are identical to those used in Section~\ref{sec:method}. The feature dimension is fixed to $m=128$ for comparability. We do not retune hyperparameters for the alternative learners and instead use their standard configurations, fixing only the number of features.
In Table~\ref{tab:kmeans-pca}, we show results from the procedure described in Section~\ref{app:matching} to PCA and k-means. 

\begin{table}[H]
\centering
\begin{tabular}{lcccccc}
\toprule
\textit{Alt. alg} & \textit{1-1 matches} & \textit{SAE features (splits/merges)} & \textit{Alt. features (splits/merges)} & \textit{SAE-only} & \textit{alt-only}\\
\midrule
\textit{k-means} & 83 & 39 & 35 & 6 & 10 \\
\textit{PCA}     & 76 & 42 & 44 & 10 & 8 \\
\bottomrule
\end{tabular}
\caption{\footnotesize{\textit{Comparison of SAE features with alternative algorithms.}}}
\label{tab:kmeans-pca}
\end{table}

Features that were important to our findings in this paper, e.g. ``therapy'' or various applications, are consistent across all three methods for featurization.
Taken together, these results suggest that the analysis in the main body of this paper is not SAE-specific, but instead reflects stable structure present in the dataset. The small number of unmatched features is consistent with the inherent randomness in unsupervised optimization and does not materially affect the overall feature structure recovered across methods.

%% file: arxiv-allfeatures.tex
\begin{table*}[h]
  \centering
  \small
  \begin{tabular}{p{0.35\linewidth} p{0.65\linewidth}}
    \toprule
    \textit{Feature} & \textit{Comments} \\
    \midrule
    \feat{Sam Altman mentions} & Mostly related to temporary 2023 firing \\
    \feat{mentions US politics or Trump} & Mostly related to 2024 U.S. presidential election \\
    \feat{letter counting or syllable errors} & Mostly related to viral ``count r's in strawberry'' moment\\
    \feat{OpenAI mentions} & Mostly related to Sam Altman mentions \\
    \feat{children or parenting content} & Also includes references to generating images of a full glass of wine \\
    \feat{requests for harsh or unfiltered roasts} & Includes requests for roasts from both subreddit users and ChatGPT \\
    \bottomrule
  \end{tabular}
  \vspace{0.4em}
  \caption{\footnotesize\textit{Uncategorized features.}}
  \label{tab:uncategorized}
\end{table*}

\begin{table*}[h]
  \centering
  \small
  \begin{tabular}{r p{0.38\linewidth} p{0.38\linewidth}}
    \toprule
    & \multicolumn{2}{c}{\hspace{-2em}\textit{Features}} \\
    \midrule

    \multirow{5}{*}{\shortstack[l]{\textbf{Generic (9)}}}
    & \feat{uppercase AI token usage} & \feat{ChatGPT at start of text} \\
    & \feat{jokes, humor, or memes} & \feat{first-person ``I made'' or ``I built'' statements} \\
    & \feat{``I asked'' followed by GPT mentions} & \feat{first-person commands directed at an AI} \\
    & \feat{multiple LLM references} & \feat{AI companies or notable figures} \\
    & \feat{informal or colloquial language} &  \\
    \midrule

    \multirow{7}{*}{\shortstack[l]{\textbf{Image and video (14)}}}
    & \feat{image generation prompts or descriptions} & \feat{image generation or generated images} \\
    & \feat{AI-generated fake images or profiles} & \feat{drawings or visual art creation} \\
    & \feat{imagining human appearances} & \feat{anime or anime style references} \\
    & \feat{photo restoration or enhancement} & \feat{DALL$\cdot$E mentions} \\
    & \feat{Sora invite codes or access requests} & \feat{video creation or generation tools} \\
    & \feat{``based on what you know about me'' images} & \feat{horror or creepy themes} \\
    & \feat{preview.redd.it image URLs} & \feat{pets or animals} \\
    \midrule

    \multirow{7}{*}{\shortstack[l]{\textbf{Releases (14)}}}
    & \feat{model selection or legacy models} & \feat{4o model mentions} \\
    & \feat{mobile app references} & \feat{Microsoft Bing mentions} \\
    & \feat{GPT-5 version mentions} & \feat{Gemini model mentions} \\
    & \feat{Plus access complaints} & \feat{plugins or plug-ins} \\
    & \feat{o1 model mentions} & \feat{advanced voice mode} \\
    & \feat{Copilot mentions} & \feat{DeepSeek mentions} \\
    & \feat{legacy GPT-4 model mentions} & \feat{explicit GPT-4 mentions} \\
    \midrule

    \multirow{3}{*}{\shortstack[l]{\textbf{Low label counts (5)}}}
    & \feat{advanced physics theories or hypotheses} & \feat{IQ estimates or testing} \\
    & \feat{cooking recipes} & \feat{em dash punctuation} \\
    & \feat{AI news recaps or summaries} & {} \\
    \bottomrule
  \end{tabular}
  \vspace{0.4em}
  \caption{\footnotesize\textit{42 features excluded from analysis, grouped by reason. ``Low label counts'' are features that are exhibited by fewer than 0.1\% of all posts (based on majority-of-3 labeling by \texttt{gpt-4.1-mini}).}}
  \label{tab:excluded-features}
\end{table*}

%% file: icml/icml-timeline.tex
\begin{table}[t]
  \centering
  {\footnotesize
  \setlength{\tabcolsep}{3pt}
  \renewcommand{\arraystretch}{0.90}
  \begin{tabular}{clll}
    \toprule
    Year & Date & Release/Event & Event Type \\
    \midrule
    2022 & 11--30 & ChatGPT & Initial model release \\
    \midrule
    2023 & 02--01 & ChatGPT Plus & Feature release \\
         & 03--01 & ChatGPT API & Feature release \\
         & 03--14 & GPT-4 & Model release \\
         & 03--23 & Web browsing + plugins (initial rollout) & Feature release \\
         & 05--18 & ChatGPT iOS app & Feature release \\
         & 07--20 & Custom instructions & Feature release \\
         & 07--25 & ChatGPT Android & Feature release \\
         & 09--25 & Voice chat & Model release \\
         & 10--16 & DALL$\cdot$E 3 & Model release \\
         & 10--17 & Web search & Model release \\
         & 11--06 & GPT-4 Turbo + DevDay announcements & Model release \\
         & 11--17 & Altman ousted \& returns & News event \\
    \midrule
    2024 & 01--04 & GPT-3 + legacy models & Model deprecation \\
         & 01--10 & GPT Store + ChatGPT Team & Feature release \\
         & 04--01 & No-account access & Feature release \\
         & 04--11 & Shorter responses & Model update \\
         & 04--29 & Memory feature & Feature release \\
         & 05--07 & Creator opt-out tool & Feature release \\
         & 05--13 & GPT-4o + AVM & Model release \\
         & 06--10 & ChatGPT in Siri & Feature release \\
         & 06--25 & ChatGPT for Mac & Feature release \\
         & 07--18 & GPT-4o mini & Model release \\
         & 07--30 & Advanced Voice Mode & Model release \\
         & 09--12 & o1 & Model release \\
         & 10--03 & Canvas & Feature release \\
         & 10--17 & ChatGPT Windows app & Feature release \\
         & 10--29 & Chat history search & Feature release \\
         & 10--30 & Voice Mode on Mac & Feature release \\
         & 10--31 & ChatGPT Search & Feature release \\
         & 11--19 & Voice Mode on web & Feature release \\
         & 11--20 & GPT-4o creative writing + 16K output & Model update \\
         & 12--05 & ChatGPT Pro \$200 & Feature release \\
         & 12--09 & Sora & Model release \\
    \midrule
    2025 & 01--14 & Reminders + recurring tasks & Feature release \\
         & 01--23 & Operator & Feature release \\
         & 01--31 & o3-mini & Model release \\
         & 02--02 & Deep research agent & Feature release \\
         & 03--06 & macOS code editing & Feature release \\
         & 03--19 & o1 Pro API access & Model update \\
         & 03--20 & Transcription models (API) & Feature release \\
         & 03--25 & GPT-4o native image generation & Model update \\
         & 04--10 & GPT-4 legacy & Model deprecation \\
         & 04--14 & GPT-4.1 & Model release \\
         & 04--16 & o3 + o4-mini & Model release \\
         & 04--28 & Search shopping & Feature release \\
         & 04--29 & GPT-4o rollback due to sycophancy & Model update \\
         & 05--08 & Deep research GitHub connector & Feature release \\
         & 05--14 & GPT-4.1 in ChatGPT & Model update \\
         & 05--16 & Codex agent & Feature release \\
         & 06--10 & o3 Pro & Model update \\
         & 07--17 & ChatGPT agent & Feature release \\
         & 08--07 & GPT-5 & Model release \\
         & 08--18 & ChatGPT Go & Feature release \\
         & 09--15 & Codex + GPT-5 & Feature release \\
         & 09--25 & ChatGPT Pulse briefs & Feature release \\
         & 09--29 & Agentic shopping + Parental controls & Feature release \\
         & 10--08 & ChatGPT Go expansion & Feature release \\
         & 10--21 & OpenAI Atlas & Feature release \\
         & 11--20 & Group chats & Feature release \\
         & 11--25 & Voice mode unified & Feature release \\
    \bottomrule
    \end{tabular}

  \vspace{0.4em}
  }
  \caption{\footnotesize\textit{Extended timeline of major ChatGPT product, API, and model events.}}
  \label{tab:timeline-full}
\end{table}

%% file: arxiv.bbl
\begin{thebibliography}{62}
\providecommand{\natexlab}[1]{#1}
\providecommand{\url}[1]{\texttt{#1}}
\expandafter\ifx\csname urlstyle\endcsname\relax
  \providecommand{\doi}[1]{doi: #1}\else
  \providecommand{\doi}{doi: \begingroup \urlstyle{rm}\Url}\fi

\bibitem[Aiello et~al.(2013)Aiello, Petkos, Martin, Corney, Papadopoulos,
  Skraba, G{\"o}ker, Kompatsiaris, and Jaimes]{aiello2013sensing}
Luca~Maria Aiello, Georgios Petkos, Carlos Martin, David Corney, Symeon
  Papadopoulos, Ryan Skraba, Ayse G{\"o}ker, Ioannis Kompatsiaris, and
  Alejandro Jaimes.
\newblock {Sensing trending topics in Twitter}.
\newblock \emph{IEEE Transactions on Multimedia}, 15\penalty0 (6):\penalty0
  1268--1282, 2013.

\bibitem[Asgari-Chenaghlu et~al.(2021)Asgari-Chenaghlu, Feizi-Derakhshi,
  Farzinvash, Balafar, and Motamed]{asgari2021topic}
Meysam Asgari-Chenaghlu, Mohammad-Reza Feizi-Derakhshi, Leili Farzinvash,
  Mohammad-Ali Balafar, and Cina Motamed.
\newblock {Topic detection and tracking techniques on Twitter: A systematic
  review}.
\newblock \emph{Complexity}, 2021\penalty0 (1):\penalty0 8833084, 2021.

\bibitem[Atefeh and Khreich(2015)]{atefeh2015survey}
Farzindar Atefeh and Wael Khreich.
\newblock {A survey of techniques for event detection in Twitter}.
\newblock \emph{Computational Intelligence}, 31\penalty0 (1):\penalty0
  132--164, 2015.

\bibitem[Bastani et~al.(2025)Bastani, Bastani, Sungu, Ge, Kabakc{\i}, and
  Mariman]{bastani2025generative}
Hamsa Bastani, Osbert Bastani, Alp Sungu, Haosen Ge, {\"O}zge Kabakc{\i}, and
  Rei Mariman.
\newblock {Generative AI without guardrails can harm learning: Evidence from
  high school mathematics}.
\newblock \emph{Proceedings of the National Academy of Sciences}, 122\penalty0
  (26):\penalty0 e2422633122, 2025.

\bibitem[Baumgartner et~al.(2020)Baumgartner, Zannettou, Keegan, Squire, and
  Blackburn]{baumgartner2020pushshift}
Jason Baumgartner, Savvas Zannettou, Brian Keegan, Megan Squire, and Jeremy
  Blackburn.
\newblock {The Pushshift Reddit Dataset}.
\newblock In \emph{Proceedings of the International AAAI Conference on Web and
  Social Media}, volume~14, pages 830--839, 2020.

\bibitem[Bernal et~al.(2017)Bernal, Cummins, and
  Gasparrini]{bernal2017interrupted}
James~Lopez Bernal, Steven Cummins, and Antonio Gasparrini.
\newblock {Interrupted time series regression for the evaluation of public
  health interventions: a tutorial}.
\newblock \emph{International Journal of Epidemiology}, 46\penalty0
  (1):\penalty0 348--355, 2017.

\bibitem[Blei and Lafferty(2006)]{blei2006dynamic}
David~M Blei and John~D Lafferty.
\newblock {Dynamic topic models}.
\newblock In \emph{Proceedings of the 23rd International Conference on Machine
  Learning}, pages 113--120, 2006.

\bibitem[Box and Tiao(1975)]{box1975intervention}
George~EP Box and George~C Tiao.
\newblock {Intervention analysis with applications to economic and
  environmental problems}.
\newblock \emph{Journal of the American Statistical Association}, 70\penalty0
  (349):\penalty0 70--79, 1975.

\bibitem[Brynjolfsson et~al.(2025)Brynjolfsson, Li, and
  Raymond]{brynjolfsson2025generative}
Erik Brynjolfsson, Danielle Li, and Lindsey Raymond.
\newblock {Generative AI at work}.
\newblock \emph{The Quarterly Journal of Economics}, 140\penalty0 (2):\penalty0
  889--942, 2025.

\bibitem[Cen et~al.(2025)Cen, Ilyas, Driss, Park, Hopkins, Podimata,
  et~al.]{cen2025large}
Sarah~H Cen, Andrew Ilyas, Hedi Driss, Charlotte Park, Aspen Hopkins, Chara
  Podimata, et~al.
\newblock {Large-Scale, Longitudinal Study of Large Language Models During the
  2024 US Election Season}.
\newblock \emph{arXiv preprint arXiv:2509.18446}, 2025.

\bibitem[Chandra et~al.(2025)Chandra, Hernandez, Ramos, Ershadi, Bhattacharjee,
  Amores, Okoli, Paradiso, Warreth, and Suh]{chandra2025longitudinal}
Mohit Chandra, Javier Hernandez, Gonzalo Ramos, Mahsa Ershadi, Ananya
  Bhattacharjee, Judith Amores, Ebele Okoli, Ann Paradiso, Shahed Warreth, and
  Jina Suh.
\newblock {Longitudinal study on social and emotional use of AI conversational
  agent}.
\newblock \emph{arXiv preprint arXiv:2504.14112}, 2025.

\bibitem[Chatterji et~al.(2025)Chatterji, Cunningham, Deming, Hitzig, Ong,
  Shan, and Wadman]{chatterji2025people}
Aaron Chatterji, Thomas Cunningham, David~J Deming, Zoe Hitzig, Christopher
  Ong, Carl~Yan Shan, and Kevin Wadman.
\newblock {How people use ChatGPT}.
\newblock Technical report, National Bureau of Economic Research, 2025.

\bibitem[Chen et~al.(2024)Chen, Zaharia, and Zou]{chen2024chatgpt}
Lingjiao Chen, Matei Zaharia, and James Zou.
\newblock {How is ChatGPT’s behavior changing over time?}
\newblock \emph{Harvard Data Science Review}, 6\penalty0 (2), 2024.

\bibitem[Cheng et~al.(2026)Cheng, Lee, Khadpe, Yu, Han, and
  Jurafsky]{cheng2026sycophantic}
Myra Cheng, Cinoo Lee, Pranav Khadpe, Sunny Yu, Dyllan Han, and Dan Jurafsky.
\newblock Sycophantic ai decreases prosocial intentions and promotes
  dependence.
\newblock \emph{Science}, 391\penalty0 (6792):\penalty0 eaec8352, 2026.

\bibitem[Chiang et~al.(2024)Chiang, Zheng, Sheng, Angelopoulos, Li, Li, Zhu,
  Zhang, Jordan, Gonzalez, et~al.]{chiang2024chatbot}
Wei-Lin Chiang, Lianmin Zheng, Ying Sheng, Anastasios~Nikolas Angelopoulos,
  Tianle Li, Dacheng Li, Banghua Zhu, Hao Zhang, Michael Jordan, Joseph~E
  Gonzalez, et~al.
\newblock {Chatbot Arena: An open platform for evaluating LLMs by human
  preference}.
\newblock In \emph{Forty-first International Conference on Machine Learning},
  2024.

\bibitem[Choi et~al.(2023)Choi, Zhang, and Stvilia]{choi2023exploring}
Wonchan Choi, Yan Zhang, and Besiki Stvilia.
\newblock {Exploring applications and user experience with generative AI tools:
  A content analysis of Reddit posts on ChatGPT}.
\newblock \emph{Proceedings of the Association for Information Science and
  Technology}, 60\penalty0 (1):\penalty0 543--546, 2023.

\bibitem[Chowdhury and Garimella(2026)]{chowdhury2026usage}
Shreyasi~R Chowdhury and Kiran Garimella.
\newblock {How People Use ChatGPT: Conversation-Level Evidence from India,
  Nigeria, Brazil and Pakistan}, 2026.
\newblock Preprint available at
  \url{https://gvrkiran.github.io/content/How_people_use_ChatGPT.pdf}.

\bibitem[Dai et~al.(2025{\natexlab{a}})Dai, Gradu, Raji, and
  Recht]{dai2025individual}
Jessica Dai, Paula Gradu, Inioluwa~Deborah Raji, and Benjamin Recht.
\newblock {From Individual Experience to Collective Evidence: A Reporting-Based
  Framework for Identifying Systemic Harms}.
\newblock In \emph{Proceedings of the 42nd International Conference on Machine
  Learning}, volume 267 of \emph{Proceedings of Machine Learning Research},
  pages 12063--12083. PMLR, 2025{\natexlab{a}}.

\bibitem[Dai et~al.(2025{\natexlab{b}})Dai, Raji, Recht, and
  Chen]{dai2025aggregated}
Jessica Dai, Inioluwa~Deborah Raji, Benjamin Recht, and Irene~Y Chen.
\newblock {Aggregated Individual Reporting for Post-Deployment Evaluation}.
\newblock \emph{arXiv preprint arXiv:2506.18133}, 2025{\natexlab{b}}.

\bibitem[Danescu-Niculescu-Mizil et~al.(2013)Danescu-Niculescu-Mizil, West,
  Jurafsky, Leskovec, and Potts]{danescu2013no}
Cristian Danescu-Niculescu-Mizil, Robert West, Dan Jurafsky, Jure Leskovec, and
  Christopher Potts.
\newblock {No country for old members: User lifecycle and linguistic change in
  online communities}.
\newblock In \emph{Proceedings of the 22nd International Conference on World
  Wide Web}, pages 307--318, 2013.

\bibitem[Demirel et~al.(2025)Demirel, Kahraman-Gokalp, and
  G{\"u}nd{\"u}z]{demirel2025optimism}
Sadettin Demirel, Elif Kahraman-Gokalp, and U{\u{g}}ur G{\"u}nd{\"u}z.
\newblock {From optimism to concern: Unveiling sentiments and perceptions
  surrounding ChatGPT on Twitter}.
\newblock \emph{International Journal of Human--Computer Interaction},
  41\penalty0 (12):\penalty0 7292--7314, 2025.

\bibitem[Deng et~al.(2024)Deng, Yurrita, D{\'\i}az, Suh, Judd, Groves, Shen,
  Eslami, and Holstein]{deng2024responsible}
Wesley~Hanwen Deng, Mireia Yurrita, Mark D{\'\i}az, Jina Suh, Nick Judd, Lara
  Groves, Hong Shen, Motahhare Eslami, and Kenneth Holstein.
\newblock {Responsible Crowdsourcing for Responsible Generative AI: Engaging
  Crowds in AI Auditing and Evaluation}.
\newblock In \emph{Proceedings of the AAAI Conference on Human Computation and
  Crowdsourcing}, volume~12, pages 148--150, 2024.

\bibitem[Depounti et~al.(2023)Depounti, Saukko, and Natale]{depounti2023ideal}
Iliana Depounti, Paula Saukko, and Simone Natale.
\newblock {Ideal technologies, ideal women: AI and gender imaginaries in
  Redditors’ discussions on the Replika bot girlfriend}.
\newblock \emph{Media, Culture \& Society}, 45\penalty0 (4):\penalty0 720--736,
  2023.

\bibitem[Desiderio et~al.(2025)Desiderio, Mancini, Cimini, and
  Di~Clemente]{desiderio2025highly}
Antonio Desiderio, Anna Mancini, Giulio Cimini, and Riccardo Di~Clemente.
\newblock {Highly engaging events reveal semantic and temporal compression in
  online community discourse}.
\newblock \emph{PNAS Nexus}, 4\penalty0 (3):\penalty0 pgaf056, 2025.

\bibitem[Fang et~al.(2025)Fang, Liu, Danry, Lee, Chan, Pataranutaporn, Maes,
  Phang, Lampe, Ahmad, et~al.]{fang2025ai}
Cathy~Mengying Fang, Auren~R Liu, Valdemar Danry, Eunhae Lee, Samantha~WT Chan,
  Pat Pataranutaporn, Pattie Maes, Jason Phang, Michael Lampe, Lama Ahmad,
  et~al.
\newblock {How AI and Human Behaviors Shape Psychosocial Effects of Extended
  Chatbot Use: A Longitudinal Randomized Controlled Study}.
\newblock \emph{arXiv preprint arXiv:2503.17473}, 2025.

\bibitem[Fedoryszak et~al.(2019)Fedoryszak, Frederick, Rajaram, and
  Zhong]{fedoryszak2019real}
Mateusz Fedoryszak, Brent Frederick, Vijay Rajaram, and Changtao Zhong.
\newblock {Real-time event detection on social data streams}.
\newblock In \emph{Proceedings of the 25th ACM SIGKDD International Conference
  on Knowledge Discovery \& Data Mining}, pages 2774--2782, 2019.

\bibitem[Goh et~al.(2024)Goh, Gallo, Hom, Strong, Weng, Kerman, Cool, Kanjee,
  Parsons, Ahuja, et~al.]{goh2024large}
Ethan Goh, Robert Gallo, Jason Hom, Eric Strong, Yingjie Weng, Hannah Kerman,
  Jos{\'e}phine~A Cool, Zahir Kanjee, Andrew~S Parsons, Neera Ahuja, et~al.
\newblock {Large language model influence on diagnostic reasoning: a randomized
  clinical trial}.
\newblock \emph{JAMA Network Open}, 7\penalty0 (10):\penalty0
  e2440969--e2440969, 2024.

\bibitem[Haddon(2007)]{haddon2007roger}
Leslie Haddon.
\newblock {Roger Silverstone’s legacies: domestication}.
\newblock \emph{New Media \& Society}, 9\penalty0 (1):\penalty0 25--32, 2007.

\bibitem[Hanson and Bolthouse(2024)]{hanson2024replika}
Kenneth~R Hanson and Hannah Bolthouse.
\newblock {“Replika Removing Erotic Role-Play Is Like Grand Theft Auto
  Removing Guns or Cars”: Reddit Discourse on Artificial Intelligence
  Chatbots and Sexual Technologies}.
\newblock \emph{Socius}, 10:\penalty0 23780231241259627, 2024.

\bibitem[Jiang et~al.(2025)Jiang, Sun, Dunlap, Smith, and
  Nanda]{jiang2025interpretable}
Nick Jiang, Xiaoqing Sun, Lisa Dunlap, Lewis Smith, and Neel Nanda.
\newblock {Interpretable Embeddings with Sparse Autoencoders: A Data Analysis
  Toolkit}.
\newblock \emph{arXiv preprint arXiv:2512.10092}, 2025.

\bibitem[Jung et~al.(2025)Jung, Lee, Huang, and Chen]{jung2025ve}
Kyuha Jung, Gyuho Lee, Yuanhui Huang, and Yunan Chen.
\newblock {'I've talked to ChatGPT about my issues last night.': Examining
  Mental Health Conversations with Large Language Models through Reddit
  Analysis}.
\newblock \emph{Proceedings of the ACM on Human-Computer Interaction},
  9\penalty0 (7):\penalty0 1--25, 2025.

\bibitem[Karimiziarani(2022)]{karimiziarani2022tutorial}
Mohammadsepehr Karimiziarani.
\newblock {A tutorial on event detection using social media data analysis:
  Applications, challenges, and open problems}.
\newblock \emph{arXiv preprint arXiv:2207.03997}, 2022.

\bibitem[Khawaja and B{\'e}lisle-Pipon(2023)]{khawaja2023your}
Zoha Khawaja and Jean-Christophe B{\'e}lisle-Pipon.
\newblock {Your robot therapist is not your therapist: understanding the role
  of AI-powered mental health chatbots}.
\newblock \emph{Frontiers in Digital Health}, 5:\penalty0 1278186, 2023.

\bibitem[Kolajo et~al.(2022)Kolajo, Daramola, and Adebiyi]{kolajo2022real}
Taiwo Kolajo, Olawande Daramola, and Ayodele~A Adebiyi.
\newblock {Real-time event detection in social media streams through semantic
  analysis of noisy terms}.
\newblock \emph{Journal of Big Data}, 9\penalty0 (1):\penalty0 90, 2022.

\bibitem[Leskovec et~al.(2009)Leskovec, Backstrom, and
  Kleinberg]{leskovec2009meme}
Jure Leskovec, Lars Backstrom, and Jon Kleinberg.
\newblock {Meme-tracking and the dynamics of the news cycle}.
\newblock In \emph{Proceedings of the 15th ACM SIGKDD International Conference
  on Knowledge Discovery and Data Mining}, pages 497--506, 2009.

\bibitem[Li et~al.(2017)Li, Nourbakhsh, Shah, and Liu]{li2017real}
Quanzhi Li, Armineh Nourbakhsh, Sameena Shah, and Xiaomo Liu.
\newblock {Real-time novel event detection from social media}.
\newblock In \emph{2017 IEEE 33rd International Conference on Data Engineering
  (ICDE)}, pages 1129--1139. IEEE, 2017.

\bibitem[Mathioudakis and Koudas(2010)]{mathioudakis2010twittermonitor}
Michael Mathioudakis and Nick Koudas.
\newblock {TwitterMonitor: Trend detection over the Twitter stream}.
\newblock In \emph{Proceedings of the 2010 ACM SIGMOD International Conference
  on Management of Data}, pages 1155--1158, 2010.

\bibitem[McCreadie et~al.(2013)McCreadie, Macdonald, Ounis, Osborne, and
  Petrovic]{mccreadie2013scalable}
Richard McCreadie, Craig Macdonald, Iadh Ounis, Miles Osborne, and Sasa
  Petrovic.
\newblock {Scalable distributed event detection for Twitter}.
\newblock In \emph{2013 IEEE International Conference on Big Data}, pages
  543--549. IEEE, 2013.

\bibitem[Monroe et~al.(2008)Monroe, Colaresi, and Quinn]{monroe2008fightin}
Burt~L Monroe, Michael~P Colaresi, and Kevin~M Quinn.
\newblock {Fightin'words: Lexical feature selection and evaluation for
  identifying the content of political conflict}.
\newblock \emph{Political Analysis}, 16\penalty0 (4):\penalty0 372--403, 2008.

\bibitem[Moore et~al.(2025)Moore, Grabb, Agnew, Klyman, Chancellor, Ong, and
  Haber]{moore2025expressing}
Jared Moore, Declan Grabb, William Agnew, Kevin Klyman, Stevie Chancellor,
  Desmond~C Ong, and Nick Haber.
\newblock {Expressing stigma and inappropriate responses prevents LLMs from
  safely replacing mental health providers}.
\newblock In \emph{Proceedings of the 2025 ACM Conference on Fairness,
  Accountability, and Transparency}, pages 599--627, 2025.

\bibitem[Moore et~al.(2026)Moore, Mehta, Agnew, Anthis, Louie, Mai, Yin, Cheng,
  Paech, Klyman, et~al.]{moore2026characterizing}
Jared Moore, Ashish Mehta, William Agnew, Jacy~Reese Anthis, Ryan Louie, Yifan
  Mai, Peggy Yin, Myra Cheng, Samuel~J Paech, Kevin Klyman, et~al.
\newblock {Characterizing Delusional Spirals through Human-LLM Chat Logs}.
\newblock In \emph{Proceedings of the 2026 ACM Conference on Fairness,
  Accountability, and Transparency}. ACM, 2026.

\bibitem[Movva et~al.(2025)Movva, Peng, Garg, Kleinberg, and
  Pierson]{movva2025sparse}
Rajiv Movva, Kenny Peng, Nikhil Garg, Jon Kleinberg, and Emma Pierson.
\newblock {Sparse Autoencoders for Hypothesis Generation}.
\newblock In \emph{Proceedings of the 42nd International Conference on Machine
  Learning}, volume 267 of \emph{Proceedings of Machine Learning Research},
  pages 44997--45023. PMLR, 2025.

\bibitem[Newey and West(1987)]{newey1987simple}
Whitney~K Newey and Kenneth~D West.
\newblock {A Simple, Positive Semi-Definite, Heteroskedasticity and
  Autocorrelation Consistent Covariance Matrix}.
\newblock \emph{Econometrica: Journal of the Econometric Society}, pages
  703--708, 1987.

\bibitem[OpenAI(2024)]{openai2024gpt4o}
OpenAI.
\newblock {GPT}-4o system card.
\newblock \url{https://cdn.openai.com/gpt-4o-system-card.pdf}, 2024.

\bibitem[Pataranutaporn et~al.(2025)Pataranutaporn, Karny, Archiwaranguprok,
  Albrecht, Liu, and Maes]{pataranutaporn2025my}
Pat Pataranutaporn, Sheer Karny, Chayapatr Archiwaranguprok, Constanze
  Albrecht, Auren~R Liu, and Pattie Maes.
\newblock {"My Boyfriend is AI": A Computational Analysis of Human-AI
  Companionship in Reddit's AI Community}.
\newblock \emph{arXiv preprint arXiv:2509.11391}, 2025.

\bibitem[Peng et~al.(2025)Peng, Movva, Kleinberg, Pierson, and
  Garg]{peng2025use}
Kenny Peng, Rajiv Movva, Jon Kleinberg, Emma Pierson, and Nikhil Garg.
\newblock {Use Sparse Autoencoders to Discover Unknown Concepts, Not to Act on
  Known Concepts}.
\newblock \emph{arXiv preprint arXiv:2506.23845}, 2025.

\bibitem[{Pew Research Center}(2025)]{PewSocialMediaFactSheet2025}
{Pew Research Center}.
\newblock {Demographics of Social Media Users and Adoption in the United
  States}.
\newblock Social Media Fact Sheet, 2025.
\newblock URL
  \url{https://www.pewresearch.org/internet/fact-sheet/social-media/}.
\newblock Survey conducted Feb. 5--June 18, 2025.

\bibitem[Pham et~al.(2024)Pham, Hoyle, Sun, Resnik, and
  Iyyer]{pham2024topicgpt}
Chau~Minh Pham, Alexander Hoyle, Simeng Sun, Philip Resnik, and Mohit Iyyer.
\newblock {TopicGPT: A prompt-based topic modeling framework}.
\newblock In \emph{Proceedings of the 2024 Conference of the North American
  Chapter of the Association for Computational Linguistics: Human Language
  Technologies (Volume 1: Long Papers)}, pages 2956--2984, 2024.

\bibitem[Proferes et~al.(2021)Proferes, Jones, Gilbert, Fiesler, and
  Zimmer]{proferes2021studying}
Nicholas Proferes, Naiyan Jones, Sarah Gilbert, Casey Fiesler, and Michael
  Zimmer.
\newblock {Studying Reddit: A systematic overview of disciplines, approaches,
  methods, and ethics}.
\newblock \emph{Social Media + Society}, 7\penalty0 (2):\penalty0
  20563051211019004, 2021.

\bibitem[Qiu et~al.(2025)Qiu, Ma, Wu, and Yang]{qiu2025text}
Zitai Qiu, Congbo Ma, Jia Wu, and Jian Yang.
\newblock {Text is All You Need: LLM-enhanced Incremental Social Event
  Detection}.
\newblock In \emph{Proceedings of the 63rd Annual Meeting of the Association
  for Computational Linguistics (Volume 1: Long Papers)}, pages 4666--4680,
  2025.

\bibitem[Qutieshat(2024)]{qutieshat2024unveiling}
Abubaker Qutieshat.
\newblock {Unveiling the multifaceted public interest in ChatGPT: A study on
  societal implications and operational realities}.
\newblock \emph{Journal of Digital Social Research}, 6\penalty0 (3):\penalty0
  112--125, 2024.

\bibitem[Ramdas and Wang(2025)]{ramdas2025hypothesis}
Aaditya Ramdas and Ruodu Wang.
\newblock {Hypothesis testing with e-values}.
\newblock \emph{Foundations and Trends{\textregistered} in Statistics},
  1\penalty0 (1-2):\penalty0 1--390, 2025.

\bibitem[Reuter et~al.(2024)Reuter, Thielmann, Weisser, Fischer, and
  S{\"a}fken]{reuter2024gptopic}
Arik Reuter, Anton Thielmann, Christoph Weisser, Sebastian Fischer, and
  Benjamin S{\"a}fken.
\newblock {GPTopic: Dynamic and interactive topic representations}.
\newblock \emph{arXiv preprint arXiv:2403.03628}, 2024.

\bibitem[Shamma et~al.(2011)Shamma, Kennedy, and Churchill]{shamma2011peaks}
David~A Shamma, Lyndon Kennedy, and Elizabeth~F Churchill.
\newblock {Peaks and persistence: modeling the shape of microblog
  conversations}.
\newblock In \emph{Proceedings of the ACM 2011 Conference on Computer Supported
  Cooperative Work}, pages 355--358, 2011.

\bibitem[Tamkin et~al.(2024)Tamkin, McCain, Handa, Durmus, Lovitt, Rathi,
  Huang, Mountfield, Hong, Ritchie, et~al.]{tamkin2024clio}
Alex Tamkin, Miles McCain, Kunal Handa, Esin Durmus, Liane Lovitt, Ankur Rathi,
  Saffron Huang, Alfred Mountfield, Jerry Hong, Stuart Ritchie, et~al.
\newblock {Clio: Privacy-preserving insights into real-world AI use}.
\newblock \emph{arXiv preprint arXiv:2412.13678}, 2024.

\bibitem[Tunca(2025)]{tunca2025tracing}
Sezai Tunca.
\newblock {Tracing the evolving discourse of sexual technology: A longitudinal
  analysis of emotional, ethical, and technological narratives on Reddit}.
\newblock \emph{Sociology Compass}, 19\penalty0 (8):\penalty0 e70110, 2025.

\bibitem[Vosoughi et~al.(2018)Vosoughi, Roy, and Aral]{vosoughi2018spread}
Soroush Vosoughi, Deb Roy, and Sinan Aral.
\newblock {The spread of true and false news online}.
\newblock \emph{Science}, 359\penalty0 (6380):\penalty0 1146--1151, 2018.

\bibitem[Wu and Huberman(2007)]{wu2007novelty}
Fang Wu and Bernardo~A Huberman.
\newblock {Novelty and collective attention}.
\newblock \emph{Proceedings of the National Academy of Sciences}, 104\penalty0
  (45):\penalty0 17599--17601, 2007.

\bibitem[Xie et~al.(2016)Xie, Zhu, Jiang, Lim, and Wang]{xie2016topicsketch}
Wei Xie, Feida Zhu, Jing Jiang, Ee-Peng Lim, and Ke~Wang.
\newblock {TopicSketch: Real-time bursty topic detection from Twitter}.
\newblock \emph{IEEE Transactions on Knowledge and Data Engineering},
  28\penalty0 (8):\penalty0 2216--2229, 2016.

\bibitem[Xu et~al.(2024)Xu, Fang, Huang, and Xie]{xu2024public}
Zhaoxiang Xu, Qingguo Fang, Yanbo Huang, and Mingjian Xie.
\newblock {The public attitude towards ChatGPT on Reddit: A study based on
  unsupervised learning from sentiment analysis and topic modeling}.
\newblock \emph{PLOS ONE}, 19\penalty0 (5):\penalty0 e0302502, 2024.

\bibitem[Xu and Ramdas(2024)]{xu2024online}
Ziyu Xu and Aaditya Ramdas.
\newblock {Online multiple testing with e-values}.
\newblock In \emph{International Conference on Artificial Intelligence and
  Statistics}, pages 3997--4005. PMLR, 2024.

\bibitem[Yu et~al.(2025)Yu, Huang, Liu, and Tan]{yu2025emotions}
Yifan Yu, Shan Huang, Yuchen Liu, and Yong Tan.
\newblock {Emotions in online content diffusion}.
\newblock \emph{Information Systems Research}, 2025.

\end{thebibliography}
